%% file: arxiv-v1.tex
\documentclass{article}
\usepackage{amsmath,amssymb,dsfont,bm,amsthm}
\usepackage{mathtools}
\usepackage[shortlabels]{enumitem}
\usepackage{authblk}
\usepackage{hyperref}
\usepackage[sort,numbers]{natbib}
\usepackage[margin=1in]{geometry}
\usepackage{xcolor}

\input{macros}

\title{Community detection using  low-dimensional \\network embedding algorithms}
\author{Aman Barot$^\dagger$, Shankar Bhamidi$^\dagger$, Souvik Dhara$^\ddagger$}

\numberwithin{equation}{section}
\newcommand\numberthis{\addtocounter{equation}{1}\tag{\theequation}}
\newcounter{common-numbering} \numberwithin{common-numbering}{section}
\newtheorem{prop}[common-numbering]{Proposition}
\newtheorem{lem}[common-numbering]{Lemma}
\newtheorem{assumption}[common-numbering]{Assumption}
\newtheorem{thm}[common-numbering]{Theorem}
\newtheorem{algo}[common-numbering]{Algorithm}
\newtheorem{rem}[common-numbering]{Remark}
\newtheorem{defn}[common-numbering]{Definition}
\newtheorem{obs}[common-numbering]{Observation}
\newcommand{\bbR}{\mathbb{R}}
\newcommand{\bbN}{\mathbb{N}}

\newcommand{\frob}[1]{\|#1\|_{\sss \mathrm{F}}}

\newcommand{\SOP}{o_{\sss \Prob}}
\newcommand{\diag}[1]{\text{diag}\left(#1\right)}
\usepackage{mathtools}
\DeclarePairedDelimiter\ceil{\lceil}{\rceil}
\DeclarePairedDelimiter\floor{\lfloor}{\rfloor}

\newcommand{\PR}{\mathbb{P}}

\begin{document}
	\maketitle
\begin{abstract}
    With the increasing relevance of large networks in important areas such as the study of contact networks for spread of disease, or social networks for their impact on geopolitics, it has become necessary to study machine learning tools that are scalable to very large networks, often containing millions of nodes. 
    One major class of such scalable algorithms is known as network representation learning or network embedding. 
    These algorithms try to learn representations of network functionals (e.g.~nodes) by first running multiple random walks and then using the number of co-occurrences of each pair of nodes 
    in observed random walk segments to obtain a low-dimensional representation of nodes on some Euclidean space.
    While there has been an exponential growth in the applications of these algorithms in clustering of nodes (community detection), link prediction, and classification of nodes, there is less insight into the correctness and accuracy of these methods. 
    
    The aim of this paper is to rigorously understand the performance of two major algorithms, DeepWalk and node2vec, in recovering communities for canonical network models with ground truth communities. 
    Depending on the sparsity of the graph, we find the length of the random walk segments required such that the corresponding observed co-occurrence window is able to perform almost exact recovery of the underlying community assignments. We prove that, given some fixed co-occurrence window,  node2vec using random walks with a low non-backtracking probability can succeed for much sparser networks compared to DeepWalk using simple random walks. 
    Moreover, if the sparsity parameter is low, we provide evidence that these algorithms might not succeed in almost exact recovery.
    The analysis requires developing general tools for path counting on random networks having an underlying low-rank structure, which are of independent interest.

\end{abstract}
\blfootnote{$^\dagger$ Department of Statistics and Operations Research, University of North Carolina at Chapel Hill}
\blfootnote{$^\ddagger$ Department of Mathematics, Massachusetts Institute of Technology}
\blfootnote{Emails: \href{mailto:abarot@live.unc.edu}{\tt abarot@live.unc.edu}, \href{mailto:bhamidi@email.unc.edu}{\tt bhamidi@email.unc.edu}, \href{mailto:sdhara@mit.edu}{\tt sdhara@mit.edu}}
\blfootnote{Acknowledgements:  AB has been partially supported by NSF award NSF-DMS-1929298 to the Statistical and Applied Mathematical Sciences Institute and by ARO grant W911NF-17-1-0010. SB is supported in part by NSF grants DMS-1613072, DMS-1606839, DMS 2113662 and ARO grant W911NF-17-1-0010. SD is partially supported by Vannevar Bush Faculty Fellowship ONR-N00014-20-1-2826.   }

\section{Introduction}

	Networks provide a useful framework to study interactions between entities, called nodes, in many complex systems such as social networks, protein-protein interactions, and citation networks.
	A number of important machine learning tasks such as clustering of nodes (community detection), link prediction, classification of nodes and visualization of node interactions can be performed using network-based methods for these systems.
	In recent decades, network data sets containing millions and even billions of nodes have become available in areas such as social networks. This has necessitated the development of new methods which are scalable to very large networks.
One major class of algorithms, often termed network representation learning or network embedding techniques, try to learn representations of network functionals including vertices and in some cases edges in a low-dimensional Euclidean space thus making large scale network valued data amenable to  well-known methods for data sets in Euclidean spaces. 
	Typical applications of these methods include network visualization by using t-SNE or PCA on the network embeddings, clustering of related nodes by applying $k$-means on the network embeddings, and classification of nodes by applying machine learning methods for Euclidean spaces. The main advantage of these algorithms lies in the fact that they are scalable to very large networks even with millions of nodes, see the comprehensive references \cite{chami2020machine,zhang2018network,hamilton2017representation} and the references therein for a description of the multitude of methods now available and their applications in various domains of network science.

Network embedding methods can be broadly classified into three types: methods based on matrix factorization \cite{belkin2001laplacian,cao2015grarep,ou2016asymmetric,ahmed2013distributed}, method based on random walks \cite{tang2015line,perozzi2014deepwalk,grover2016node2vec}, and methods based on deep neural networks \cite{cao2016deep,wang2016structural}. In this paper,  we will focus on methods based on random walks. 
The first step for these methods consists of running multiple random walks on the underlying graph and creating a matrix of \emph{co-occurrences} which keeps track of how frequently random walks visit one node from another within certain time-span.
Pairs of nodes which are closer to each other have a higher co-occurrence than pairs which are further apart. 
The second step then consists of optimizing for network embeddings so that the Euclidean inner products of the embeddings are proportional to the co-occurences (cf.~Section~\ref{sec:embedding-algorithms}).
These algorithms turn out to be heavily used in practice owing to a number of reasons including: 

(a) network representations can be constructed using random walk based methods for very large scale networks efficiently (see e.g. \cite{perozzi2014deepwalk,grover2016node2vec} for applications on hundreds of thousands and in some cases million node networks); (b) these algorithms are easily parallelizable; and (c) they can be easily adapted for local changes in the graph and thus easy to use for streaming network data. 
These methods have been empirically observed to perform well for link prediction and node classification for large sparse graphs as co-occurences provide a flexible measure of strength of the relationship between any two nodes as compared to deterministic measures based on node degrees and co-neighbors.

This paper focuses on the problem of finding clustering of nodes, often referred to as the \emph{community detection problem}. 
A well-known model for generating synthetic benchmarks for validating new community detection algorithms is the stochastic block model (SBM)~\cite{fortunato2016community}. 
The literature for theoretical studies for community detection on SBM is extensive and these existing methods primarily use spectral algorithms, semi-definite programming, and message passing approaches. We refer the reader to the survey \cite{A18} for an overview.
Since the more recent network embedding methods are designed to be scalable, one may apply the $k$-means algorithm on the  embeddings from these algorithms to find clusters in the corresponding embeddings in Euclidean space. The primary scientific question we aim to address in this paper is: 
\begin{quote}
	 \emph{In settings when the underlying network is generated from SBM, when can a network-embedding followed by $k$-means clustering recover the community assignments? 
	 Since community detection is generally more difficult for sparse graphs, it is also of interest to know the relationship between sparsity levels and the co-occurrence length which result in meaningful guarantees as well as to understand regimes where these algorithms might fail.}
\end{quote}
Here, co-occurrence length being $t$ means that the algorithm constructs the co-occurance matrix by looking at the co-occurrence of nodes at $t$ steps of the random walk. 
To answer the above questions, we look at perhaps the two most popular of them (1) DeepWalk due to Perozzi, Al-Rfou \& Skiena~\cite{perozzi2014deepwalk},  and (2) node2vec due to Grover \& Leskovec~\cite{grover2016node2vec}. 
While DeepWalk uses simple random walks for the network embedding task, node2vec makes use of a weighted random walk with a different weight for backtracking. 
To understand our results from a high-level, let $n$ denote the network size,  $\rho_n$ denotes the sparsity level and $t$ denotes the co-occurrence length (these are defined more precisely in Sections~\ref{sec:SBM},~\ref{sec:embedding-algorithms}). 
We will establish the following:
\begin{enumerate}
    \item[(R1)] \label{Result-1} There exists a $\phi = \phi(t)$ such that when $n^{t-1} \rho_n^t \times \frac{1}{(n\rho_n)^\phi} \gg (\log n)^C$, then DeepWalk recovers communities of all but $o(\sqrt{n})$ with high probability (cf.~Theorem~\ref{thm:DeepWalk misclassified nodes}). 
    \item[(R2)] If $n^{t-1} \rho_n^t \gg (\log n)^C$ and the backtracking parameter is sufficiently small, then node2vec recovers communities of all but $o(\sqrt{n})$ with high probability (cf.~Theorem~\ref{thm:node2vec misclassified nodes}).
    
    \item[(R3)] If $n^{t-1} \rho_n^t \ll 1$, then the co-occurrence matrix might be quite far from the ground truth, which provides strong evidence that DeepWalk, node2vec might end up misclassifying a positive fraction of nodes (cf.~Theorem~\ref{prop:Frob norm tL to tU}~\eqref{frob-norm-bound-original-2}, and Theorem~\ref{prop:node2vec frob norm}~\eqref{prop:node2vec frob norm-2}).
    
\end{enumerate}
The results show that if the random walks used for the embedding are approximately  non-backtracking, then they can succeed up to the almost optimal sparsity level of the network. Intuitively, backtracks do not provide much new information about the network structure and their effect becomes more prominent as the network becomes sparser. The effect due to backtracks, as well as the interpretation of $\phi$ in (R1), is discussed in more detail in Remark~\ref{rem:effect-of bactracking}.
This phenomenon is in line with the effectiveness of spectral methods using non-backtracking matrices in the extremely sparse setting (networks with bounded average degree), which has been studied by Krzakala,  Moore,  Mossel, Neeman,  Sly, Zdeborov\'a, and Zhang~\cite{Krzakala2013} and Bordenave, Lelarge, and Massouli{\'e}~\cite{Bordenave2015} in the context of recovering communities partially. 
The closest paper to this work is the
important recent result of Zhang and Tang~\cite{zhang2021consistency} which proves an analogue of (R1) above for the DeepWalk. 
Our result (R1) further improves the regime of success for DeepWalk compared to \cite{zhang2021consistency}.

\paragraph*{Organization.}
The rest of the paper is organized as follows. Section~\ref{sec:setup} describes the background such as the stochastic block model, the DeepWalk and node2vec algorithms, and community detection using spectral clustering.
Section~\ref{sec:main-result} describes our main results in details, followed by proof ideas.
The rest of the paper is dedicated to proof of this main results. 
In Section~\ref{sec:path counting deepwalk}, we set up the path counting estimates that are crucially used in our proof. 
We complete analysis for DeepWalk in Section~\ref{sec:deepwalk}, and node2vec in Section~\ref{sec:path counting node2vec} and Section~\ref{sec:node2vec}.
\section{Problem Setup} \label{sec:setup}
We  begin this section by describing the notation used throughout the paper. 
We then describe the stochastic block model (SBM) in Section~\ref{sec:SBM}. The underlying graphs for our results will be generated by this model.
Next, we will describe the DeepWalk and node2vec network embedding algorithms in Section~\ref{sec:embedding-algorithms}. These algorithms will be run on  graphs generated from the SBM.
We will then end this section by describing the approach to recover communities from the solution of the DeepWalk and node2vec algorithms in Section~\ref{sec:spectral-description}.
	
\paragraph*{Notation.}
For a graph $G = (V,E)$, $\AG$ will denote the adjacency matrix representation of $G$. We will drop the subscript $G$ on $\AG$  and denote the adjacency matrix by $A$ when the graph is clear from context.
We write $S^{a\times b}$ for the set of $a\times b$ matrices with entries taking values in $S$. 
For any matrix $M$, $M_{i\star}$ denotes the $i$-th row and $M_{\star i}$ denotes the $i$-th column.
For any matrix or  vector $X$, $|X|$ will denote the sum of its entries. 
Also $\frob{X} = (\sum_{i,j}X_{ij}^2)^{1/2}$ will denote the Frobenius norm of a  $X$.
We denote $[n] = \{1,2,\dots,n\}$.	We often use the Bachmann–Landau notation $o(\cdot), O(\cdot)$, $\omega(\cdot), \Omega(\cdot)$ etc. For random variables $(X_n)_{n\geq 1}$, we write $X_n = \oP(1)$ and $X_n = \OP(1)$ as a shorthand for $X_n \to 0$ in probability, and $(X_n)_{n\geq 1}$ is a tight sequence respectively.

\subsection{Stochastic Block Model} \label{sec:SBM}
We now describe the stochastic block model \cite{holland1983stochastic}. 
A random graph $G = (V,E)$ with $V = [n]$ from the SBM is generated as follows. 
Each node in the graph is assigned to one of $K$ categories, called communities or blocks. 
We fix the set of communities to be $[K]$.
Let $g(i) \in [K]$ denote the community of node $i$. 
We also define a $n \times K$ \emph{membership matrix} $\Theta_0$ which contains the communities of the nodes, defined as follows:
\[	(\Theta_0)_{ir} := \ind\{g(i) = r \}.\]
Let $B$ be a $K \times K$ matrix of probabilities, i.e. $0 \leq B_{rs} \leq 1$. The matrix $B$ will be called the matrix of block or community density parameters.
Let 
\[	P := \Theta_0 B \Theta_0^T.	\]
The edges of the graph $G$ are generated using $P$ as follows: 
	\begin{gather*}
		A_{ij} \sim \mathrm{Bernoulli}(P_{ij}), \	A_{ij} := A_{ji} \quad \text{for }i < j,\quad \text{ and } A_{ii} := 0. \numberthis\label{def:SBM}
	\end{gather*}
The graph generated from this model does not have self-loops. 
Note that the parameter of the Bernoulli distribution used for generating edge $\{i,j\}$  depends only on the communities of the nodes $i$ and $j$.

For our results in this paper, we are interested in the case when the number of communities, $K$, is fixed, and the number of nodes, $n$, tends to infinity. 
The size of community $r$ in the graph will be $n_r$, where $n_r/n \to \pi_r$ with $\pi_r>0$ for all $r\in [K]$.
Then to generate a sequence of graphs, one graph for each $n \in \bbN$, we first fix a $K \times K$ matrix of probabilities $B_0$.
Let $(\rho_n)_{n\geq 1} \subset \bbR$ be a sequence such that $0 \leq \rho_n \leq 1$. The sequence $\rho_n$ will control the sparsity of the graphs as a function of $n$.  
The matrix~$B$ of block density parameters for the graph on $[n]$ is then given by 
\[		B = \rho_n B_0.	\]
We make the following assumption about $B_0$:
\begin{assumption}
	There exists constants $c_L$ and $c_U$ such that $0 < c_L \leq (B_0)_{ij} \leq c_U \leq 1$ and $\rank(B) = k$.
\end{assumption}

In our set up, the membership matrix $\Theta_0$ is the unknown and we want to estimate it by observing one realization of $A$.
Next, we discuss how to evaluate the accuracy of the predicted community assignments.
We will assume throughout that $K$ is known. Let $\hat{\Theta}$ be an estimator of the community assignments. 
Note that the communities are identifiable only upto a permutation of the community labels. Taking this into account, we measure the prediction error by
\begin{align} \label{defn:error-rate}
	\err(\hat{\Theta}, \Theta_0) := \frac{1}{n} \min_{J \in S_K}	\sum_{i \in [n]} \ind \big\{  (\hat{\Theta}J)_{i \star} \neq
	(\Theta_0)_{i \star}  \big\},
\end{align}
where the minimum is over all $K \times K$ permutation matrices.
If $J$ is the minimizer in \eqref{defn:error-rate}, then we call a node to be \emph{misclassified} under $\hat{\Theta}$ when $(\hat{\Theta}J)_{i \star} \neq (\Theta_0)_{i \star}$. 
Thus, we can note that $\err(\hat{\Theta}, \Theta_0)$ just computes the proportion of misclassified nodes.

We conclude our description of the SBM by defining some notation used in the context of random-walks. 
Towards this, for a graph with adjacency matrix $A$, let $\DA$ denote the diagonal matrix with the diagonal entries as the degrees of the nodes $[n]$, i.e., $(\DA)_{ii} =  \sum_{j=1}^{n} A_{ij}$. 
Let $\WA$ denote the one-step transition matrix of a simple symmetric random walk given by $A \DA^{-1}$. Similarly, we define $\WP = P \DP^{-1}$ where $\DP$ is the diagonal matrix containing row sums of $P$.

\subsection{DeepWalk and node2vec} \label{sec:embedding-algorithms}
We describe the two random-walk based network embedding algorithms, namely DeepWalk due to Perozzi, Al-Rfou and Skiena~\cite{perozzi2014deepwalk} and node2vec due to Grover and Leskovec~\cite{grover2016node2vec}. 
	Let $G = (V,E)$ be an undirected, connected (and possibly weighted) graph with $V = [n]$, and let $\AG$ be the adjacency matrix of $G$ if $G$ is unweighted. If $G$ is weighted, we can take $\AG$ to be the matrix of edge weights. 
	Both the algorithms have two key steps. 
	In step 1, the algorithm generates multiple random walks on~$G$. 
	In step 2, one uses a ``word embedding'' algorithm from the natural language processing literature by interpreting these random walks as sequences of words. 
	
	\paragraph*{Step 1: Generating walks.}
	For both the methods, we generate $r$ random walks, each of length $l$ on $G$. In the case of DeepWalk, a simple random walk is performed which has the one-step transition probability given as
	\begin{align*}
		p(v_{i+1} | v_i) &=
		\begin{cases}
		 \frac{1}{|(\AG)_{i\star}|}, &\quad \text{ if }  (v_i, v_{i+1}) \in E,\\ \numberthis \label{eqn:deepwalk transition probs}
		 0 , &\quad \text{ otherwise.}
		\end{cases}
	\end{align*}
	In the case of node2vec, a second-order random walk with two parameters $\alpha$ and $\beta$ is performed. If $\cN(v)$ denotes the neighborhood of $v$, then the one-step transition probability in this case is given by
	\begin{align*}
		p(v_{i+1} | v_i,v_{i-1}) &\propto
		\begin{cases}
		    \alpha  &\quad  \text{ if } v_{i+1} = v_{i-1}, \\
		 1  &\quad \text{ if } v_{i+1} \in  \cN(v_{i-1}) \cap \cN(v_{i}), \\ 
         \beta  &\quad \text{ if } v_{i+1} \in  \cN(v_{i-1})^c \cap \cN(v_{i}), \\
		0 &\quad \text{ otherwise.} \numberthis \label{eqn:node2vec transition probs}
		\end{cases}
	\end{align*} 
	Following Qiu, Dong, Ma, Li,  Wang, and  Tang~\cite{qiu2018network}, we initialize using the stationary distribution of the random walks. For DeepWalk, the initialization is given by
	\begin{align*}
		p(v_i) = \frac{|(\AG)_{v_i\star}|}{|\AG|}, \quad \forall v_i \in V. \numberthis \label{eqn:deepwalk init dist}
	\end{align*}
	And for node2vec, we initialize by
	\begin{align*}
	p(v_1,v_2) = \frac{(\AG)_{v_1 v_2}}{|\AG|}, \quad \forall v_1, v_2 \in V. \numberthis \label{eqn:node2vec init dist}
	\end{align*}
	The initial distribution in  \eqref{eqn:node2vec init dist} places equal probability on each ordered pair of nodes in the edge set~$E$. This initial distribution also happens to be the invariant distribution for the second-order random walk if $\beta = 1$.

\paragraph*{Step 2: Implementing word embedding algorithm.} 
	Next, the random walks are input to word2vec algorithm due to Mikolov, Sutskever, Chen, Corrado, and Dean~\cite{mikolov2013distributed}.
    This generates two node representations for each of the nodes. The latter is described, in the context of graphs, in steps 2a and 2b below.

	\paragraph*{Step 2a: Computing the co-occurence matrix $C$.} 
	In this step, we construct an $n \times n$ matrix $C$ which counts how often two nodes appear in a certain distance of each other with respect to the observed random walks. Formally, let $t_L, t_U \in \bbN$ be such that $t_L \leq t_U$. Let $\blw^{(m)}$ for $1 \leq m \leq r$ be the random walks from Step 1, each represented as a sequence of $l$ nodes. Then
	\begin{align*}
		C_{ij} =& \sum_{m=1}^{r} \sum_{t = t_L}^{t_U}  \sum_{k=1}^{l-t} \big(\ind\{\blw^{(m)}_{k} = i, \blw^{(m)}_{k+t} = j \} + \ind\{\blw^{(m)}_{k} = j, \blw^{(m)}_{k+t} = i \} \big).
	\end{align*}
	DeepWalk and node2vec was proposed to have $t_L =1$. However, we will allow $t_L$ to vary for our theoretical results. Also, set up explained intuitively in the introduction corresponds to the particular case $\tL = \tU = t$.

	\paragraph*{Step 2b: Optimizing for node representations.} 
	The node representations are then computed using the Skip-gram model with negative sampl,ing then (SGNS). 
	Let $d$ be the embedding dimension. The algorithm takes the co-occurence matrix $C$ as an input, and then it outputs two node vectors $\blf_i, \blf_i' \in \bbR^d$ for each node $i \in [n]$. The vector $\blf_i$ is the ``input" representation of the node and the vector $\blf_i'$ is the ``output" representation of the node. We collect these node vectors in two $d \times n$ matrices $F$ and $F'$.
	Although \cite{grover2016node2vec} initially proposed to have $F = F'$, in practice this requirement is often dropped.
	So we will not consider the assumption $F = F'$, which was also done in the theoretical analysis of \cite{zhang2021consistency}. 
	The objective function of the optimization problem is described as follows. Let 
	\begin{align*}
		P_{\sss C} (j) := \frac{\sum_{i=1}^{n} C_{ij} }{\sum_{i=1}^{n} \sum_{j=1}^{n} C_{ij}} = \frac{|C_{\star j}|}{|C|},
	\end{align*}
	be the empirical distribution on $[n]$ constructed using the column sums of $C$. We note that in \cite{mikolov2013distributed}, a distribution proportional to  $P_{\sss C}^{3/4}$ was used instead, which performs better in practice. 
	However, in theoretical analysis \cite{zhang2021consistency,qiu2018network,levy2014neural}, one always considers $	P_{\sss C}$ for simplicity. 
	Then one computes $(F,F')$ by maximizing the objective function 
	\begin{equation} \label{eq:loss1}
		L'_{\sss C}(F,F') = \sum_{i,j=1}^{n}  \bigg( C_{ij} \log \sigma(\langle \blf_i, \blf_{j}'\rangle) + 
		\sum_{ \{	l_m \stackrel{iid}{\sim} P_{\sss C} | 1 \leq m \leq b C(i,j)	\} } \log \sigma(-\langle \blf_i, \blf_{l_m}'\rangle ) \bigg), 
	\end{equation}
	where $b$ samples are taken from $P_{\sss C}$ for each pair of nodes counted as a co-occurrence in $C$, and $\sigma (x) = (1+\e^{-x})^{-1}$ denotes the sigmoid function. 
	For our theoretical analysis, we again follow
	simplifications considered in earlier works \cite{levy2014neural,qiu2018network,zhang2021consistency}. First, instead of \eqref{eq:loss1}, we look at a related maximization problem 
	\begin{align} \label{eq:loss2}
		L_{\sss C}(F,F') = \sum_{i,j=1}^{n} 
		C_{ij} \big( \log \sigma(\langle \blf_i, \blf_{j}'\rangle) + b \E_{l\sim P_{\sss C}} [ \log \sigma(-\langle \blf_i, \blf_{l}'\rangle ) ] \big). 
	\end{align} 
Once we have an optimizer $(F,F')$, we embed node $i$ in $\R^d$ using $F_{i\star}$. 
The following result  due to Levy and Goldberg~\cite{levy2014neural} gives us a way to compute the optimizer in \eqref{eq:loss2} using factorization of a certain matrix:
\begin{prop} \label{lem:factorization}
Given a matrix $C\in \R^{n\times n}$, let $\bar{M}_{\sss C}$ be a matrix with $(i,j)$-th entry given by 
\begin{eq}
(\bar{M}_{\sss C})_{ij}:= \log \bigg(\frac{C_{ij} \cdot  |C|}{ |C_{i \star}| |C_{\star j}|}  \bigg) - \log b.
\end{eq}
Let $F, F' \in \R^{n\times d}$
be such that  $\bar{M}_{\sss C} = F F'^T$. Then $(F,F')$ minimizes \eqref{eq:loss2}.
\end{prop}
We give a proof in Appendix~\ref{appendix-1}.
To simplify the form of the optimizers, one either takes $r \to \infty$ or $l \to \infty$ \cite{qiu2018network,zhang2021consistency}. 
These assumptions ensure that the co-occurence of $(i,j)$ is observed sufficiently many times for all $(i,j)$ and also simplify the form of $\bar{M}_{\sss C}$. 
For our results we will take $r \to \infty$ in which case we have the following result:
\begin{prop}\label{lem:M for node2vec}
	Let $A$ be the adjacency matrix of a graph. For node2vec, we will assume $A$ to be unweighted.
	Let $\blw^{(m)}$ for $1 \leq m \leq r$ be the random walks generated for DeepWalk or for node2vec with the parameter $\beta$ set to be equal to $1$.
	If $\sum_{t=t_L}^{t_U} A^{(t)}_{ij} > 0$ then 
	\begin{align*}
		(\bar{M}_{\sss C})_{ij} \xrightarrow[r \to \infty]{a.s.}  (M_{\sss C})_{ij} := \log\bigg(	\frac{\sum_{t = t_L}^{t_U}	(l-t) \big(  \Prob( \blw^{(1)}_1 =i ,  \blw^{(1)}_{1+t} = j )
    	+ \Prob( \blw^{(1)}_1 = j , \blw^{(1)}_{1+t} = i )	\big)	}{2 b \gamma(l,t_L, t_U) \Prob(\blw^{(1)}_1 = i) \Prob(\blw^{(1)}_1 = j)}	\bigg),
	\end{align*}
	where $\gamma(l,t_L, t_U) = \frac{(2l - t_L - t_U)(t_U - t_L + 1)}{2}$.
	The limiting term is well-defined if  $\sum_{t=t_L}^{t_U} A^{(t)}_{ij} > 0$ and the left hand side is also well-defined for large enough $r$. 
\end{prop}
We give a proof in  Appendix~\ref{appendix-1}.
If $\sum_{t=t_L}^{t_U} A^{(t)}_{ij} = 0$ then $(M_{\sss C})_{ij}:= 0$.
We will refer to $M_{\sss C}$ as the \emph{M-matrix} associated to DeepWalk or node2vec on a  graph.

\subsection{Spectral recovery using \texorpdfstring{$M$}{M}-matrix}\label{sec:spectral-description}
We simply write $M$ to denote the $M$-matrix of SBM. 
The idea would be to apply spectral clustering to the $M$ matrix in order to recover the communities, which considers the top $K$ eigenvalues of $M$ and applies an approximate $k$-means algorithm to recover the communities. 
Let us describe this more precisely below. 
Given a graph, we can always factorize $M = FF'^T$ such that $F,F'\in \R^{n\times d}$ for $d\geq n$. That would result in an $n$-dimensional embedding of the nodes.  
However, since the underlying graph  has an approximate rank-$K$ structure, it might make more sense to try to find an embedding in $\R^K$ using an approximate factorization of $M$. To that end, we can find 
\begin{align*}
    \argmin_{F,F' \in \R^{n\times K} } \frob{M - F F'^T}. 
\end{align*}
The solution can be obtained using the singular value decomposition of $M$. Indeed, if $V$ (resp.~$U$) is the matrix of top $K$ left (resp.~right) singular vectors, and $S$ is the diagonal matrix with $K$ largest singular values, then $F = V$ and $F' = S U$. 
To understand if $F$ should preserve the community structure, let's look at the graph $G_0$ which is a weighted complete graph with weights given by $P$.
Then the corresponding $M$-matrix, denoted as $M_0$, is deterministic. 
This can happen as the rank of a matrix can drop after taking element-wise logarithm. 
To see this, consider the following simple counter-example:
\begin{align*}
A = \begin{bmatrix}
1 & 1\\
1 & e
\end{bmatrix},
\text{ and }
\log A = \begin{bmatrix}
0 & 0\\
0 & 1
\end{bmatrix}.
\end{align*}
We provide the following lemma which says that such cases only happen on a set of matrices of measure zero. 
This implies that outside a set of measure zero, if $B$ rank $K$ then $M_0$ also has rank $K$. 
\begin{lem}\label{lem:rank M is K}
Let $X \in \mathbb{R}_{+}^{K\times K}$. Let $\bar{\log}: \mathbb{R}_{+}^{K\times K} \to \mathbb{R}^{K\times K}$ be the mapping given by taking the element-wise $\log$ of the matrix entries, and $\bar{\log} (x) =0$ if $x = 0$. 
Then 
\begin{align*}
\rank(X) = K \implies \rank(\bar{\log}(X)) = K \quad a.s.\, \lambda,
\end{align*} where $\lambda$ is the Lebesgue measure on $\mathbb{R}_{+}^{K\times K}$. 
\end{lem} 
See Appendix~\ref{appendix-1} for a proof. For this reason, we assume the following throughout:
\begin{assumption}
    $\rank(M_0) = K$. 
\end{assumption}

Recall that $\Theta_0$ is the matrix of community assignments. 
Also, let $V_0\in \R^{n\times K}$ be the matrix of the top $K$ left singular vectors of $M_0$. 
We will prove the following: 
\begin{prop}\label{prop:eigenvector M_0}
	There exists $X_0\in \R^{K\times K}$ such that $V_0 = \Theta_0 X_0 + E_0$ where $(E_0)_{ij} = O(n^{-3})$. 
\end{prop}
The proof of this fact in provided in Section~\ref{sec:M-0}  
for DeepWalk and in Section~\ref{sec:M matrix node2vec} for node2vec.
Motivated by this, we compute the community assignments as a $k$-means algorithm and a ($1+\varepsilon$)-approximate solution to the same. 
The $K$-means algorithm on the rows of $V$ solves the following optimization problem.
	\begin{align}
		(\hat{\Theta}, \hat{X}) =  \argmin_{\Theta \in \{0,1\}^{n\times K}, X \in \bbR^{K \times K}} \frob{\Theta X - V }^2.  \label{eqn:kmeans1}
	\end{align}
Since it is NP-hard 
to find the minimizer for the above problem \cite{aloise2009np}, we can seek an approximate solution instead \cite{kumar2004simple}. Let $\varepsilon > 0$ be given. 
Then we say $(\bar{\Theta}, \bar{X})$ is an $(1 + \varepsilon)$-approximate solution to the $K$-means problem in  \eqref{eqn:kmeans1} if $\bar{\Theta} \in \{0,1\}^{n\times K}, \bar{X} \in \bbR^{K \times K}$ and 

\begin{align}
	\frob{ \bar{\Theta} \bar{X} - V }^2 \leq (1+\varepsilon)  \min_{\Theta \in \{0,1\}^{n\times K}, X \in \bbR^{K \times K}} \frob{\Theta X - V }^2. \label{eqn:kmeans2}
\end{align}
Thus the final community assignment is computed as follows: 
\begin{algo}[Low-dimensional embedding using DeepWalk and node2vec] \label{algo-k-means} \leavevmode
\normalfont
\begin{enumerate}
    \item[(S1)] Compute the $M$-matrix $M$ and compute $V \in \R^{n\times K}$, the matrix containing the top $K$-eigenvectors (in absolute value of the eigenvalue) of $M$. 
    \item[(S2)] Compute $(\bar{\Theta}, \bar{X}) \in \{0,1\}^{n \times K} \times \R^{K \times K}$ as an ($1+\varepsilon$)-approximate solution to the $k$-means problem in~\eqref{eqn:kmeans1} and output $\bar{\Theta}$.   
\end{enumerate}
\end{algo}

The $\{0,1\}$-valued matrix $\bar{\Theta}$ has the predicted community assignments for each of the nodes. 
We note that each node is assigned exactly one community by the algorithm

\section{Main results}\label{sec:main-result}
Following the intuitions from Section~\ref{sec:spectral-description}, the primary objective in our community detection task will be to bound $\frob{M - M_0}$. 
We describe our results in two parts: we first describe the results for the DeepWalk in Section~\ref{sec:deep-walk-results}, and then state the results for the node2vec in Section~\ref{sec:node2vec-results}. We end this section with a proof outline.

\subsection{Results for DeepWalk}\label{sec:deep-walk-results}
We describe the following proposition for bounding the Frobenius norm $\frob{M - M_0}$. 
We will assume that $\tL \geq 2$, since if $\tL =1$, then  there are $\OP(n^2)$ entries of $M$ which are equal to $0$. This makes $\frob{M - M_0}$ of the order $n$. 
The following result gives estimates on $\frob{M - M_0}$ based on $\tL,\tU$ and the sparsity parameter $\rho_n$: 

\begin{thm}\label{prop:Frob norm tL to tU}
Fix $\eta > 0$. Suppose $t_L \geq 2$, and let $\phi = \phi(\tL):= \lfloor \tL/2\rfloor $ if $\tL \geq 3$
and let $\phi = 0$ if $t_L = 2$. Also let $c_0 = c_0 (\tL,\eta) = 4 + (\tL+1)\eta$,
and suppose that 
\begin{align}\label{eq:t-condn}
    n^{\tL - 1} \rho_n^{\tL} \times \frac{1}{(n\rho_n)^\phi} \gg (\log n)^{c_0}. 
\end{align}
	Then
	\begin{align}\label{frob-norm-bound-original-1}
	    \frob{M - M_0} = \OP\bigg(n (\log n)^{-\eta} + \ind\{t_L=2\} \bigg(n\sqrt{\frac{\log n}{n \rho_n^2}}\bigg)  \bigg). 
	\end{align}
	If $\rho_n$ is such that $n^{t_U-1}\rho_n^{t_U} \ll 1$ and $n \rho_n \gg 1$, then given $\varepsilon > 0$ there exists a constant $C_{\varepsilon} > 0$ such that
	\begin{align} \label{frob-norm-bound-original-2}
		\Prob \left( \frob{M - M_0} \geq C_{\varepsilon} n \right) \geq 1 - \varepsilon.
	\end{align}
\end{thm}
The proposition says that $\frob{M - M_0} = \oP(n)$ as long as $n^{\tL - 1 - \phi} \rho_n^{\tL-\phi} $ is growing faster than an appropriate power of $\log n$. 
The power of the $\log n$ term is required for  concentration of the transition matrix of the random walk on SBM. 
The last part of the proposition says that the Frobenius norm is $\Omega(n^2)$ when $n^{\tU-1}\rho_n^{\tU} \ll 1$. This suggests that spectral clustering on the $M$-matrix may not give good results at this level of sparsity. Thus in order to ensure a good low-dimensional embedding, we should take $\tL$ suitably large.

\begin{rem}[Effect of backtracking] \normalfont \label{rem:effect-of bactracking}
To understand the quantity $\phi$ intuitively, let us consider a simple case with $\tL = \tU = t\geq 4$ and $t$ is even. If we were to compute the entry say $M_{ii}$, we need to understand how many walks of length $t$ are possible from $i$ to $i$. If we consider the edges in this walk to be distinct, then the expected number of paths is of order $n^{t-1}\rho_n^t$, as we need to choose $t-1$ intermediate vertices and $t$ specific edges need to appear. On the other hand, if we consider the path where the walk alternatively visits a new node and then backtracks to $i$, then the expected number of paths due to such backtracks is $(n\rho_n)^{t/2} = (n\rho_n)^\phi$. The condition in \eqref{eq:t-condn-2} is then just ensuring that the contribution due to the backtracking path is smaller than the contribution from the non-backtracking path. These kinds of backtracking walks do not contribute significantly for node2vec with $\alpha$ small.
\end{rem}

With the bounds on the Frobenius norms in Theorem~\ref{prop:Frob norm tL to tU}, 
we will conclude the following:

\begin{thm}\label{thm:DeepWalk misclassified nodes}
Suppose that \eqref{eq:t-condn} holds. 
Fix $\varepsilon>0$ and let $(\bar{\Theta}, \bar{X})$ be a ($1+\varepsilon$)-approximate solution in Algorithm~\ref{algo-k-means}. Then
\begin{align*}
    \err(\bar{\Theta}, \true) = \SOP(n^{-1/2}), 
\end{align*}i.e., $\bar{\Theta}$ misclassifies at most $\SOP(n^{1/2})$ many node-labels. 
\end{thm}

\subsection{Results for node2vec}\label{sec:node2vec-results}
We now describe our results for the node2vec algorithm. 
As mentioned earlier, for our theoretical analysis, we allow the parameter $\alpha = \alpha_n$ to vary with $n$ and the parameter $\beta$ to be fixed to be equal to $1$. We will also consider the cases when $t_L > 2$. This is because when $t_L = t_U =2$, $M_0$ may not have a block structure, even asymptotically, and so spectral clustering of $M$ may not give us the communities of the nodes (cf.~Lemma~\ref{lem:order of M(alpha)})

We look at the case where $\alpha = \alpha_n \to 0$. In particular, we consider $\alpha = O(1/n)$.
The following proposition bounds the Frobenius norm $\frob{M - M_0}$ for node2vec. 

\begin{thm}\label{prop:node2vec frob norm}
	Fix $\eta > 0$. Suppose $t_L \geq 3$, and let $\phi = \phi(\tL):= \lfloor \tL/2\rfloor $. Also let $c_0 = c_0 (\tL,\eta) = 4+ (\tL+2)\eta$, $\alpha = O\left(\frac{1}{n}\right)$, and suppose that
    	\begin{align}
    	n^{\tL - 1} \rho_n^{\tL} \gg (\log n)^{c_0}.\label{eq:t-condn-2}
    	\end{align}
\noindent Then
	\begin{align}\label{prop:node2vec frob norm-1}
		\frob{M - M_0} = \OP\left(n (\log n)^{-\eta}	\right).
	\end{align} 
	On the other hand, If $\rho_n$ is such that $n^{\tU-1}\rho_n^{\tU} \ll 1$ and $n\rho_n \gg 1$ then given $\varepsilon > 0$ there exists a constant $C_{\varepsilon} > 0$ such that
	\begin{align}\label{prop:node2vec frob norm-2}
		\Prob \big( \frob{M - M_0} \geq C_{\varepsilon} n \big) \geq 1 - \varepsilon.
	\end{align}
\end{thm}
This shows that the communities can be recovered for sparser graphs compared to the DeepWalk case. In particular, if we take $t_L =  t_U = t \geq 3$, then the result says that we can recover the community assignments even when close to the regime when $n^{t-1}\rho_n^t \gg 1$, while $\frob{M - M_0} = \Omega(n)$ for $n^{t-1}\rho_n^t \ll 1$ .  
The intuitive reason for such good performance of node2vec is that, if $\alpha$ is small, then situations such as Remark~\ref{rem:effect-of bactracking} does not arise for the biased random-walks in node2vec.
	
The bounds on the Frobenius norm in Theorem~\ref{prop:node2vec frob norm} leads to the following theorem about the proportion of misclassified nodes.

\begin{thm}\label{thm:node2vec misclassified nodes}
	Fix $\varepsilon >0$ and let $t_L \geq 3$. 
	Suppose $\rho_n$ satisfy the respective conditions for the two regimes as in Theorem~\ref{prop:node2vec frob norm}, let
	be a ($1+\varepsilon$)-approximate solution in Algorithm~\ref{algo-k-means}. Then
	\begin{align*}
    \err(\bar{\Theta}, \true) = \oP(n^{-1/2}), 
    \end{align*}i.e., $\bar{\Theta}$ misclassifies at most $\oP(n^{1/2})$ many node-labels. 
\end{thm}	
\paragraph*{Proof outline.} 
Before going into the proofs of these results, let us give a brief outline.  
We will mainly provide a proof for the $\tL = \tU = t$ case and the general case can be reduced to this special case.
The main idea for proving the upper bounds in Theorems~\ref{prop:Frob norm tL to tU},~\ref{prop:node2vec frob norm} is to get good estimates on the moments of the total number of paths of length $t$.
We will count these paths for each of the possible community assignments of the intermediate vertices in the path. 
The goal is to show that the main contribution on the $k$-moments come from $k$ disjoint paths. Proving that the contribution due to rest of the paths is small requires a novel combinatorial analysis due to the possibilities of backtracks. We first develop these methodology for DeepWalk (cf.~Section~\ref{sec:path counting deepwalk}).
The path counting estimates allow us to bound $\frob{M-M_0}$ using Proposition~\ref{lem:M for node2vec}. To prove Theorem~\ref{thm:DeepWalk misclassified nodes}, we use perturbation analysis for the eigenspaces of $\frob{M-M_0}$. Due to Proposition~\ref{prop:eigenvector M_0}, the eigenvectors of $M_0$ are such that two of its rows are the same if the nodes are in the same community and two rows are different if they are in different community. Applying the perturbation analysis, the same property remains approximately true for the top eigenvectors of $M$ as well, which allows us to prove the success of Algorithm~\ref{algo-k-means} (cf.~Section~\ref{sec:deepwalk}). The proof for node2vec uses similar ideas though the different weights for the backtracking parameter of the random walk requires careful analysis (cf.~Sections~\ref{sec:path counting node2vec},~\ref{sec:node2vec}).

\section{Path counting for DeepWalk}\label{sec:path counting deepwalk}
In this section, we focus on computing the asymptotics for the number of paths having a some specified community assignments for the intermediate vertices. In Section~\ref{sec:counting-type-b-paths}, we bound its $k$-th moment and we end with a concentration cound in Section~\ref{sec:concent-deepwalk}.

\subsection{Bounding moments for paths of different type}\label{sec:counting-type-b-paths}
Let us first set up some notation. 
Recall that $g(u)$ denotes the community assignment for vertex~$u$.
We say a path $(i_0, i_1, i_2, \ldots, i_t)$ has composition $(b_0, b_1, \ldots, b_t)$ if $g(i_l) = b_l$. Define the collection of path compositions for paths between two nodes $i$ and $j$ as
		\begin{align}
			\cB_{ij} := \left\{ (b_0, b_1, \ldots, b_t): b_l \in [K] \text{ for } 0 \leq l \leq t,  b_0= g(i), b_t = g(j)	\right\}. \label{eqn:path types set def}
		\end{align}
For $b = (b_0, b_1, \ldots, b_t) \in \cB_{i,j}$, we define the collection of paths  between $i$ and $j$ with composition~$b$
in the complete graph $K_n$ as \begin{align}
	\cP_b := \left\{	(i_0, i_1, \ldots, i_t) : i_l \in [n] \text{ for } 0 \leq l \leq t,  i_0 = i, i_t = j, g(i_l) = b_l	\right\}.
			\label{eqn:P_b def}
\end{align}
	For any path $p = (i_0, i_1,  \ldots, i_t)  \in \cP_b$, we associate the random variable
		\begin{align}
			X_p := A_{i_0 i_1} A_{i_1 i_2} \cdots A_{i_{t-1} i_t},
			\label{eqn:X-p-def}
		\end{align}
		and  define
		\begin{align} \label{eq:Y-b-def}
			Y_b := \sum_{p \in \cP_b} X_p.
		\end{align}
		To each element $b = (b_0, b_1, \ldots, b_t)  \in \cB_{i,j}$
		we associate the term
		\begin{align}
			U_{(b_0, b_1, \ldots, b_t)} = U_b := \left(\prod_{i=1}^{t-1} n_{b_{i}} B_{b_{i-1} b_i} \right) B_{b_{t-1} b_t}.\label{eqn:U-b-def}
		\end{align}
		We will upper bound the moments of $Y_b$ in terms of $U_b$.
		Similarly 
        for the lower bound, we define
		\begin{align}
			L_{(b_0, b_1, \ldots, b_t)} = L_b := \left(\prod_{i=1}^{t-1} (n_{b_{i}} - (k(t-1) + 1)) B_{b_{i-1} b_i} \right) B_{b_{t-1} b_t}
			\label{eqn:L-b-def}
		\end{align}
		With this setup, we can state the following bounds on $\E Y_b^k$.
		\begin{prop}\label{prop:A^t kth moment}
			Let $\tL = \tU = t \geq 3$ be given and suppose that \eqref{eq:t-condn} holds.
			Then we have 
			\begin{align*}
				L_{b}^k \leq \E Y_b^k \leq U_{b}^k (1 +  o(1) ).
			\end{align*}
		\end{prop}
		The idea is to show that the leading term for $\E Y_b^k$ is due to $\E\big(	\prod_{\alpha=1}^{k}	X_{p_\alpha}	\big)$ of $k$ ordered paths having $kt$ distinct edges between them. 
		The contribution of the rest of the terms are of a smaller order.
		We summarize the second claim below. For any path $p = (i_0, i_1, i_2, \ldots, i_t) \in \cP_b$, let
		\begin{align*}
			e(p):= \left\{	\{i_l, i_{l+1}\} : 0 \leq l < t  \right\},
		\end{align*}
		be the set of edges in the path $p$. Let
		\begin{equation}\label{E-m-defn}
		    E_m := \sum_{(p_1, p_2, \ldots, p_k): p_i \in \cP_b, |\cup_{\alpha\in [k]} e(p_\alpha)| = m } \Prob (X_{p_1} X_{p_2} \cdots X_{p_k} = 1).
		\end{equation}
		We will show the following: 
		\begin{prop}\label{prop:A^t kth moment-2}
		Under identical conditions as Proposition~\ref{prop:A^t kth moment}, we have $\sum_{m<kt} E_m = o(U_b^k)$.  
		\end{prop}
\begin{proof}[Proof of Proposition~\ref{prop:A^t kth moment} using Proposition~\ref{prop:A^t kth moment-2}]
Note that, we can write 			\begin{align*}
	\E Y_b^k &=  \E \bigg(\sum_{p \in \cP_b} X_p \bigg)^k= \sum_{(p_1, p_2, \ldots, p_k): p_{\alpha} \in \cP_b} \Prob (X_{p_1} X_{p_2} \cdots X_{p_k} = 1). \numberthis\label{eqn:kth moment and k paths}
	\end{align*}
	For the upper bound, Proposition~\ref{prop:A^t kth moment-2} shows that it is enough to bound the summands  corresponding to sequences $(p_1, p_2, \ldots, p_k)$ that satisfy $\vert \cup_{\alpha =1}^k e(p_{\alpha})	\vert = kt$, i.e. sequences of paths consisting of $kt$ distinct edges. In this case, $ \Prob (X_{p_1} X_{p_2} \cdots X_{p_k} = 1) = \prod_{\alpha=1}^{k}  \Prob (X_{p_\alpha} = 1)$. 
	We note that each of the paths $p_r$ has path type $b$, and we bound
	\begin{align*}
		\Prob (X_{p_\alpha} = 1) \leq U_{b} = \bigg(\prod_{i=1}^{t-1} n_{g_{i}} B_{g_{i-1} g_i} \bigg) B_{g_{t-1} g_t},
	\end{align*}
	and thus the upper bound follows using Proposition~\ref{prop:A^t kth moment-2}.
	For the lower bound, we can simply restrict the summands in \eqref{eqn:kth moment and k paths} to the case $\vert \cup_{\alpha =1}^k e(p_{\alpha})	\vert = kt$. 
	We compute
	\begin{align*}
		\Prob (X_{p_\alpha} = 1) &\geq  \bigg(\prod_{i=1}^{t-1} (n_{b_{i}} - (2 + (\alpha-1)(t-1) + (i-1))) B_{b_{i-1} b_i} \bigg) B_{b_{t-1} b_t}\\
		&\geq  \bigg(\prod_{i=1}^{t-1} (n_{b_{i}} - (k(t-1) + 1)) B_{b_{i-1} b_i} \bigg) B_{b_{t-1} b_t}= L(b).
	\end{align*}
	Above for the marked vertices in path $p_\alpha$, the term $(\alpha-1)(t-1)$ is to account for not choosing vertices seen in the first $\alpha-1$ paths, the summand $2$ is for the nodes $i$ and $j$, and $(i-1)$ is for the nodes upto index $(i-1)$ in path $p_\alpha$. Hence the proof of the lower bound is also complete.

\end{proof}

The rest of this section is devoted to the proof of Proposition~\ref{prop:A^t kth moment-2}. Let us start by setting up some definitions that will be useful for us to count the contributions coming from intersecting paths. All these definitions are demonstrated in Figure~\ref{fig:path-defns}.

\begin{defn}[Marked edge and marked vertex] \normalfont 
Let $(p_1, p_2, \ldots, p_k)$ be an ordered sequence of $k$ paths in $\cP_b$. 
Fix one of the paths $p_\alpha = (i_0, i_1, i_2, \ldots, i_t)$ and consider the directed edge $(i_l, i_{l+1})$ appearing at the $l$-th step. 
We will call $(i_l, i_{l+1})$ to be a \emph{marked edge} if the undirected edge $\{i_l, i_{l+1}\}$ is not present in the paths $p_{\alpha'}, 1 \leq \alpha' < \alpha$ and also it is not equal to previous edges in the path $p_\alpha$, i.e., $\{i_{l}, i_{l+1}\} \neq \{i_{l'}, i_{l'+1}\}$ for $0 \leq l' < l$. 
Intuitively, we call $(i_l, i_{l+1})$ a marked edge if it is the first time we see it as we count the edges along the paths $p_1$ to $p_k$.
For a marked edge  $(i_l, i_{l+1})$, we will call $i_{l+1}$ to be a \emph{marked vertex} if it was not seen before in previous paths and also in $(i_0,\dots,i_l)$. 
\end{defn}
\begin{defn}[Backtrack] \normalfont 
A directed edge $(i_l, i_{l+1})$ in a path $p_\alpha = (i_0, i_1, i_2, \ldots, i_t)$ is called a \emph{backtrack} if $i_{l+1} = i_{l-1}$.
\end{defn}

\begin{defn}[Segment] \normalfont \label{defn:segment}
Let $0 \leq l < l' \leq t$. We will say that $(i_l, i_{l'})$ is a \emph{segment} in path $p_\alpha$ if the following conditions hold:
	\begin{enumerate}
		\item $(i_l, i_{l+1})$ is a marked edge, i.e., segments always start with a marked edge. 
		\item $(i_j, i_{j+1})$ is a marked edge or a backtrack for all $l \leq j < l'$.
		\item There does not exist $0 \leq l'' < l$ such that $(i_{j}, i_{j+1})$ is a marked edge or a backtrack for all $l{''} \leq j < l$ and $(i_{l''}, i_{l'' + 1})$ is a marked edge.
		\item Either $l' = t$ or if $l' < t$ then $(i_{l'}, i_{l' + 1})$ is neither a marked edge nor a backtrack.
	\end{enumerate}
Intuitively, segments are maximal parts of paths consisting  of marked edges and their backtracks. 
The last two conditions ensure that segments cannot be extended to the left and to the right. 
The edges outside the segments will often be referred to as \emph{unmarked} edges. 
Thus, an unmarked edge is an edge that was previously visited and it is not a backtrack of the last visited marked edge. In Figure~\ref{fig:path-defns}, the dotted lines are unmarked edges, and we can note that the corresponding undirected edges had appeared previously in the path.  
Notice also that any two segments are separated by one or more unmarked edges.
Finally, we remark that the edge preceding a segment may be a backtrack but it can only be a backtrack of an unmarked edge; see for example the second segment in Figure~\ref{fig:path-defns}.
\end{defn}

\begin{figure}
    \centering
    \includegraphics[scale = 0.05]{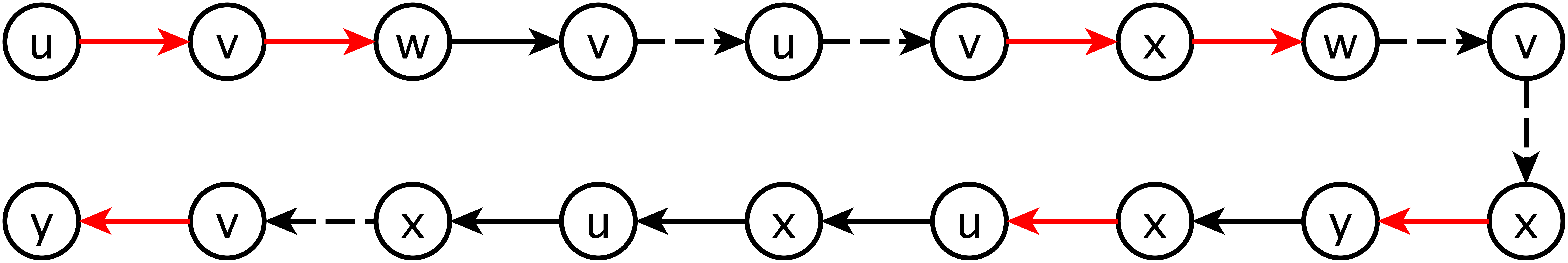}
    \caption{Illustrating marked edges (red), backtracks (black), and unmarked edges (dotted). The segments in this path are given by  $S_1 = (u, v, w, v)$, $S_2 = (v, x, w)$, $S_3 = (x,y,x,u,x,u,x)$, $S_4 = (v,y)$. Here $S_3$ is a Type II segment with $k_2 - k_2 = 2$ and $k_3 - k_2 = 4$. The rest are Type I segments. }
    \label{fig:path-defns}
\end{figure}

\begin{defn}[Type I/II Segments and Paths] \normalfont 
Let $(i_l,i_{l'})$ be a segment with $r$ marked edges. Suppose $l = k_1 < k_2 < \cdots k_{r+1} = l'$ be integers such that $(i_{k_{r'}}, i_{k_{r'} + 1})$ for $1 \leq {r'} \leq r$ constitute the set of all marked edge. 
Then $(i_l,i_{l'})$ is said to be a \emph{Type II segment} if $k_{r'+1} - k_{r'}$ is an even number for $1 \leq r' \leq r$. 
In all other cases, $(i_l,i_{l'})$ is said to be a \emph{Type I segment}. 
Intuitively, a Type II
segment just represents going back and forth on the same vertex. 
In Figure~\ref{fig:path-defns}, the third segment is Type II and the rest of the segments are Type I. 
We say that a path $p$ is a \emph{Type I path} if there is at least one Type I segment in it. Otherwise, we call it a \emph{Type II path}. Notice that Type II path may not have any segment.
\end{defn}

\begin{defn}\normalfont  
We call a path $p$ \emph{saturated} if all the edges in $p$ are part of some segment, i.e., there are no unmarked edges in a saturated path. 
\end{defn}
We now state the following elementary lemma which will be used throughout: 
\begin{lem} \label{lem:choice-vertices}
There are at most $n^{r-1}$ ways to choose marked vertices for a Type I path with $r$ marked edges, and there are at most $n^r$ ways to do the same for a Type II path with $r$ marked edges.
\end{lem}
\begin{proof} 
Let us focus only on a path $p$ that is  saturated, since we want to maximize the choice of marked vertices and they can only be found by walking through a marked edge.
Consider the path $p$, and  note that the endpoints of $p$ are fixed at $i,j$, and all the edges are either marked or is a backtrack.  
Let $(i_{k_\alpha},i_{k_\alpha +1})$ with $1\leq \alpha \leq r$ be the marked edges. 
We construct a graph $H_p$ using only the marked edges (ignoring their orientation).  
Since $p$ is saturated, we have $H_p$ is connected, and also the vertices of $H_p$ (except possibly $i,j$) are marked vertices. If $p$ is Type II, then $H_p$ is a star centered at $i$ having $r$ edges, and $j$ can either be $i$ or be one of the leaves. 
Thus the vertices except $i,j$ can be chosen in at most $n^r$ ways. 
If $p$ is Type I, in order to get the maximum number of vertices in $H_p$, one can have $H_p$ to be a tree if $i\neq j$ or a unicyclic graph if $i = j$. In both cases, there are at most $r-1$ vertices (other than $i,j$) to choose, and these can be chosen in $n^{r-1}$ ways. 
\end{proof}
To complete the proof of Proposition~\ref{prop:A^t kth moment-2}, let $E_{m,r}$ denote the summands in \eqref{E-m-defn} restricted to the case that there are $r$ Type I segments, so that
\begin{align*}
    E_m = \sum_{r= \rmin}^{\rmax} E_{m,r}, 
\end{align*}
where, given $m$, $[\rmin,\rmax]$ denotes the possible range of $r$. The analysis will consist of two steps. In the first step, we analyze $E_{m,\rmin}$. Subsequently, we will show that $E_{m,r}$ is much smaller than $E_{m,\rmin}$ for $r>\rmin$. 

Intuitively, $E_{m,\rmin}$ is the largest term as minimizing the number of Type I segments leads to maximizing the number of choices for marked vertices. This can be seen from Lemma \ref{lem:choice-vertices} which says that there are at most $r-1$ marked vertices to choose for a Type I segment with $r$ marked edges. In contrast, for a Type II segment we have an upper bound of $r$ choices for $r$ marked edges.

\subsubsection{Computing \texorpdfstring{$E_{m,\rmin}$}{E{m,r*(m)}}.} 
We first find an expression of $\rmin$. Note the following:
\begin{obs} \normalfont 
 A Type I path has maximum number of marked edges if there is only one Type I segment having $t$ marked edges. 
 We refer to these paths \emph{maximal Type I} paths. (Thus,  maximal paths are saturated). 
\end{obs}

\begin{obs} \normalfont 
    For a Type II path, let $h = h(t,i,j)$ be the maximum number of marked edges. 
    Note that each marked edge in a Type II segment has at least one backtrack. We refer to such paths as \emph{maximal Type II} paths.
    Note that $h = t/2$ if $t$ is even and $i=j$ as each marked edge in a Type II segment has at least one backtracking edge. 
	If $t$ is even and $i \neq j$ then $h=(t-2)/2$ as Type II segments must start and end at the same vertex and Type II segments are of an even length.
	If $t$ is odd and $i \neq j$, then $h = (t-1)/2$ marked edges. 
	If $t$ is odd and $i = j$ then it may have a maximum of $(t-3)/2$ marked edges. In particular for the last case, we cannot have $(t-1)/2$ marked edges as if that were the case then the first or the last edge in the path would be a self-loop at node $i$ which has probability $0$ in our model. 
	We can also note that as a result, we cannot have a Type II path when $t=3$ and $i=j$ and we take $h = 0$ in the case. Note that $h\leq \phi(t) = \lfloor t/2\rfloor$ for $t\geq 3$, 
	where $\phi(t)$ is defined in Theorem~\ref{prop:Frob norm tL to tU}. 
\end{obs}
We have that for the case of $r$ Type I segments, the number of marked edges satisfy $m \leq rt + (k-r)h$. Let 
	\[	f(r) := rt + (k-r)h, \quad 0 \leq r \leq k.	\]
Then $\rmin$ is obtained by inverting $f$ as described below: 
\begin{lem}
    $\rmin = \max\{0, \lfloor\frac{m-kh}{t-h}\rfloor\}$.
\end{lem}
\begin{proof}
Let $f(r_0-1)<m \leq f(r_0)$ for some $r_0\geq 1$. We would like to put as many edges in Type~I segments as possible to minimize the number of Type~I segments. However, since each path has length $t$ and $m > (r_0-1)t$, we need $r_0$ Type I segments, so that $\rmin \geq r_0$. Also, this lower bound is attained by taking $r_0-1$ many maximal Type I paths and another Type I path that is not maximal. 
If $m\leq f(0)$, then we can get away with having no Type I segments. This completes the proof. 
\end{proof}
Next we count the possible ways of rearranging the Type I segments in $k$ paths.
\begin{lem}\label{lem:m-rmin-arrangement}
Given $m$ marked edges and $\rmin$ Type I segments, the number of configurations of Type~I,~II segments and unmarked edges $N_{m,\rmin}$ satisfies
\begin{align*}
N_{m,\rmin} \leq \binom{k}{\rmin} k^{f(\rmin) - m} 3^{k-\rmin} C^{f(\rmin) - m}.
\end{align*}
\end{lem}
\begin{proof}
We treat four cases separately: \\

\noindent \textbf{Case I:} $m = f(r_0)$ for some $1 \leq r_0 \leq k$. $\rmin = r_0$. If $m = f(r_0)$, then we must take $r_0$ paths to place the $r_0 = \rmin$ Type I segments which can be chosen in $\binom{k}{r_0}$ ways. The rest of the Type II segments are placed in the remaining $k-r_0$ paths.
The Type I, II paths have to be maximal in order to place these $m$ marked edges. 
We note that there are at most $t - 2 h$ unmarked edges in each of these Type II paths, which are not part of the segments, and these can be chosen in at most $m$ ways each. This is because each unmarked edge is chosen to be one of the $m$ marked edges.
We also note that for each Type II path, there are $t-2h \leq 3$ ways of placing the Type II segment in the path. So the overall bound for Case I is
\begin{align*}
    \binom{k}{\rmin} 3^{(k - \rmin)}.
\end{align*}

\noindent \textbf{Case II:} $f(r_0-1)< m <f(r_0)$ for some $1 \leq r_0 \leq k$. Again, we have $\rmin = r_0$.
Let $l = f(r_0) - m$. 
Then we need $r_0$ distinct paths to place the $r_0 = \rmin$ Type I segments, which again can be chosen in $\binom{k}{r_0}$. However, in this case, these paths might not be maximal. 
Take $l_1$ and $l_2$ such that $l_1 + l_2 \leq l$, all except $l_1$ of the $r_0$ Type I paths are maximal and and all except $l_2$ of the $k-r_0$ Type II paths are maximal. 
We can choose these $l_1 + l_2$ many non-maximal paths in at most $O(k^l)$ ways. 
Also, note that we can have at most $2l$ more unmarked edges compared to the case of $m = f(r_0)$. 
The total number of ways of arranging the segments and unmarked edges for the non-maximal paths is $O(C^l)$, where $C$ is a constant that might depend on $t$. Thus the bound for this case is 
\begin{align*}
    \binom{k}{\rmin} k^l 3^{(k - \rmin)} C^l.
\end{align*}

\noindent \textbf{Case III:} $m = f(0)$. The same argument and the bound as Case I holds here. We note that this case does not occur if $h=0$. \\

\noindent \textbf{Case IV:} $m < f(0)$.
Thus $m = f(0) - l$ for $l > 0$. Note that this case also does not occur if $h = 0$. Recall that for $m = f(0)$ case, all the paths have Type II segments containing $h$ marked edges. 
	We would like to count the configurations for $m = f(0) - l$ marked edges with no Type I segments.
	We note that then there exist $1 \leq u \leq l$ paths such that they have less than $h$ marked edges and the rest of the paths have $h$ marked edges.
	We can choose the $u$ paths in at most $O(k^l)$ ways. We can then arrange segments in these $l$ paths in $C^l$ ways. 
	Note that we can have at most $2l$ more unmarked edges for this case compared to the $m = f(0)$ case. The overall bound for this case is $(Ck)^l 3^k$.
\end{proof}
We are now ready to prove asymptotics of $E_{m,\rmin}$. 
\begin{lem}\label{lem:E-m-rmin-asymptotics deepwalk}
    $\sum_{m<kt}E_{m,\rmin} = o(U_b^k)$. 
\end{lem}
\begin{proof}
Let us start by noting that $U_b = \Theta (n^{t-1} \rho_n^t)$. 
The number of choices of segments is given by  Lemma~\ref{lem:m-rmin-arrangement}. 
Since we have $\rmin$ Type I segments, we must have at least $s(m) = \max\left\{\rmin - (f(\rmin) - m), 0\right\}$  maximal Type I segments, with each of them having probability at most $U_b$. In the rest, we have one Type I segment. 
Thus by Lemma~\ref{lem:choice-vertices}, the vertices in the rest of the paths can be chosen in at most $n^{m - ts(m) - (\rmin - s(m))}$ ways. 
Also the $m-  t s(m)$ many remaining marked edges give us a contribution of at most $(C \rho_n)^{m- t s(m)}$.
The number of unmarked edges is specified in each of the cases in the proof of Lemma~\ref{lem:m-rmin-arrangement} above.
Combining all these, we get 
\begin{align*}
   E_{m,\rmin} &\leq \binom{k}{\rmin} k^{f(\rmin) - m} 3^{k-\rmin} C^{f(\rmin) - m}\\
   &\quad \times U_b^{s(m)} n^{m - ts(m) - (\rmin - s(m))} 
   \rho_n^{m- t s(m)}\\
   &\quad \times m^{(k-\rmin) \cdot (t - 2 h) + 2 (f(\rmin) - m)}.
\end{align*}
Using these bounds we see that with the choice of $k = \ceil{\log n}$
\begin{align*}
\sum_{l=0}^{t-h-1} E_{f(r_0)-l,r_{\star}(f(r_0)-l)}  &=  E_{f(r_0)-l,r_{\star}(f(r_0)-l)} \left( 1 + O\left( \frac{k^3}{n\rho_n} \right) \right), \quad r_0 \geq 1,\\ 
\sum_{l=0}^{f(0)-1} E_{f(0)-l,r_{\star}(f(0)-l)}  &=  E_{f(0)-l,r_{\star}(f(0)-l)} \left( 1 + O\left( \frac{k^3}{n\rho_n} \right) \right),\\
\sum_{r_0 = 0}^k E_{f(r_0),r_{\star}(f(r_0))}  &= E_{f(k),r_{\star}(f(k))} \left( 1 + O\left( \frac{k^{1 +  t - 2h} (n\rho_n)^{h}}{n^{t-1}\rho_n^t} \right) \right).
\end{align*}
These bounds in turn imply that
\begin{align*}
   \sum_{m=1}^{kt-1} E_{m, \rmin} =   o\left( U_b^k \right).
\end{align*}

\end{proof}

\subsubsection{Computing \texorpdfstring{$E_{m,r}$}{E{m,r}} for \texorpdfstring{$r>\rmin$}{r > r*(m) }.} 
We start by noting the additional number of configurations with $r$ Type I segments as compared to Lemma~\ref{lem:choice-vertices}. 
\begin{lem}\label{lem:num arrangement deepwalk}
Given $m$ marked edges and $r$ Type I segments, let $N_{m,r}$ be the number of configurations of Type~I,~II segments and unmarked edges. 
Then, for any $r>\rmin$,
\begin{align*}
    N_{m,r} \leq  N_{m,\rmin} \times (Ck)^{(r-\rmin) (t+1)}. 
\end{align*}
\end{lem}
\begin{proof}
\begin{figure}
    \centering
    \includegraphics[scale = .05]{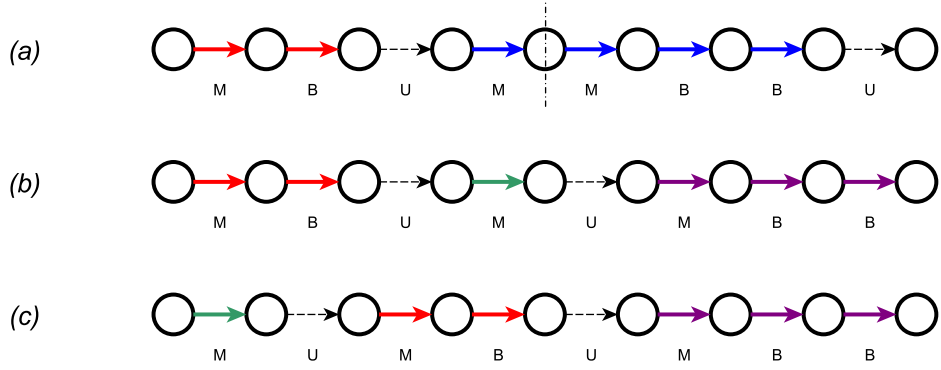}
    \caption{Example of a splitting a segment. $(a)$ The original configuration of the segments in the chosen path. The edges in the two segments are colored red and blue respectively. The edges are also labeled by M, B and U if the edge is a marked edge, a backtrack of a marked edge and an unmarked edge respectively. The dotted line indicates the location of the split.
    $(b), (c)$ Two examples of configurations after the splitting the segment. The two new segments are colored by green and violet.  
    }
    \label{fig:segment_split}
\end{figure}
Let $T_r$ be the set of all configurations of $m$ marked edges and $r$ Type I segments. Note that a ``configuration'' here only specifies location of marked edges, segments, and the unmarked edges. The backtracks of the marked edges are already specified.  
We will inductively bound $T_{r+1}$ in terms $T_r$. 
For that, we consider two cases depending on whether the elements of $T_{r+1}$ has a  Type I path with two Type I segments or not. 
In both cases, we will find a relation between  $T_r$ to $T_{r+1}$. 

\paragraph*{Case I. } Consider elements of  $T_{r+1}$ that has a Type I path with two Type I segments. Let us call this subset $T_{r+1}^{\sss I}$.
We consider the elements in $T_r$ which will be related to these elements in $T_{r+1}^{\sss I}$. 
Let $T_r^{\sss I} \subset T_r$ consisting of configuration such that there is at least one path $p$ so that the following hold: 
\begin{enumerate}
    \item \textbf{Extra unmarked edges.} $p$ has $l$ segments and $l'\geq l$ unmarked edges for some $l,l' \geq 1$. 
    \item \textbf{Well-splittable.} Suppose $p$ has a Type I segment $(i_l, i_{l+1}, \ldots, i_{l'})$ with marked edges given by $(i_{k_s} , i_{k_{s}+1})$ for $l \leq  k_s  < l'$, $1 \leq s \leq m'$, $m'\geq 1$,
    and the number of $s$ with $k_{s+1} - k_{s}$ being odd is at least $2$. 
\end{enumerate}
Note that a path with $l$ segments can be formed by just putting $l-1$ unmarked edges between segments. The first condition ensures that we have extra $l'-l+1$ of them. We will put these extra unmarked edges inside segments to split them. 
Regarding condition 2, notice that a Type I segment always has one of the $k_{s+1}-k_s$ being odd (by definition), but the well-splittable condition requires $k_{s+1}-k_s$ to be odd additionally in a second place.
This allows us to split a well-splittable Type I segment into two Type I segments.
For example, the blue segment in Figure~\ref{fig:construction 2}~(a) is well-splittable and it can be split into two parts with $k_{s+1} - k_s$ being odd.
We will split the segment by moving an unmarked edge in between as illustrated in Figure~\ref{fig:construction 2}. 

The general description for splitting is that given a path with extra marked edges and well-splittable property, we can think of unmarked edges as ``bars'', and segments as ``labelled balls''. The well-splittable segment is viewed as two Type I sub-segments and corresponds to two ``labelled balls''. Permute these bars and balls such that there is at least one bar between any two labelled balls (bars can be adjacent). 
Also, there may be multiple ways to split  well-splittable segment and we consider all possible such splits. 
This creates multiple elements in $T_{r+1}^{\sss I}$ from an element in $T_{r}^{\sss I}$. See Figure~\ref{fig:segment_split}~(b),~(c) for two possible elements. 
Moreover, we can get preimages of all the elements in $T_{r+1}^{\sss I}$,. 
To do this, we can first take a path with two Type I segments, rearranging the segments so that two Type I segments appear one after another. Then we can remove the unmarked edges between them, which joins the two Type I segments. The removed unmarked edges can be placed in other places adjacent to an unmarked edge. This is basically the inverse operation of the above splitting constructions. 

We can note there are $O(k)$ preimages in total of an element in $T_{r+1}^{\sss I}$, where the factor $k$ comes from choice of the path $p$ and the rest of the choices for permuting unmarked edges and segments are $O(1)$ as they are functions of only $t$ which is bounded.
Thus, 
\begin{align*}
    |T_{r+1}^{\sss I}| \leq O(k) \times |T_{r}^{\sss I}| \leq O(k) \times |T_{r}|.
\end{align*}

\paragraph*{Case II.} 
\begin{figure}[ht]
    \centering
    \includegraphics[scale = .05]{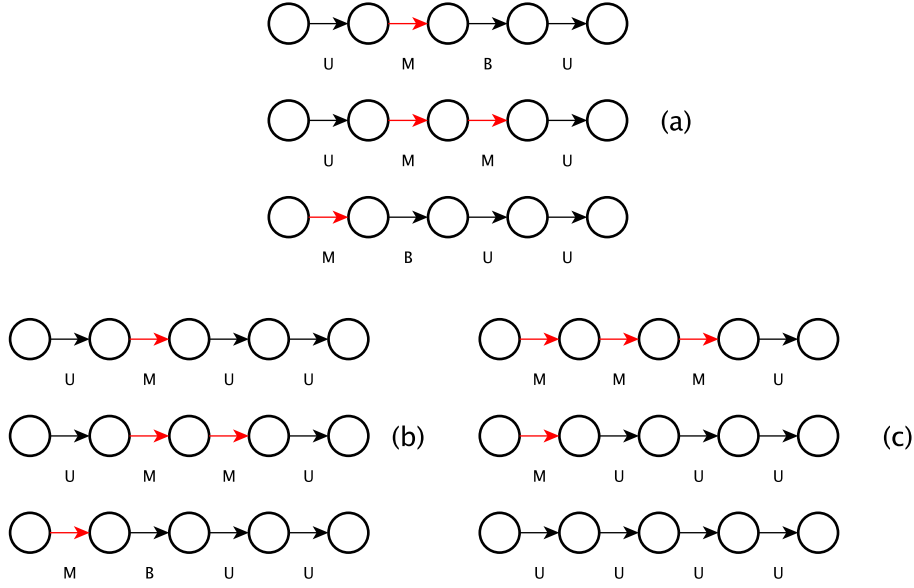}
    \caption{Example of the second construction for the case of $t=4$, $m=4$ and $s=1$. The marked edges are colored red. The letters M, B and U denote a marked edge, a backtrack of a marked edge and an unmarked edge respectively.
    $(a)$ The top path is the chosen path where we would place a Type I segment. The second path contains one Type I segment and the bottom path is a Type II path.
    $(b), (c)$ Two examples of new configurations from the construction. In both the cases the top path now has a Type I segment and the middle path continues to have a Type I segment.
    }
    \label{fig:construction 2}
\end{figure}
We next relate arrangements of marked edges where there is at most one Type I segment per path.
We denote such arrangements as $T_{r+1}^{\sss II} \subset T_{r+1}$. 
Let $T_{r}^{\sss II}$ be the set of all ways of specifying locations of segments such that there are a total of $r$ Type I segments and each of the Type I segments is placed on distinct paths.
Suppose that $r+1 \leq \min(k, m)$ and $r \geq r_{\star}(m)$.
Then we give a multi-map from $T_{r}^{\sss II}$ onto $T_{r+1}^{\sss II}$ using a construction. 
The condition $r+1 \leq m$ is necessary so that $T_{r+1}^{\sss II}$ is non-empty as we must have at least $r+1$ marked edges in order to have $r+1$ paths each having a Type I segment. 
The condition $r+1 \leq k$ is also necessary for $T_{r+1}^{\sss II}$ to be non-empty as we must have at least $r+1$ distinct paths to place the $r+1$ Type I segments.
The last condition $r \geq r_{\star}(m)$ is to ensure that $T_{r}^{\sss II}$ is non-empty. 
To fix notation, let $S_1(l)$ be the set of all ways of arranging $l$ marked edges in one path such that there is at least one Type~I segment in the path. 
Similarly, let $S_{2}(l)$ be the set of all ways of arranging $l$ marked edges in one path such that there are no Type I segments in the path.
We now describe the construction. 
Let $A \in T_{r}^{\sss II}$ and choose a path $p$ not containing a Type I segment. We can do so as $r+1 \leq k$.
Suppose $p$ has $l \geq 0$ marked edges. Choose $0 \leq u_1 + u_2 \leq t-l$ paths where $u_1$ of the paths are paths containing Type I segments with at least two marked edges and the rest $u_2$ paths are paths which are distinct from $p$, do not contain Type~I segments, and contain at least one marked edge.
If $l=0$ we require $u_1 + u_2 > 0$. It is feasible to choose such path(s) as $r+1 \leq m$.
Suppose the $u_1$ paths are labeled as $q_1, q_2, \ldots, q_{u_1}$ and the $u_2$ paths indices $q'_1, q'_2, \ldots, q'_{u_2}$. Suppose these paths have $l_{q_1}, l_{q_2}, \cdots, l_{q_{u_1}}$ and $l_{q'_1}, l_{q'_2}, \cdots, l_{q'_{u_2}}$ marked edges respectively.
Let $v_{q_1}, v_{q_2}, \cdots, v_{q_{u_1}}$ be such that $0 < v_{q_i} < l_{q_i}$.
Similarly let $v_{q'_1}, v_{q'_2}, \cdots, v_{q'_{u_2}}$ be such that $0 < v_{q'_i} \leq l_{q'_i}$. 
We require $\sum_i v_{q_i} + \sum_i v_{q'_i} \leq t-l$.
Then we modify the arrangements of marked edges in the paths so that the new arrangements for the sequence of paths $(p, q_1, q_2, \ldots, q_{u_1}, q'_1, q'_2, \ldots, q'_{u_2})$ is any element of
\begin{align*}
S_1\left(l + \sum_i v_{q_i} + \sum_i v_{q'_i}\right) \times \left(\prod_{i=1}^{u_1} S_1(l_{q_i} - v_{q_i})  \right) \times \left(\prod_{i=1}^{u_2} S_{2}(l_{q'_i} - v_{q'_i})  \right).
\end{align*}
We keep the arrangements of marked edges in the rest $k - (1 + u_1 + u_2)$ paths unchanged. 
This leads to multiple images of $A$ in $T_{r+1}^{\sss II}$.
We note that there are $O(k^{t+1})$ images of $A$ due to the choice of the paths and since the number of ways of choosing the new arrangements for the $1 + u_1 + u_2$ paths is $O(1)$ as $t$ is fixed.
We also note that since we modify at most $t+1$ paths in this construction, the number of unmarked edges increases by at most $t(t+1)$.

Now we show that the multi-map given by the construction above is surjective onto $T_{r+1}^{\sss II}$. For this, let $A' \in T_{r+1}^{\sss II}$. Let $p$ be a path containing a Type I segment. If $p$ has $l \leq h$ edges we choose an element in $x \in S_{2}(l)$. We can then see that $A'$ is in the image of the element $A \in T_{r}^{\sss II}$ where the path $p$ has the arrangement of marked edges given by $x$ and the rest of the paths have the same arrangement of marked edges as $A'$. We note that since we have replaced the Type I segment in the path $p$ with a Type II segment, the element $A$ has one less Type I segment. 
Now suppose that $p$ has $l > h$ edges. Choose $u_1 + u_2 \leq l - h$ paths so that $u_1$ paths contain Type I segments and these paths are not equal to $p$ and, $u_2$ paths do not contain Type I segments segments.  		
Suppose the $u_1$ paths are labeled as $q_1, q_2, \ldots, q_{u_1}$ and the $u_2$ paths are labeled as $q'_1, q'_2, \ldots, q'_{u_2}$. Suppose that these paths have $l_{q_1}, l_{q_2}, \cdots, l_{q_{u_1}}$ and $l_{q'_1}, l_{q'_2}, \cdots, l_{q'_{u_2}}$ marked edges respectively. We require that $l_{q_i} < t$ and $l_{q'_i} < h$ i.e. that these chosen paths are not saturated.
Let $v_{q_1}, v_{q_2}, \cdots, v_{q_{u_1}}$ be such that $0 < v_{q_i} \leq t - l_{q_i}$.
Similarly let $v_{q'_1}, v_{q'_2}, \cdots, v_{q'_{u_2}}$ be such that $0 < v_{q'_i} \leq h - l_{q'_i}$.  
We require that 
\begin{align*}
	\sum_{i=1}^{u_1} v_{q_i} + \sum_{i=1}^{u_2} v_{q'_i} = l -h.
\end{align*}
This is feasible as long as $r \geq r_{\star}(m)$. The above condition says that the chosen paths have enough spaces to move $l- h$ edges from path $p$ to the chosen paths. 
Choose a new arrangement $x$ of the marked edges for the sequence of paths $(p, q_1, q_2, \ldots, q_{u_1}, q'_1, q'_2, \ldots, q'_{u_2})$ from the set
\begin{align*}
	S_{2}\left(h\right) \times \left(\prod_{i=1}^{u_1} S_1(l_{q_i} + v_{q_i})  \right) \times \left(\prod_{i=1}^{u_2} S_{2}(l_{q'_i} + v_{q'_i})  \right).
\end{align*} 
Then keeping the arrangements of marked edges of the rest of the paths the same as in $A'$ and choosing the arrangements for the chosen paths as $x$, we have a preimage $A \in T_{r}^{\sss II}$ under the construction described above.

From the two constructions above we have
\begin{align*}
	\vert T_{r+1} \vert \leq  O(k) \vert T_r \vert + O(k^{t+1}) \vert T_r \vert. \numberthis \label{eqn:bound for segments}
\end{align*}
			
\end{proof}

We can now compute asymptotics for $E_{m,r}$. 
\begin{lem}\label{lem:sum over r>rmin deepwalk}
$\sum_{r= \rmin}^{\rmax} E_{m,r} =  E_{m, \rmin} (1 + o(1))$.
\end{lem}
\begin{proof}
We start by giving a bound for $E_{m,r}$.
The probability of the $m$ marked edges is bounded by $\rho_n^{m}$. 
By Lemma \ref{lem:choice-vertices}, the upper bound for the number of marked edges is $m-r$ as there are $r$ Type I segments.
Let 
\[u(m) = (k-\rmin) \cdot (t - 2 h) + 2 (f(\rmin) - m),\] 
be the upper bound for the number of unmarked edges for the case of $\rmin$ segments obtained from the proof of Lemma \ref{lem:E-m-rmin-asymptotics deepwalk}.
Then by the two constructions in Lemma \ref{lem:num arrangement deepwalk}, the number of unmarked edges for the case of $r$ Type I segments is at most $u(m) + (r- \rmin) \cdot (t(t+1))$.
Combining these we have
\begin{align*}
E_{m,r} \leq N_{m,r} n^{m-r} \rho_n^{m} m^{u(m) + (r- \rmin) \cdot t(t+1)}. 
\end{align*}
Then by using Lemma \ref{lem:num arrangement deepwalk} we have
\begin{align*}
\sum_{r= \rmin}^{\rmax} E_{m,r}  
= E_{m,\rmin} \sum_{r= \rmin}^{\rmax} \left(\frac{O(k^{t+1}) \cdot m^{t(t+1)}}{n}\right)^{r - \rmin}
= E_{m,\rmin} (1 + o(1)).
\end{align*}
\end{proof}
By Lemmas \ref{lem:E-m-rmin-asymptotics deepwalk} and \ref{lem:sum over r>rmin deepwalk} we then have
\begin{align*}
\sum_{m < kt} E_m \leq C \sum_{m < kt} E_{m, \rmin} = o(U_b^k).
\end{align*}
This completes the proof of Proposition \ref{prop:A^t kth moment-2}.

\subsection{Concentration of path counts}\label{sec:concent-deepwalk}
In this section, we prove concentration of $Y_b$ defined in \eqref{eq:Y-b-def}. Let us start with a general result which we will apply for indicator random variables appearing in paths: 
\begin{lem}\label{lem:expectation centered moment}
Let $\cE$ be a finite set and let $\xi_{e}$ be indicator random variables for $e \in \cE$ with $\pi_e = \PR(\xi_{e} = 1) > 0$. 
For a subset $S \subset \cE$, we define  $\xi_S := \prod_{e\in S} \xi_e$.  
Suppose $S_1, S_2, \ldots S_k$ be non-empty subsets of $\cE$ and let $n_e:= |\{j\in [k]: e\in S_j\}|$.
Let
$\cE_1 := \{	e \in \cup_{j=1}^k S_j : n_e = 1	\}$ and $\cE_2 = \left(\cup_{j=1}^k S_j \right) \backslash \cE_1$.  
Then we have
\begin{align} \label{eq:bound-kth-mom-expt}
\bigg| \E \prod_{j=1}^{k} (\xi_{S_j} - \E \xi_{S_j} ) \bigg| \leq   \bigg(\prod_{e \in \cE_1} \pi_e\bigg) \times \bigg(\prod_{e \in \cE_2} \pi_e (1 + \pi_e)\bigg).
\end{align}
\end{lem}
\begin{proof}
Let $S_{j}^{\sss (1)} \subseteq S_{j}^{\sss (2)} \subseteq S_j$ and $m = | \cup_{j=1}^k S_{j}^{\sss (1)} |$. 
We define  $n_{e}^{\sss (1)}$ and $\cE_1^{\sss (1)}$ similarly as $n_e$,  $\cE_1$ but now with sets $S_{j}^{\sss (1)} $. 
We prove by induction on $m$ that, for all possible choices of $S_{j}^{\sss (1)}, S_{j}^{\sss (2)}$, 
\begin{eq} \label{eq:gen-bound-kth-mom-expt}
\bigg| \E \prod_{j=1}^{k} \big(\xi_{S_j^{\sss (1)}} - \E \xi_{S_j^{\sss (2)}} \big) \bigg| \leq   \bigg(\prod_{e \in \cE_1^{\sss (1)}} \pi_e\bigg) \times \bigg(\prod_{e \in \cE_2^{\sss (1)}} \pi_e (1 + \pi_e)\bigg). 
\end{eq}
Throughout, we use the convention that a product over an empty index set is $1$. 
For the induction base case, suppose $m=1$. Let $e$ be the only element in $\cup_{i=1}^k S_j^{\sss (1)}$ and without loss of generality let $e \in S_j^{\sss (1)}$ for $1 \leq j \leq n_{e}^{\sss (1)}$. 
If $n_{e}^{\sss (1)}=1$, then 
\begin{align*}
	\E \big(\xi_{S_1^{\sss (1)}} - \E \xi_{S_1^{\sss (2)}} \big) = \pi_e \big(	1 - \E \xi_{S_1^{\sss (2)}}'	\big) \leq \pi_e,
\end{align*} 
where $\xi_{S_1^{\sss (2)}}' := \xi_{S_1^{\sss (2)} \setminus \{e\}}$. 
The terms with $S_j^{\sss (1)} = \varnothing$ can be bounded by 1.
This shows the base case for $m=1$ and $n_{e}^{\sss (1)}=1$. If $n_{e}^{\sss (1)} > 1$,  then 
\begin{align*}
	&\E \bigg[\prod_{j=1}^{n_{e}^{\sss (1)}} \big(\xi_{S_j^{\sss (1)}} - \E \xi_{S_j^{\sss (2)}} \big) \bigg] 	 = \E \bigg[\prod_{j=1}^{n_{e}^{\sss (1)}} \big(\xi_e - \E \xi_{S_j^{\sss (2)}} \big) \bigg]\\
	&= \pi_e \cdot \prod_{j=1}^{n_{e}^{\sss (1)}} \big(1 - \E \xi_{S_j^{\sss (2)}} \big)+ (1- \pi_e) \cdot \prod_{j=1}^{n_{e}^{\sss (1)}} \big( - \E \xi_{S_j^{\sss (2)}} \big)\leq \pi_e+ (\pi_e)^{n_{e}^{\sss (1)}} \leq \pi_e(1 + \pi_e).
\end{align*}
This completes the proof of the base case $m=1$. 
			
Next let $m > 1$. We will split in two cases depending on whether $\cE_{1}^{\sss(1)} \neq \varnothing$ and $\cE_{1}^{\sss(1)} = \varnothing$. First, if $\cE_{1}^{\sss(1)} \neq \varnothing$, then pick an element $e \in \cE_{1}^{\sss(1)}$. 
Note that $e$ can only be in one of $S_j^{\sss (1)}$'s, and 
without loss of generality let $e \in S_1^{\sss (1)}$. 
Let $\sF_e$ is the minimum sigma-algebra with respect to which $(\xi_f)_{f\neq e}$ are measurable, and analogously to before, define $\xi_{S_1^{\sss (1)}}' := \xi_{S_1^{\sss (1)} \setminus \{e\}}$ and $\xi_{S_1^{\sss (2)}}' := \xi_{S_1^{\sss (2)} \setminus \{e\}}$. Taking iterated conditional expectation with respect to $\cF_e$, we get
\begin{align*}
	\E \bigg[\prod_{j\in [k]: S_j^{\sss (2)} \neq \varnothing}  \big(\xi_{S_j^{\sss (1)}} - \E \xi_{S_j^{\sss (2)}} \big) \bigg] =  \pi_e \cdot \E \bigg[	\big(\xi_{S_1^{\sss (1)}}' - \E \xi_{S_1^{\sss (2)}}' \big) \prod_{j \in \{2, 3,\ldots k\} : S_j^{\sss (2)} \neq \varnothing } \big(\xi_{S_j^{\sss (1)}} - \E \xi_{S_j^{\sss (2)}} \big)	\bigg].
\end{align*}
Now we can conclude \eqref{eq:gen-bound-kth-mom-expt} by induction step using    $S_1^{\sss (1)} \setminus \{e\}$,  $S_1^{\sss (2)} \setminus \{e\}$ and $S_j^{\sss (1)}$, $S_j^{\sss (2)}$ for $j\geq 2$, and noting that $\cE_1^{\sss (2)}$ is unchanged in this new set up. 
			
Next suppose that $\cE_1^{\sss (1)} = \varnothing$. Pick any $e \in \cE_2^{\sss (1)}$. 
Then $n_e^{\sss (1)} > 1$ and without loss of generality let $e \in S^{\sss (1)}_j$ for $1 \leq j \leq n_e^{\sss (1)}$.
We again take iterated conditional expectation with respect to $\cF_e$ to get
\begin{eq}\label{eqn:condition on Y_a}
	&\E \bigg[\prod_{j\in [k]: S_j^{\sss (2)} \neq \varnothing}  \big(\xi_{S_j^{\sss (1)}} - \E \xi_{S_j^{\sss (2)}} \big) \bigg] \\
	&=\E\bigg[\prod_{j \in \{ n_e^{\sss (1)} + 1,\dots,k\} : S_j^{\sss (2)} \neq \varnothing } \big(\xi_{S_j^{\sss (1)}} - \E \xi_{S_j^{\sss (2)}} \big)  \E\bigg[ \prod_{j=1 }^{n_e^{\sss (1)}} \big(\xi_{S_j^{\sss (1)}} - \E \xi_{S_j^{\sss (2)}} \big) \ \Big| \cF_e  \bigg] \bigg]. 
\end{eq}
We simplify 
\begin{align}\label{eqn:integrate Y_a}
	&\E\bigg[ \prod_{j=1 }^{n_e^{\sss (1)}} \big(\xi_{S_j^{\sss (1)}} - \E \xi_{S_j^{\sss (2)}} \big) \ \Big| \cF_e  \bigg] = \pi_e \cdot \prod_{j=1 }^{n_e^{\sss (1)}} \big(\xi_{S_j^{\sss (1)}}' - \E \xi_{S_j^{\sss (2)}} \big)  + (1 - \pi_e) \prod_{j=1 }^{n_e^{\sss (1)}} \big(-\E \xi_{S_j^{\sss (2)}} \big). 
\end{align}
Suppose that $m-l = \lvert \cup_{j=n_e^{\sss (1)}+1}^k  S^{\sss (2)}_j \rvert$ for some $l \geq 1$ (we have $l \geq 1$ since $e$ is not in the union).
For the second term in \eqref{eqn:integrate Y_a}, we have 
\begin{align*}
	\bigg|  (1 - \pi_e) \prod_{j=1 }^{n_e^{\sss (1)}} (-\E \xi_{S_j^{\sss (2)}}) \bigg| \leq  \pi_e^{n_e^{\sss (1)}-1} \prod_{f \in \cup_{j=1}^{n_e^{\sss (1)}}  S^{\sss (2)}_j} \pi_{f} =:Z_1,
\end{align*}
as the term $\xi_e$ is repeated $r$ times. 
We bound the term in \eqref{eqn:condition on Y_a} outside conditional expectation by induction. 
If we consider the sets $S_j^{\sss (1)}$,  $S_j^{\sss (2)}$ for $j >n_e^{\sss (1)}$ with $S_j^{\sss (2)} \neq \varnothing$ and create the sets $\tilde{\cE}_{1}$ and $\tilde{\cE}_{2}'$ analogously to $\cE_{1}^{\sss (1)}$, $\cE_{1}^{\sss (2)}$ when restricted to this smaller class of subsets. 
Then
\begin{align*}
	\bigg|\E\bigg[\prod_{j \in \{ n_e^{\sss (1)} + 1,\dots,k\} : S_j^{\sss (2)} \neq \varnothing } \big(\xi_{S_j^{\sss (1)}} - \E \xi_{S_j^{\sss (2)}} \big) \bigg]\bigg| \leq \prod_{f \in \tilde{\cE}_{1}} \pi_f \prod_{f \in \tilde{\cE}_{2}} \pi_f (1 + \pi_f)=: Z_2.
\end{align*}
Hence, the term in \eqref{eqn:condition on Y_a} is at most 
\begin{align*}
	& \pi_e \bigg| \E\bigg[\prod_{j \in \{ n_e^{\sss (1)} + 1,\dots,k\} : S_j^{\sss (2)} \neq \varnothing } \big(\xi_{S_j^{\sss (1)}} - \E \xi_{S_j^{\sss (2)}} \big)   \prod_{j=1 }^{n_e^{\sss (1)}} \big(\xi_{S_j^{\sss (1)}}' - \E \xi_{S_j^{\sss (2)}}  \bigg]
	\bigg|  + Z_1Z_2\\
	&\leq \pi_e \cdot \left(\prod_{f \in \cE_1^{\sss (1)}} \pi_f \cdot \prod_{f \in \cE_2^{\sss (1)} \setminus \{e\}} \pi_e (1 + \pi_e) \right) + \pi_e^{n_e^{\sss (1)}-1} \prod_{f \in \cup_{j=1}^{n_e^{\sss (1)}}   S^{\sss (2)}_j} \pi_{f} \prod_{f \in \tilde{\cE}_{1}} \pi_f \prod_{f \in \tilde{\cE}_{2}'} \pi_f (1 + \pi_f). 
\end{align*}
The last term above is at most the bound in \eqref{eq:gen-bound-kth-mom-expt}, which follows by noting that the second term has a factor of $(\pi_e)^{n_e^{\sss (1)}}$ with  $n_e^{\sss (1)} \geq 2$, and for each of the variables $f \in \cE_1^{\sss(1)}, f \neq e$ the second term has a factor smaller or equal to $\pi_f$, and for each of the variables $f \in \cE_1^{\sss(1)}$ the second term has a factor smaller or equal to $\pi_f (1+\pi_f)$.
This completes the proof.
\end{proof}
We now prove the following concentration result for $Y_b$: 
\begin{prop}\label{prop:Lower tail A^t}
Let $\tL = \tU = t \geq 3$ be given and $k = \ceil{\log n}$. Suppose that \eqref{eq:t-condn} holds.
Then we have 
\begin{align*}
\Prob\left(|Y_b - \E Y_b| > \delta \E Y_b\right)  = O(n^{-c}),
\end{align*}
where $\delta = \Theta((\log n)^{-\eta})$ for some $\eta > 0$, and $c>0$ is any real number. 
\end{prop}

\begin{proof}
We will be using notation and terminology from proof of Proposition \ref{prop:A^t kth moment-2}. By Markov's inequality, and using Proposition~\ref{prop:A^t kth moment}, 
\begin{align} \label{eqn:lower tail bound 1}
	\Prob\left(|Y_b - \E Y_b| > \delta \E Y_b\right) 
	&\leq \frac{\E (Y_b - \E 	Y_b)^{2k}}{\delta^{2k} \left(	\E Y_b\right)^{2k}} \leq  \frac{\E (Y_b - \E 	Y_b)^{2k}}{\delta^{2k} (L_b)^{2k}},
\end{align}
and moreover, 
\begin{align*}
	\E (Y_b - \E 	Y_b)^{2k} 
	&= \E \bigg(\sum_{p \in \cP_b} (X_p - \E X_p)	\bigg)^{2k}= \sum_{(p_1, p_2, \ldots, p_{2k}): p_l \in \cP_b} \E \prod_{l=1}^{2k} (X_{p_l} - \E X_{p_l}). \numberthis\label{eqn:lower tail bound 2}
\end{align*}
Fix an ordered sequence $(p_1, p_2, \ldots, p_{2k})$. We can first note that $\E \prod_{l=1}^{2k} (X_{p_l} - \E X_{p_l})$ is equal to $0$ if there is a path $p_l$ which does not have any edges in common with the other $2k-1$ paths. 
Suppose now that each of the paths share at least one edge with some other path. 
Then the minimum number of unmarked edges or repeats of edges is at least $k$. This  minimum $k$ arises from the worst case where we have $k$ pairs of paths with each pair having one edge in common. 
Thus, the number of marked edges $m\leq 2kt-k$.

To bound $\E \prod_{l=1}^{2k} (X_{p_l} - \E X_{p_l})$, we use Lemma~\ref{lem:expectation centered moment}. 
For each marked edge $e$ with endpoints in block $b_i$ and $b_{i'}$ between two nodes blocks $b_i$ and $b_{i'}$, we assign a weight $w(e) = B_{b_i, b_{i'}}$. 
For each unmarked edge $e$, we assign a weight $w(e) = (1+ B_{b_i, b_{i'}})$ instead. 
Then we can see by Lemma~\ref{lem:expectation centered moment}, 
$\E \prod_{l=1}^{2k} (X_{p_l} - \E X_{p_l})$ is bounded by the product of the weights on the edges in the $2k$ paths, i.e., 
\begin{align}
\E \prod_{l=1}^{2k} (X_{p_l} - \E X_{p_l}) \leq \prod_{e: e\text{ is marked}}w(e) \times \prod_{e: e\text{ is unmarked}}(1+w(e)). \label{eqn:lower tail edge weights} 
\end{align}
We note that while bounding $E_m$ in the proof of Proposition~\ref{prop:A^t kth moment-2}, we bounded the probability of each of the marked edges by $w(e)$ as well. 
Let $E'_m$ be the sum of summands in \eqref{eqn:lower tail bound 2} corresponding to paths with $m$ marked edges.
We use the bound $1 + w(e) \leq 2$ for the weight on the unmarked edges and proceed as in the proof of Proposition~\ref{prop:A^t kth moment-2} to have for $m \leq 2kt-k$:
\begin{align*}
   E'_{m} &\leq C' \binom{2k}{\rmin} (2k)^{f(\rmin) - m} 3^{2k-\rmin} C^{f(\rmin) - m}\\
   &\quad \times U_b^{s(m)} n^{m - ts(m) - (\rmin - s(m))} 
   \rho_n^{m- t s(m)}\\
   &\quad \times (2m)^{(2k-\rmin) \cdot (t - 2 h) + 2 (f(\rmin) - m)}.
\end{align*}
Let $m_0$ be such that $m_0 = f(\rmin)$ and define $E'_{m_0}$ with the same expression as above. 
Then bounding as in the proof of Proposition~\ref{prop:A^t kth moment-2} we have from \eqref{eqn:lower tail bound 2}
\begin{align*}
\E (Y_b - \E 	Y_b)^{2k}  \leq C E'_{m_0}.
\end{align*}
Thus, \eqref{eqn:lower tail bound 1} together with the fact that $\frac{U_b}{L_b} = 1 + O(\frac{k}{n})$  yields
\begin{align*}
\Prob\left(|Y_b - \E Y_b| > \delta \E Y_b\right) 
&\leq \frac{\binom{2k}{\rmin} 3^{2k-\rmin}  U_b^{\rmin} (n \rho_n)^{m -t \rmin} (2m)^{(2k-\rmin) \cdot (t - 2 h)}}
{\delta^{2k} L_b^{2k}}\\
&\leq \left(\frac{U_b}{L_b}\right)^{\rmin} 
\cdot \left(
\frac{Ck^{1 + t -2h} (n \rho_n)^h \delta^{-\frac{2k}{2k - \rmin}} }{n^{t-1}\rho_n^t}
\right)^{2k-\rmin},\\
&\leq C \left(
\frac{C(\log n)^{1 + t -2h -2\eta(t-h)} (n \rho_n)^h }{n^{t-1}\rho_n^t}
\right)^{\ceil{\log n} (t-h)^{-1}} \leq n^{-c},
\end{align*}
for any $c > 0$ when \eqref{eq:t-condn} holds.
This completes the proof of Proposition~\ref{prop:Lower tail A^t}. 
			
\end{proof}

\section{Analysis of spectral clustering for DeepWalk}\label{sec:deepwalk}
In this section, we first prove properties of the $M_0$ matrix in Section~\ref{sec:M-0}. In Section~\ref{sec:frob-deepwalk} we bound $\frob{M-M_0}$ and complete the proof of Theorem~\ref{prop:Frob norm tL to tU}. Finally, we bound the number of misclassified nodes in Section~\ref{sec:miss-deepwalk} and complete the proof of Theorem~\ref{thm:DeepWalk misclassified nodes}.

\subsection{Analysis of noiseless \texorpdfstring{$M$}{M}-matrix}\label{sec:M-0}
Recall the definition of the matrix $M_0$ from Section~\ref{sec:spectral-description}. Also recall that $\true$ is the matrix of true community assignments. Let $(\lambda_i,v_i)$, $1 \leq i \leq K$ be the $K$ largest eigenvalue-eigenvector pairs of $M_0$, and let $V_0 =(v_1,\dots,v_K)$. 
We will prove the following collection of claims for $M_0$ and its eigenspace: 
\begin{prop}\label{prop:null-matrix}
\begin{enumerate}[(a)]
    \item \label{prop:null-matrix-1} There exists a full-rank 
    matrix $Z_0 \in \R^{K\times K}$ such that $M_0 = 	\true Z_0 \true^T $. 
     Moreover, $(M_0)_{ij} = \Theta(1)$ uniformly in $i,j\in [n]$.
     
     \item \label{prop:null-matrix-2} We have $\lambda_i = \Theta(n)$ and there exists a full rank matrix $X_0\in \R^{K\times K}$ such that $V_0 = \true X_0 $.
    
    \item \label{prop:null-matrix-3} If $i$ and $j$ are two nodes such that $g(i) \neq g(j)$, then we have
	\begin{align} \label{eq:dist-V-rows}
	    \frob{(V_0)_{i \star} - (V_0)_{j \star}} = \sqrt{\frac{1}{n_{g(i)}} + \frac{1}{n_{g(j)}}}.
	\end{align}  
\end{enumerate}
\end{prop}
\begin{proof}
\emph{Part~\ref{prop:null-matrix-1}.}
Let $P^{\sss (t)}$ be the transition matrix for a simple random walk on a complete graph with edge-weights $P$. 
Also, recall from \eqref{eqn:deepwalk init dist} that the random walks are initialized with distribution  $(|P_{i\star}|/|P|)_{i\in [n]}$. 
By Proposition~\ref{lem:M for node2vec},  we have
\begin{align*}
(M_0)_{ij} = \log \bigg(	\frac{\sum_{t = t_L}^{t_U} (l-t) \cdot (P^{\sss (t)}_{ij} + P^{\sss (t)}_{ji}) }{2b \gamma(l, t_L, t_U) \frac{|P_{i\star}|}{|P|}\times  \frac{|P_{j\star}|}{|P|} }	\bigg).
\end{align*}
So in order to show that $M_0 = \true Z_0 \true^T$, it is enough to show that $(P^{\sss (t)}+(P^{\sss (t)})^T) = \true Z \true^T$ for some matrix $Z$.
Towards this we have
\begin{align}
	(P^{\sss (t)})_{ij} = \frac{1}{\vert P \vert}  \sum_{(i_0, i_1, \ldots, i_t) : i_0 = i, i_t = j} P_{i_0 i_1} P_{i_1 i_2} \cdots P_{i_{t-1} i_t}  \bigg(	\prod_{l=1}^{t-1} |P_{i_l\star}|^{-1} \bigg). \label{eqn:deepwalk P(i,j,t) expansion}
\end{align}
Similarly, we can compute $((P^{\sss (t)})^T)_{ij}$, and it is clear from these expressions that $(P^{\sss (t)}+(P^{\sss (t)})^T)_{ij}$ only depends on $i$ and $j$ through their block labels $g(i)$ and $g(j)$. This shows that $(P^{\sss (t)}+(P^{\sss (t)})^T) = \true Z \true^T$ for some appropriate matrix $Z$, and hence $M_0 = \true Z_0 \true^T$. 
We now show that $Z_0$ has full rank. Since $B$ has rank $K$, $P^{\sss (t)}$ has rank $K$ and has the same block structure as $B$.
We note that $\diag{|P_{1\star}|^{-1},|P_{1\star}|^{-1}, \ldots, |P_{n\star}|^{-1}}$ is an invertible matrix.
This implies that $\bigg(	\frac{\sum_{t = t_L}^{t_U} (l-t) \cdot (P^{\sss (t)}_{ij} + P^{\sss (t)}_{ji}) }{2b \gamma(l, t_L, t_U) \frac{|P_{i\star}|}{|P|}\times  \frac{|P_{j\star}|}{|P|} }	\bigg)$ has rank $K$ and has the same block structure as $B$.
Then by Lemma~\ref{lem:rank M is K}, $Z_0$ has rank $K$ a.s.

Next, we estimate the order of the coefficients $(M_0)_{ij}$.
Note that 
\begin{align}
	\sum_{(i_0, i_1, \ldots, i_t) | i_0 = i, i_t = j} P_{i_0 i_1} P_{i_1 i_2} \cdots P_{i_{t-1} i_t} = \Theta(n^{t-1}\rho_n). \label{eqn:deepwalk P(i,j,t) expansion 2}
\end{align}	
Indeed, the leading contribution is due to paths with $t$ distinct edges and $t-1$ choices of intermediate vertices. 
If we have less distinct vertices among $(i_1,\dots,i_{t-1})$, then the number of choices is at most $O(n^{t-2})$. This proves \eqref{eqn:deepwalk P(i,j,t) expansion 2}. 
Also, $\vert P \vert = \Theta(n^2 \rho_n)$ and $|P_{i_l\star}| = \Theta(n \rho_n)$.
Combining these orders implies that $(M_0)_{ij} = \Theta(1)$. \\

\noindent \emph{Part~\ref{prop:null-matrix-2}.} Note that
\begin{align*}
	\lambda_i v_i  = M_0 v_i =   \true  Z_0 \true^T v_i 
	\quad \implies \quad  v_i = \true \left( \lambda_i^{-1}  Z_0 \true^T v_i \right).
\end{align*}
Taking $X_0 = Z_0 \true^T V_0 \Lambda^{-1}$ with $\Lambda = \diag{\lambda_1,\dots,\lambda_K}$, we get $V_0 = \true X_0$. 
Since $Z_0$ is full rank, $\rank(Z_0\true^T) = K$. Since $\rank(V_0 \Lambda^{-1}) = K$, we have $X_0$ is full rank. 
Next we establish the order of the eigenvalues. For this, let $D = \diag{\sqrt{n_1}, \sqrt{n_2}, \ldots, \sqrt{n_K}}$.  We write
\begin{align*}
	M_0 = \true D^{-1} \left(D Z_0 D \right) D^{-1} \true^T.
\end{align*}
This shows that $M_0$ and $Z_0$ have the same eigenvalues. 
Note that the entries of $(DZ_0D)_{ij} = \sqrt{n_in_j} (Z_0)_{ij}$ and $(Z_{0})_{ij} = O(1)$ as $(M_0)_{ij} = \Theta(1)$. Hence its non-zero eigenvalues of $DZ_0D$, and hence the non-zero eigenvalues of $M_0$ will be of order $\Theta(n)$. 
\\

\noindent \emph{Part (c).} 
We start by noting that $V_0 = \true X_0$,  $V_0$ has $K$ distinct rows, as
$X_0$ has $K$ distinct rows. Thus, $(V_0)_{i\star} = (V_0)_{j\star} $ whenever $g(i) = g(j)$ and $(V_0)_{i\star} \neq (V_0)_{j\star} $ whenever $g(i) \neq g(j)$. 
We now compute the inner products of rows of $X_0$.
For this, we first note that
\begin{align*}
\langle (V_0)_{\star i}, (V_0)_{\star j} \rangle &= \langle (\Theta_0 X_0)_{\star i}, (\Theta_0 X_0)_{\star j} \rangle,\\
&= \sum_r {n_r} ((X_0)_{ri} \cdot (X_0)_{si}),\\
&= \langle D (X_0)_{\star i}, D (X_0)_{\star j} \rangle.
\end{align*}
This shows that $D X_0$ has orthogonal columns and as a consequence, orthogonal rows.
This shows that 
\begin{align*}
  \langle (X_0)_{r \star}, (X_0)_{s \star} \rangle  = 0 \text{ if }r \neq s \quad \text{ and } 
  \quad \langle (X_0)_{r\star} , (X_0)_{s\star} \rangle = \frac{1}{n_r} \text{ if }r = s.
\end{align*}
Thus, we have 
$$\langle (V_0)_{i\star} , (V_0)_{j\star} \rangle = 0 \text{ if }g(i) \neq g(j) \quad \text{ and } \quad \langle (V_0)_{i\star} , (V_0)_{j\star} \rangle = \frac{1}{n_i} \text{ if }g(i) = g(j).  $$
Therefore, \eqref{eq:dist-V-rows} follows immediately. 
\end{proof}
We finish this section with the following perturbation result about the eigenspace of $M$, which is a direct consequence of Davis-Kahan $\sin \theta$ theorem \cite[Theorem 2]{yu2015useful}:

\begin{prop}\label{lem:eigenvector bound}
Let $V$ be the matrix of $K$ largest eigenvectors of $M$. 
There exists an orthonormal matrix $O \in \bbR^{K \times K}$ such that
\[  \frob{V  - V_0 O} \leq \frac{\sqrt{8K} \frob{M - M_0}}{\min_{1\leq r \leq K} \vert \lambda_r \vert}    .\]
\end{prop}

\subsection{Bound on \texorpdfstring{$\frob{M-M_0}$}{Frobenius norm {M-M0}}.}\label{sec:frob-deepwalk}
In this section, we prove Theorem~\ref{prop:Frob norm tL to tU}. We start by proving it for a simple case $t= \tL = \tU \geq 2.$
\begin{prop}\label{prop:Frob norm t > 2}
The conclusion for Theorem~\ref{prop:Frob norm tL to tU} holds with $t= \tL = \tU \geq 2.$
\end{prop}
\begin{proof}
Let us first give a proof for $t= \tL = \tU \geq 3$ using the estimates from Section~\ref{sec:path counting deepwalk}. The proof for $t= \tL = \tU = 2 $ is similar and we will give the required modifications at the end. 
Throughout, the constants term $C, C'$ may change from line to line in this proof.
Let $a_n = 4 n (\log n)^{-\eta}$.
Recall the notation $\WA$ and $\WP$ from Section~\ref{sec:SBM}.
Also, recall from Proposition~\ref{lem:M for node2vec} that 
\begin{align}\label{eq:M-0-M-defns-deepwalk}
	M_{ij} = \log \bigg[	\frac{2 \vert A \vert}{b \gamma(l, t_L, t_U)} (l-t) \cdot (\DA^{-1} \WA^t)_{ij}	\bigg] \1_{A^{(t)}_{ij} > 0}, \quad (M_0)_{ij} = \log \bigg[	\frac{2 \vert P \vert}{b \gamma(l, t_L, t_U)} (l-t) \cdot (\DP^{-1} \WP^t)_{ij}	\bigg].
	\end{align}
By Proposition \ref{prop:Lower tail A^t} with $k = \ceil{\log n}$ we have for any $1 \leq i,j, \leq n$
	\begin{align*}
		\Prob\left(	A^{(t)}_{ij} = 0 \right) \leq O(n^{-3}),
	\end{align*}
And this implies that 
	\begin{align*}
		&\Prob\left(	A^{(t)}_{ij}	= 0 \text{ for  some } 1 \leq i,j \leq n	\right) = o(1),\\
		\implies &\Prob\left(	M_{ij}	= 0 \text{ for  some } 1 \leq i,j \leq n	\right) = o(1).\numberthis\label{bound-1-zero-terms}
	\end{align*}
Then we have
\begin{align}\label{eqn: t > 2 overall bound 1}
	\Prob\left(	\frob{M-M_0} \geq a_n	\right) 
	&	\leq o(1) + \Prob\bigg(	\sum_{(i,j)} ( M - M_0)^2_{ij} \geq {a^2_n}		\bigg). 
\end{align}
Next we bound the second term in \eqref{eqn: t > 2 overall bound 1}.
By Chernoff bound, we have that for a sufficiently large constant $C_0>0$, 
\begin{eq} \label{degree-concentration}
\Prob\bigg( \bigg|\frac{\vert A \vert}{\vert P \vert} - 1\bigg| \geq C_0 n^{-1/2} \bigg) 
	&\leq 2\exp \{	- C' n^2 \rho_n\times C_0^2/n  \} =o( n^{-4}),\\
\PR\bigg(\exists i\in [n]: \bigg|\frac{|A_{i\star}|}{|P_{i\star}|} - 1\bigg| >C_0\sqrt{\frac{\log n}{n \rho_n}} \bigg) &\leq 2n\exp(- C' n\rho_n\times C_0^2\log n /n\rho_n)\leq  o(n^{-4}). 
\end{eq}
Using \eqref{degree-concentration}, we simplify the second term in \eqref{eqn: t > 2 overall bound 1} as
\begin{eq} \label{eqn: t > 2 overall bound 2}
	&\Prob\bigg(	\sum_{(i,j)} (M-M_0)^2_{ij} \geq a_n^2 	\bigg) 
	\leq \sum_{(i,j)} \Prob\bigg((M-M_0)^2_{ij} \geq  \frac{a^2_n}{n^2}  \bigg)\\
	&\leq  \sum_{(i,j)}\Prob\bigg( \max \bigg\{ \frac{(\DA^{-1} \WA^t)_{ij}}{(\DP^{-1} \WP^t)_{ij}}, \frac{(\DP^{-1} \WP^t)_{ij}}{(\DA^{-1} \WA^t)_{ij}}\bigg\} \geq (1+O(n^{-1/2}))\exp\Big(\frac{a_n}{n}\Big)  \bigg) + o(n^{-4}).
\end{eq}
To analyze this, recall  from \eqref{eqn:P_b def} $\cP_b$ is the set of paths with vertices having community assignment $b$ for $b\in \cB_{i,j}$. For $p = (i_0,\dots,i_t)\in \cP_b$, let 
\begin{gather*}
    \bar{X}_p = \frac{1}{|A_{i_0\star}|}\prod_{l=1}^t \frac{A_{i_{l-1}i_l}}{|A_{i_l\star}|},\quad \text{and} \quad  \bar{Y}_b = \sum_{p\in \cP_b} \bar{X}_p,\\ \bar{X}_{p}^* = \frac{1}{|P_{i_0\star}|}\prod_{l=1}^t \frac{P_{i_{l-1}i_l}}{|P_{i_l\star}|}, \quad \text{and} \quad \bar{Y}_{b}^* = \sum_{p\in \cP_b} \bar{X}_{p}^*.
\end{gather*}
Then we have $(\DA^{-1} \WA^t)_{ij} = \sum_{b \in \cB_{i,j}} \bar{Y}_b$ and $(\DP^{-1} \WP^t)_{ij} = \sum_{b \in \cB_{i,j}} \bar{Y}_{b}^*$. 
Now, for $b \in \cB_{i,j}$, $\Prob(\bar{Y}_b = 0) = o(n^{-4})$ by Proposition \ref{prop:Lower tail A^t}, 
and on the set $\{\bar{Y}_b \neq 0\}$, we have 
\begin{align*}
	\frac{(\DP^{-1} \WP^t)_{ij}}{(\DA^{-1} \WA^t)_{ij}} \leq \sum_{b \in \cB_{i,j}} \frac{\bar{Y}_b^*}{\bar{Y}_b},
	\qquad &\frac{(\DA^{-1} \WA^t)_{ij}}{(\DP^{-1} \WP^t)_{ij}} \leq \sum_{b \in \cB_{i,j}}\frac{\bar{Y}_b}{\bar{Y}_b^*}.
\end{align*}
Thus, in order to bound \ref{eqn: t > 2 overall bound 2}, it is enough to bound the probabilities for $\bar{Y}_b^*/\bar{Y}_b$ or $\bar{Y}_b/\bar{Y}_b^*$ being large. 
Since the row sums of $A$ are concentrated by \eqref{degree-concentration}, we will bound $Y_b^*/Y_b$, $Y_b/Y_b^*$ instead, where $Y_b^*$ is as defined below: 
\begin{gather*}
    {X}_{p}^* = \prod_{l=1}^t {P_{i_{l-1}i_l}}, \quad \text{and} \quad {Y}_{b}^* = \sum_{p\in \cP_b} {X}_{p}^*.
\end{gather*}
Fix $(i,j)$. 
We estimate the difference 
\begin{align}\label{eqn:deepwalk E X'_b - E'_b}
	\left| \E Y_b - Y_b^* \right| &= \bigg| \E \sum_{p \in \cP_b} (A_{i_0 i_1} \cdots A_{i_{t-1} i_t} -  P_{i_0 i_1} \cdots P_{i_{t-1} i_t} )\bigg|.
\end{align}
The summands in the equation above are equal to zero if the associated path $(i_0, i_1, \ldots, i_t)$ has $t$ distinct edges and there are no self-loops. 
Consider the first set of summands in \eqref{eqn:deepwalk E X'_b - E'_b}. By Proposition \ref{prop:A^t kth moment-2} the sum over summands with less than $t$ distinct edges is $O(1/n\rho_n) \E[Y_b]$.
We now give an upper bound on the summands $P_{i_0 i_1} \cdots P_{i_{t-1} i_t}$ as follows.
Suppose a path has less than $t$ distinct edges. If the path is a Type I path then by Lemma~\ref{lem:choice-vertices} the number of choices of distinct vertices along the path is less than $t-1$.
If the the path is a Type II path, then again by Lemma~\ref{lem:choice-vertices} the number of choices of distinct vertices along the path is at most $\floor{t/2}$. 
Finally, if there are self-loops then the number of choices of vertices is less than $t-1$.
This implies that the upper bound on the second set of summands is $O(1/n) \E[Y_b]$.
Thus in summary we have
\begin{align*}
	\left| \E Y_b - Y_b^* \right| = O\bigg(\frac{1}{n\rho_n}\bigg) \E Y_b .
	\numberthis \label{eqn:deepwalk bd EX'b E'b 1}
\end{align*}
These computations show that 
\begin{gather*}
\frac{Y_b}{Y_b^*}= \frac{Y_b}{\E Y_b \left( 1 + O((n \rho_n)^{-1}) \right)}, 
\quad \text{and} \quad
\frac{Y_b^*}{Y_b}= \frac{\E Y_b \left( 1 + O((n \rho_n)^{-1}) \right)}{Y_b}.
\numberthis \label{eqn: Y_b Y_b* bound deepwalk}
\end{gather*} 
Recall that  $a_n = 4 n (\log n)^{-\eta}$. To compute \eqref{eqn: t > 2 overall bound 2}, we now use Proposition \ref{prop:Lower tail A^t}, \eqref{degree-concentration} and \eqref{eqn: Y_b Y_b* bound deepwalk} to obtain 
\begin{align*}
	\Prob\bigg( \frac{(\DP^{-1} \WP^t)_{ij}}{(\DA^{-1} \WA^t)_{ij}} \geq (1+O(n^{-1/2}))\exp \Big(  \frac{a_n}{n}  \Big)  \bigg)
	&\leq \sum_{b \in \cB_{i,j}}\Prob\bigg(	\frac{\bar{Y}_b^*}{\bar{Y}_b} \geq \exp\{C (\log n)^{-\eta}\}\bigg) + o(n^{-3})\\
	&\leq \sum_{b \in \cB_{i,j}}\Prob\bigg(	\frac{Y_b^*}{Y_b} \geq 1 + (\log n)^{-\eta}\bigg) + o(n^{-3})
	= o(n^{-3}).
\end{align*}
A similar bound can be computed with  $\frac{(\DA^{-1} \WA^t)_{ij}}{(\DP^{-1} \WP^t)_{ij}}$ as well repeating the same computations. 
Hence we have established that the term in \eqref{eqn: t > 2 overall bound 2} is at most $o(n^{-3})$, and thus combining \eqref{eqn: t > 2 overall bound 1} and \eqref{bound-1-zero-terms}, we conclude that $\frob{M-M_0} = \OP(a_n)$, and thus \eqref{frob-norm-bound-original-1} follows for $\tL = \tU = t\geq 3$. 
For $\tU = t = 2$, the argument is exactly similar except that we use \cite[Theorem 2.8]{JLR00} for showing \eqref{bound-1-zero-terms}, 
we use Bernstein's inequality in place of the concentration inequality in Proposition \ref{prop:Lower tail A^t}, 
we don't need \eqref{eqn:deepwalk bd EX'b E'b 1} for $i \neq j$ case, and for the $i=j$ case we  use the bound $\frac{\E Y_b}{Y_b^{*}} = O(\rho_n^{-1})$.

Next, we prove \eqref{frob-norm-bound-original-2} in the case $\tL=\tU = t\geq 2$. Let $n^{t-1}\rho_n^t \ll 1$.
We will show that, for any $\varepsilon>0$,
\begin{align*}
	\Prob\big(\frob{M - M_0} \geq C_\varepsilon n^2 \big) \geq 
	 \Prob\bigg(	\sum_{(i,j)}	(M_0)^2_{ij} \ind\{ M_{ij} = 0\} \geq  C_\varepsilon n^2	\bigg) \geq 1 - \varepsilon,
\end{align*}
for some constant $C_\varepsilon>0$ depending on $\varepsilon$ and the last inequality holds for large enough $n$.

By Proposition~\ref{prop:null-matrix}, the entries of $(M_0)_{ij}$'s are constant over all $i,j$ pairs such that  $g(i) = r, g(j) = s$. 
Also, $M_0$ will have some non-zero entries since $\rank(M_0) = K$.
Let $r$ and $s$ be such that $g(i) =r$ and $g(j) =s$ and $\vert (M_0)_{ij} \vert = C_1 >0$ for all $i,j$ such that  $g(i) =r$ and $g(j) =s$, and the number of such pairs of $i,j$ is at least $C_2 n^2$ for some $0 < C_2 < 1$. 
Let 
\begin{align*}
	S_{r,s} := \{ (i,j): A^{t}_{ij} = 0, g_i =r, g_j =s		\}, \qquad 
	T_{r,s} := \{ (i,j): A^{t}_{ij} > 0, g_i =r, g_j =s		\}.
\end{align*}
Then
\begin{align*}
	\Prob\bigg(	\sum_{(i,j)}	(M_0)^2_{ij} \ind\{ M_{ij} = 0\}  \geq  C_\varepsilon n^2	\bigg)
	&\geq \Prob \bigg(\sum_{(i,j) \in S_{r,s}} C_1^2 \geq C_\varepsilon n^2 \bigg)= \Prob\bigg(\vert S_{r,s} \vert \geq \frac{C_\varepsilon n^2}{C_1^2}	\bigg).
\end{align*}
Next let $i$ and $j$ be any two nodes such that $g_i = r$, $g_j = s$, and $i \neq j$. Then by Proposition~\ref{prop:A^t kth moment}
we have
\begin{align*}
	\Prob\big(	A^{t}_{ij} > 0\big) \leq \sum_{b \in \cB_{i,j}} \Prob\left(Y_b >0\right) \leq C_3 n^{t-1} \rho_n^t,
\end{align*}
for some constant $C_3 > 0$. This along with Markov inequality implies that
\begin{align*}
	\Prob\left(	\vert T_{r,s} \vert \geq C_3  \varepsilon^{-1} n^{t+1} \rho_n^t	\right) 
	&\leq \frac{n^2 \cdot C_3 n^{t-1} \rho_n^t}{C_3  \varepsilon^{-1} n^{t+1} \rho_n^t} 
	\leq \varepsilon.
\end{align*}
This shows that for large enough $n$ we have
\begin{align*}
	\Prob \left(    \vert S_{r,s} \vert	\geq \vert S_{r,s} \vert +  \vert T_{r,s} \vert - C_3  \varepsilon^{-1} n^{t+1} \rho_n^t  \right) \geq 1 - \varepsilon 
	\implies \Prob \left(    \vert S_{r,s} \vert	\geq C_2 n^2 - C_3  \varepsilon^{-1} n^{t+1} \rho_n^t  \right) \geq 1 - \varepsilon.
\end{align*}
Taking $C_\varepsilon = C_1^2 C_2/2$ and noting that $n^{t+1} \rho_n^t = o(n^2)$ completes the proof.
\end{proof}	
Next we complete the proof of Theorem~\ref{prop:Frob norm tL to tU} for general $\tL,\tU$.
\begin{proof}[Proof of Theorem~\ref{prop:Frob norm tL to tU}] The general idea is to reduce the computations to the analogous computations for the $\tL = \tU$ case.  
If $t_L =2$ and $i \neq j$ then by \cite[Theorem 2.8]{JLR00}, 
			\[ \Prob\left(A^{(2)}_{ij} = 0\right) \leq \exp\{-\theta(n\rho_n^2)	\}.\]
			Similarly if $i = j$ (and $t_L =2$) we have 
			\[ \Prob\left(A^{(2)}_{ii} = 0\right) \leq \exp\{-\theta(n\rho_n)	\}.\]
			By the assumption of $n^{t_L-1}\rho_n^{t_L} \gg (\log n)$ when $t_L =2$, we have
			\begin{align*}
				\Prob\left(	A^{(2)}_{ij}	= 0 \text{ for  some } 1 \leq i,j \leq n	\right) = o(1).
			\end{align*}
			Next let $\max(3, t_L) \leq  t \leq t_U$. Then by Proposition \ref{prop:Lower tail A^t} with $k = \ceil{\log n}$ we have for any $1 \leq i,j, \leq n$
			\begin{align*}
				\Prob\left(	A^{(t)}_{ij} = 0 \right) \leq O(n^{-3}).
			\end{align*}
			And this implies that 
			\begin{align*}
				\Prob\left(	A^{(t)}_{ij}	= 0 \text{ for  some } 1 \leq i,j \leq n \text{ and for some } t \in \{t_L, t_L+1, \ldots, t_U\}	\right) = o(1).
			\end{align*}
Then proceeding as in proof of Proposition \ref{prop:Frob norm t > 2} we have 
\begin{align*}
	&\Prob\left(\frob{M_0 - M} \geq a_n\right)\\
	&\leq \Prob\left(\frob{M_0 - M} \geq a_n, A^{t}_{ij} > 0 \text{ for }1\leq i,j \leq n, t_L \leq t \leq t_U \right) + o(1),\\
	&\leq o(1) + \sum_{1 \leq i \neq j \leq n} \Prob\left(\frac{\sum_{t=t_L}^{t_U}(l-t) \left(	\DP^{-1} \WP^t\right)_{ij}}
	{\sum_{t=t_L}^{t_U}(l-t) \left(	\DA^{-1} \hat{\WP}^t\right)_{ij}}	 
	\geq \exp\left\{	\frac{a_n}{\sqrt{2}n} - \theta(n^{-1/2})	\right\} \right)\\
	&\quad + \sum_{1 \leq i \leq n} \Prob\left(\frac{\sum_{t=t_L}^{t_U}(l-t) \left(	\DP^{-1} \WP^t\right)_{ii}}
	{\sum_{t=t_L}^{t_U}(l-t) \left(	\DA^{-1} \hat{\WP}^t\right)_{ii}}	 
	\geq \exp\left\{	\frac{a_n}{\sqrt{2n}} - \theta(n^{-1/2})	\right\}\right)\\
	&\quad + \sum_{1 \leq i \neq j \leq n} \Prob\left(\frac{\sum_{t=t_L}^{t_U}(l-t) \left(	\DA^{-1} \hat{\WP}^t\right)_{ij}}
	{\sum_{t=t_L}^{t_U}(l-t) \left(	\DP^{-1} \WP^t\right)_{ij}}	 
	\geq \exp\left\{	\frac{a_n}{\sqrt{2}n} - \theta(n^{-1/2})	\right\} \right)\\
	&\quad + \sum_{1 \leq i \leq n} \Prob\left(\frac{\sum_{t=t_L}^{t_U}(l-t) \left(	\DA^{-1} \hat{\WP}^t\right)_{ii}}
	{\sum_{t=t_L}^{t_U}(l-t) \left(	\DP^{-1} \WP^t\right)_{ii}}	 
	\geq \exp\left\{	\frac{a_n}{\sqrt{2n}} - \theta(n^{-1/2})	\right\}\right).
	\numberthis \label{eqn:Frob norm t_L to t_U overall bound}
\end{align*}
			We show how to bound the second term in the equation \ref{eqn:Frob norm t_L to t_U overall bound}. For this we see that
			\begin{align*}
				&\Prob\left(\frac{\sum_{t=t_L}^{t_U}(l-t) \left(	\DP^{-1} \WP^t\right)_{ij}}
				{\sum_{t=t_L}^{t_U}(l-t) \left(	\DA^{-1} \hat{\WP}^t\right)_{ij}}	 
				\geq \exp\left\{	\frac{a_n}{\sqrt{2}n} - \theta(n^{-1/2})	\right\} \right)\\ 
				&\leq \sum_{t = t_L}^{t_U} \Prob\left(\frac{ \left(	\DP^{-1} \WP^t\right)_{ij}}
				{ \left(	\DA^{-1} \hat{\WP}^t\right)_{ij}}	 
				\geq \exp\left\{	\frac{a_n}{\sqrt{2}n} - \theta(n^{-1/2})	\right\} \right).
			\end{align*}
			Then each of the probabilities for fixed $t$ can be bounded as in the proof of Propositions \ref{prop:Frob norm t > 2}. Analogously the rest of the terms in \eqref{eqn:Frob norm t_L to t_U overall bound} can be bounded. This completes the proof.
\end{proof}

\subsection{Bounding the number of missclassified nodes}\label{sec:miss-deepwalk}

\begin{proof}[Proof of Theorem \ref{thm:DeepWalk misclassified nodes}]
This proof uses standard arguments to bound the proportion of misclassified nodes such as given in \cite{lei2015consistency}.
For this proof, we choose $O \in \bbR^{K \times K}$ obtained by an application of Proposition~\ref{lem:eigenvector bound} which satisfies
\begin{align}
	\frob{V - V_0 O} \leq \frac{\frob{M - M_0}}{Cn} = \SOP(1), \label{eqn:thm1 Davis Kahan}
\end{align}
where the last step follows using Theorem~\ref{prop:Frob norm tL to tU}.
To simplify notation, we denote $U =V_0O$ in the rest of the proof.
Let $n_{\max} = \max_{r \in [K]} n_r$. 
Recall $(\bar{\Theta},\bar{X})$ from \eqref{eqn:kmeans2} and let $\bar{V} = \bar{\Theta}\bar{X}$.
Then let 
\begin{align*}
	S = \bigg\{	i \in [n] : \frob{\bar{V}_{i \star}  - U_{i \star}} \geq \frac{1}{5} \sqrt{\frac{2}{n_{\max}}}	\bigg\}.
\end{align*}
We will show that for $i \notin S$ the community is predicted correctly using $\bar{\Theta}$.
The proof of the theorem is in two steps.
\paragraph*{Step 1: Bounding $\vert S \vert$.}
By the definition of $S$ we have
\begin{align*}
	\sum_{i \in S} \sqrt{\frac{2}{n_{\max}}} \leq \sum_{i \in S} 5 \frob{\bar{V}_{i \star}  - U_{i \star}}
	\leq 5 \frob{\bar{V}  - U}
	\implies  \vert S \vert \leq  \frac{5}{\sqrt{2}} \sqrt{n_{\max}}  \frob{\bar{V} - U}.
	\numberthis \label{eqn:|S| bound 1}
\end{align*}
Next, we recall the optimization problem in \eqref{eqn:kmeans2} is given as follows.
\begin{align*}
	\frob{ \bar{\Theta} \bar{X} - V }^2 \leq (1+\varepsilon)  \min_{\Theta \in \{0,1\}^{n\times K}, X \in \bbR^{K \times K}} \frob{\Theta X - V }^2. 
\end{align*}
We substitute $U$ for $\Theta X$ to get the following upper bound:
\begin{align}
	\frob{\bar{V} - V}^2 \leq (1 + \varepsilon) \frob{U- V}^2. \label{eqn:Vbar - V bound}
\end{align}
Then by \eqref{eqn:thm1 Davis Kahan} and \eqref{eqn:Vbar - V bound} we have 
\begin{align*}
	\frob{\bar{V} - U} &\leq \frob{\bar{V} - V} + \frob{V- U}\leq (1 + \sqrt{1+\varepsilon}) \frob{V- U} = \SOP(1).
\end{align*}
This combined with equation \ref{eqn:|S| bound 1} we have $\vert S \vert = \SOP(\sqrt{n})$.
			
\paragraph*{Step 2:  Bounding the prediction error.}
For any community $r$, there exists $i_r \in [n]$ such that $g_{i_r} =r$ and $i_r \notin S$ as $n_r = \theta(n)$ and $\vert S \vert = \SOP(\sqrt{n})$.
By Proposition \ref{prop:null-matrix}\ref{prop:null-matrix-3}, for $r\neq s$ we have
\begin{align*}
	\frob{\bar{V}_{i_r \star} - \bar{V}_{i_s \star}} &\geq \frob{ U_{i_r \star} - U_{i_s \star} } - \frob{\bar{V}_{i_r \star} - U_{i_r \star}}	- \frob{\bar{V}_{i_s \star} - U_{i_s \star}}\\
	&\geq \sqrt{\frac{1}{n_{r}} + \frac{1}{n_{s}}} - \frac{2}{5}  \sqrt{\frac{2}{n_{\max}}}\geq \frac{3}{5} \sqrt{\frac{2}{n_{\max}}}. \numberthis\label{eqn:step 2 bd 1}
\end{align*}
Next, let $i$ be such that $g_i = r$ and $i \notin S$. We show that $\bar{V}_{i \star} = \bar{V}_{i_r \star}$. 
For this we note that $U_{i \star} = U_{i_r \star}$ and
\begin{align*}
	\frob{\bar{V}_{i \star} - \bar{V}_{i_r \star}} &\leq \frob{\bar{V}_{i \star} - U_{i \star}} +
	\frob{U_{i_r \star} - \bar{V}_{i_r \star}}\\
	&< \frac{1}{5} \sqrt{\frac{2}{n_{\max}}} + \frac{1}{5} \sqrt{\frac{2}{n_{\max}}}  < \frac{2}{5} \sqrt{\frac{2}{n_{\max}}}.  \numberthis\label{eqn:step 2 bd 2}
\end{align*}
In view of \eqref{eqn:step 2 bd 1} and \eqref{eqn:step 2 bd 2} we must have $\bar{V}_{i \star} = \bar{V}_{i_r \star}$ as each node is assigned exactly one community by the $(1+\varepsilon)$-approximate $k$-means algorithm and there are exactly distinct $K$ rows in $\bar{V}$ (again by \eqref{eqn:step 2 bd 1}).
Let $C$ be a permutation matrix defined so that $\true C$ assigns community $r$ to node $i_r$, for $1 \leq r \leq K$.
Then we have, 
\begin{align*}
	\sum_i \ind\{	\bar{\Theta}_{i\star} \neq (\true C)_{i \star}	\} \leq \vert S \vert.
\end{align*}
From the bound on $\vert S \vert$ from Step 1, we have that $\err(\bar{\Theta},\true)$ is $\SOP\left(1/\sqrt{n}\right)$ and this completes the proof of Theorem~\ref{thm:DeepWalk misclassified nodes}.
\end{proof}

		\section{Path counting for node2vec}\label{sec:path counting node2vec}
		
		In this section, we focus on computing the asymptotics for the sum of weighted paths having  some specified community assignments for the intermediate vertices. In section~\ref{sec:counting-weighted-paths}, we bound its $k$-th moment and we end with a concentration inequality in section~\ref{sec:concent-node2vec}.  
		
		\subsection{Bounding moments of path counts}\label{sec:counting-weighted-paths}
		
		We compute upper bounds for the $k$-th moment and $k$-th centered moments for weighted paths between two nodes. 
		We first fix some notation.
		Let $\cB_{ij}$ and $\cP_b$ be defined as in \eqref{eqn:path types set def} and \eqref{eqn:P_b def}.

		For any path $p = (i_0, i_1, i_2, \ldots, i_t)  \in \cP_b$, let
		\begin{align*}
			\mathfrak{N}((i_0, i_1, \ldots, i_t)) := \left\{ l | 2 \leq l \leq t, i_{l - 2} = i_l 	\right\},
		\end{align*}
		be the set of locations of backtracks in the path. Let 
		\begin{align*}
		\mathfrak{n} = \mathfrak{n}((i_0, i_1, \ldots, i_t)) = \vert \mathfrak{N}((i_0, i_1, \ldots, i_t)) \vert,
		\end{align*}
		be the number of backtracks in $p$.
		For any path $p = (i_0, i_1, i_2, \ldots, i_t)  \in \cP_b$ and $\alpha > 0$, we associate the random variable 
		\begin{align*}
			X_{p,\alpha} := A_{i_0 i_1} A_{i_1 i_2} \cdots A_{i_{t-1} i_t}
			\alpha^{\mathfrak{n}((i_0, i_1, \ldots, i_t))}, \numberthis\label{eqn:node2vec X-p-def}
		\end{align*}
		and let 
		\begin{align*}
			Y_{b,\alpha} : =  \sum_{p \in \cP_b} X_{p,\alpha}. \numberthis\label{eqn:node2vec Y_b beta =1 }
		\end{align*}
		We note that when $\alpha =1$, $X_{p,1} = X_p$ and $Y_{b,1} = Y_b$ where $X_p$ and $Y_b$ are as defined in \eqref{eqn:X-p-def} and \eqref{eq:Y-b-def} respectively.
		To simplify notation in this section, we will drop the subscript $\alpha$ and simply write $X_p$ and $Y_b$ in place of $X_{p,\alpha}$ and $Y_{b,\alpha}$ respectively.
		Let $U_b$ and $L_b$ be the upper and lower bounds for path type $b \in \cB_{i,j}$ as defined in \eqref{eqn:U-b-def} and \eqref{eqn:L-b-def} respectively.
		Then we have the following bounds on $\E Y_b^k$
		
		\begin{prop}\label{prop:node2vec kth moment}
		Let $\tL = \tU = t \geq 3$ be given and suppose that $\alpha = O\left(\frac{1}{n}\right)$ and \eqref{eq:t-condn-2} holds.
		Then we have 
		\begin{align*}
				L_{b}^k \leq \E Y_b^k \leq U_{b}^k (1 +  o(1) ).
			\end{align*} 
		\end{prop}
		
		Again the idea, as in section \ref{sec:path counting deepwalk}, is to show that the leading term for $\E Y_b^k$ is due to $\E\big(	\prod_{\pi=1}^{k}	X_{p_{\pi}}	\big)$ of $k$ ordered paths $p_{\pi}$ having $kt$ distinct edges between them. 
		The contribution of the rest of the terms are of a smaller order.
		Similar to section \ref{sec:path counting deepwalk}, let
		\begin{equation}\label{E-m-defn node2vec}
		    E_m =  E_{m,\alpha} := \sum_{(p_1, p_2, \ldots, p_k): p_{\pi} \in \cP_b, |\cup_{i\in [k]} e(p_{\pi})| = m } \E (X_{p_1} X_{p_2} \cdots X_{p_k}).
		\end{equation}
		We will show the following: 
		\begin{prop}\label{prop:node2vec kth moment-2}
		Under identical conditions as in Proposition~\ref{prop:node2vec kth moment}, we have $\sum_{m<kt} E_m = o(U_b^k)$.  
		\end{prop}
		\begin{proof}[Proof of Proposition~\ref{prop:node2vec kth moment} using Proposition~\ref{prop:node2vec kth moment-2}]
Note that, we can write 			\begin{align*}
	\E Y_b^k &=  \E \bigg(\sum_{p \in \cP_b} X_p \bigg)^k= \sum_{(p_1, p_2, \ldots, p_k)| p_{\pi} \in \cP_b} \E (X_{p_1} X_{p_2} \cdots X_{p_k}). \numberthis\label{eqn:kth moment and k paths node2vec}
	\end{align*}
	For the upper bound, Proposition~\ref{prop:node2vec kth moment-2} shows that it is enough to bound the summands  corresponding to sequences $(p_1, p_2, \ldots, p_k)$ that satisfy $\vert \cup_{\pi =1}^k e(p_{\pi})	\vert = kt$, i.e. sequences of paths consisting of $kt$ distinct edges. 
	We note that there are no backtracks in this case and so, $ \E (X_{p_1} X_{p_2} \cdots X_{p_k}) = \Prob (X_{p_1} X_{p_2} \cdots X_{p_k} = 1) = \prod_{\pi=1}^{k}  \Prob (X_{p_{\pi}} = 1)$.
	Thus we have the same upper and lower bounds $U_b^k$ and $L_b^k$ as in Proposition \ref{prop:A^t kth moment}.
    
\end{proof}	
The rest of this section is devoted to the proof of Proposition~\ref{prop:node2vec kth moment-2}.
Towards this let $E_{m,r}$ denote the summands in \eqref{E-m-defn node2vec} restricted to the case that there are $r$ segments, Type I or Type II, so that
\begin{align*}
    E_m = \sum_{r= \rmin}^{\rmax} E_{m,r},
\end{align*}
where, given $m$, $[\rmin,\rmax]$ denotes the range of $r$. 
We note that this is in contrast to the proof of Proposition \ref{prop:A^t kth moment-2} where $r$ was the number of Type I segments. 
We also note that we are reusing the notation $\rmin$ and $\rmax$ from the proof of Proposition \ref{prop:A^t kth moment-2} to simplify the notation but the values of $\rmin$ and $\rmax$ will be different in this proof.

The analysis will again consist of two steps. In the first step, we analyze $E_{m,\rmin}$ and in the second step, we will show that $E_{m,r}$ is much smaller than $E_{m,\rmin}$ for $r>\rmin$.  

The following is the intuition for why $E_{m,\rmin}$ is the largest term.
Due to the presence of backtracks, the contribution from Type II segments is of the same or a smaller order than Type I segments. 
By Lemma~\ref{lem:choice-vertices} minimizing Type I segments leads to the maximum number of choices of marked edges.
Combining these two ideas, we see that we must minimize the number of segments.
A formal proof is provided in the rest of the section.

\subsubsection{Computing \texorpdfstring{$E_{m,\rmin}$}{E{m,r*(m)}}.} 
We begin by noting that $\rmin$ is given by
\begin{align*}
\rmin = \ceil*{\frac{m}{t}}.
\end{align*}
To see this we note that each path can have at most $t$ marked edges. 
So we can place $m$ marked edges in a minimum of $\ceil{\frac{m}{t}}$ paths. 
The next lemma counts the number of configurations of segments and unmarked edges.
For this, let $N_{m,\rmin}$ be the number of configurations of $m$ marked edges placed in $r$ segments, Type I or II. 
Again, we are reusing the notation $N_{m,\rmin}$ from section \ref{sec:path counting deepwalk}.

\begin{lem}\label{lem:m-rmin-arrangement node2vec}
\begin{align*}
N_{m,\rmin} \leq C \binom{k}{\rmin} k^{t \rmin - m}.
\end{align*}
\end{lem}
\begin{proof}
The proof is divided into two cases:

\noindent \textbf{Case I:} $\frac{m}{t} \in \bbN$. Note that $\rmin = \frac{m}{t}$ in this case. We can choose $\frac{m}{t}$ paths containing all the $m$ marked edges in $\binom{k}{\frac{m}{t}}$ ways.
Al the edges in the chosen $\frac{m}{t}$ paths are marked edges and all the edges in the rest of the paths are unmarked edges. 

\noindent \textbf{Case II:} $\frac{m}{t} \notin \bbN$. We can first choose $\ceil*{\frac{m}{t}}$ paths to place the marked edges.
By pigeonhole principle, there are $0 < l \leq t \ceil*{\frac{m}{t}} - m$ of the $\ceil*{\frac{m}{t}}$ chosen paths which are not saturated. We can choose arrangements for these $l$ paths in $C^l$ ways where $C$ is a constant that may depend on $t$.
The chosen paths can have at most $l$ unmarked edges.
The rest $k - \ceil*{\frac{m}{t}}$ all have only unmarked edges. 
\end{proof}

We now complete the proof with the following lemma.
\begin{lem}\label{lem:sum over m node2vec}
$\sum_{m<kt}E_{m,\rmin} = o(U_b^k)$. 
\end{lem}
\begin{proof}
We recall that $U_b = \Theta (n^{t-1} \rho_n^t)$. 
The number of choices of segments is given by  Lemma~\ref{lem:m-rmin-arrangement node2vec}.
Let $s(m)$ be the number of maximal (or saturated) Type I paths when placing $m$ marked edges in $\rmin$ paths.
When $\frac{m}{t} \in \bbN$, all the $\rmin$ paths containing marked edges are maximal Type I paths and each of them have probability at most $U_b$ as there are no backtracks.
When $\frac{m}{t} \notin \bbN$ we have $s \geq \max\left\{\ceil*{\frac{m}{t}} -  \left(t\ceil*{\frac{m}{t}} - m\right),0\right\}$ maximal Type I paths.
The number of unmarked edges is at most $kt - m$.
The number of marked edges in the non-maximal Type I paths is equal to $m - s(m)t$.

Choose a Type I segment with $m'$ marked edges in a non-maximal Type I path. By Lemma~\ref{lem:choice-vertices}, the vertices can be chosen in at most $n^{m'-1}$ ways. The number of backtracks in the Type I segment can be equal to $0$ or larger than $0$. So $\alpha^0 = 1$ is an upper bound for the factor coming from backtracks in \eqref{eqn:node2vec X-p-def}. 
And so the upper bound for the contribution coming from the choice of vertices and backtracks is $n^{m'-1} \alpha^0 = n^{m' - 1}$.
Now for a Type II segment (in a non-maximal Type I path) with $m'$ marked vertices, there must be at least ${m'}$ backtracks. Thus the corresponding upper bound for a Type II segment is $n^{m'} \alpha^{m'} = O(1)$. 
We can note that we can have a Type II segment in only a non-saturated path and so the number of Type II segments are $O(1)$.
Using this analysis, for any segment, Type I or II, with $m'$ marked edges in a non-maximal (Type I) path, we upper bound the choices of marked vertices and the factors from backtracks by $C n^{m'-1}$.

Combining all these, we get
\begin{align*}
   E_{m,\rmin} &\leq C \binom{k}{\rmin} k^{t \rmin - m}
     \times U_b^{s} n^{m - s(m)t - (\rmin - s(m))} \rho_n^{m - s(m)t}
    \times m^{kt-m}.
\end{align*}

Using these bounds we see that with the choice of $k = \ceil{\log n}$
\begin{align*}
\sum_{l=0}^{t-1} E_{r_0 t - l, r_{\star}(r_0 t-l)} &= E_{r_0 t, r_{\star}(r_0 t)}
\left( 1 + O\left( \frac{k^2}{n\rho_n} \right) \right), \quad r_0 \geq 1,\\
\sum_{r_0 = 1}^k E_{r_0 t, r_{\star}(r_0 t)}  &= E_{kt,r_{\star}(kt)} \left( 1 + O\left( \frac{k^{t+1}}{n^{t-1}\rho_n^t} \right) \right).
\end{align*}
These bounds in turn imply that
\begin{align*}
   \sum_{m=1}^{kt-1} E_{m, \rmin} =   o\left( U_b^k \right).
\end{align*}

\end{proof}

\subsubsection{Computing \texorpdfstring{$E_{m,r}$}{E{m,r}} for \texorpdfstring{$r > \rmin$}{r > r*(m)}.}
We start with a lemma to bound the number of configurations of segments and unmarked edges as in Lemma \ref{lem:m-rmin-arrangement node2vec}.
\begin{lem}\label{lem:num arrangement node2vec}
Given $m$ marked edges and $r$ segments, let $N_{m,r}$ be the number of configurations of segments and unmarked edges. 
Then, for any $r>\rmin$,
\begin{align*}
    N_{m,r} \leq  N_{m,\rmin} \times O(k^{(r-\rmin) (t+1)}). 
\end{align*}
\end{lem}
\begin{proof}
Let $T_r$ be the set of all configurations of $m$ marked edges and $r$, Type I or Type II, segments. 
Note that we are reusing the notation $T_r$ from proof of Proposition \ref{lem:num arrangement deepwalk} but $T_r$ is defined differently here.
We will inductively bound $T_{r+1}$ in terms $T_r$. 
For that, we consider two cases depending on whether the elements of $T_{r+1}$ has a  path with two  segments or not. 
In both cases, we will find a relation between  $T_r$ to $T_{r+1}$. 

\paragraph*{Case I. } Suppose that $T_{r+1}$ has a path with two  segments. Let us call this subset $T_{r+1}^{\sss I}$.
We consider the elements in $T_r$ which will be related to these elements in $T_{r+1}^{\sss I}$. 
Let $T_r^{\sss I} \subset T_r$ consisting of configuration such that there is at least one path $p$ so that the following condition holds: 
\begin{itemize}
    \item \textbf{Extra unmarked edges.} $p$ has $l$ segments and $l'\geq l$ unmarked edges for some $l,l' \geq 1$. 
\end{itemize}
The first condition is the same as in the proof of Lemma \ref{lem:num arrangement deepwalk}.
The second condition condition is absent as in this construction we will split any segment, Type I or II, as long as it has at least two marked edges. 
Similar to the proof of Lemma \ref{lem:num arrangement deepwalk}, we split a segment only at a marked edge. This is to ensure that splitting creates two segments.
The rest of the details of this construction, and the proof that the relation given by the construction is surjective onto $T_{r+1}^{\sss I}$ are similar to Case I in the proof of Lemma~\ref{lem:num arrangement deepwalk}.
As in Lemma~\ref{lem:num arrangement deepwalk} we have
\begin{align*}
    |T_{r+1}^{\sss I}| \leq O(k) \times |T_{r}^{\sss I}| \leq O(k) \times |T_{r}|.
\end{align*}

\paragraph*{Case II. } 

We next relate arrangements of marked edges where there is at most one  segment per path.
We denote such arrangements as $T_{r+1}^{\sss II} \subset T_{r+1}$. 
Let $T_{r}^{\sss II}$ be the set of all ways of specifying locations of segments such that there are a total of $r$ segments and each of the  segments is placed on distinct paths.
We note that $T_{r}^{\sss II}$ is defined differently as compared to the construction for DeepWalk in Lemma~\ref{lem:num arrangement deepwalk}.
Suppose that $r+1 \leq \min(k, m)$ and $r \geq r_{\star}(m)$.
Then we give a multi-map from $T_{r}^{\sss II}$ onto $T_{r+1}^{\sss II}$ using a construction. 
The condition $r+1 \leq m$ is necessary so that $T_{r+1}^{\sss II}$ is non-empty as we must have at least $r+1$ marked edges in order to have $r+1$ paths each having a Type I segment. 
The condition $r+1 \leq k$ is also necessary for $T_{r+1}^{\sss II}$ to be non-empty as we must have at least $r+1$ distinct paths to place the $r+1$ Type I segments.
The last condition $r \geq r_{\star}(m)$ is to ensure that $T_{r}^{\sss II}$ is non-empty. 
To fix notation, let $S(l)$ be the set of all ways of arranging $l$ marked edges in one path. 
We now describe the construction. 
Let $A \in T_{r}^{\sss II}$ and choose a path $p$ not containing a segment. We can do so as $r+1 \leq k$.
Suppose $p$ has $l \geq 0$ marked edges. Choose $0 \leq u \leq t-l$ paths where the $u$ paths are distinct from $p$, and contain at least one marked edge.
If $l=0$ we require $u > 0$. It is feasible to choose such path(s) as $r+1 \leq m$.
Suppose the $u$ paths are labeled as $q_1, q_2, \ldots, q_{u}$. 
Suppose these paths have $l_{q_1}, l_{q_2}, \cdots, l_{q_{u}}$ marked edges respectively.
Let $v_{q_1}, v_{q_2}, \cdots, v_{q_{u}}$ be such that $0 < v_{q_i} < l_{q_i}$.
We require $\sum_i v_{q_i} \leq t-l$.
Then we modify the arrangements of marked edges in the paths so that the new arrangements for the sequence of paths $(p, q_1, q_2, \ldots, q_{u})$ is any element of
\begin{align*}
S\left(l + \sum_i v_{q_i} \right) \times \left(\prod_{i=1}^{u} S(l_{q_i} - v_{q_i})  \right).
\end{align*}
We keep the arrangements of marked edges in the rest $k - (1 + u)$ paths unchanged. 
This leads to multiple images of $A$ in $T_{r+1}^{\sss II}$.
We note that there are $O(k^{t+1})$ images of $A$ due to the choice of the paths and since the number of ways of choosing the new arrangements for the $1 + u$ paths is $O(1)$ as $t$ is fixed.
We also note that since we modify at most $t+1$ paths in this construction, the number of unmarked edges increases by at most $t(t+1)$.

Now we show that the multi-map given by the construction above is surjective onto $T_{r+1}^{\sss II}$. For this, let $A' \in T_{r+1}^{\sss II}$. Let $p$ be a path containing a Type I segment. 
Suppose $p$ has $l$ marked edges.
Choose $u \leq l$ paths so that $u$ paths contain marked edges and these paths are not equal to $p$.
Suppose the $u$ paths are labeled as $q_1, q_2, \ldots, q_{u}$. Suppose that these paths have $l_{q_1}, l_{q_2}, \cdots, l_{q_{u}}$ marked edges respectively. We require that $l_{q_i} < t$ i.e. that these chosen paths are not saturated.
Let $v_{q_1}, v_{q_2}, \cdots, v_{q_{u}}$ be such that $0 < v_{q_i} \leq t - l_{q_i}$.
We require that 
\begin{align*}
	\sum_{i=1}^{u} v_{q_i} = l.
\end{align*}
This is feasible as long as $r \geq r_{\star}(m)$. The above condition says that the chosen paths have enough spaces to move $l$ edges from path $p$ to the chosen paths. 
Choose a new arrangement $x$ of the marked edges for the sequence of paths $(p, q_1, q_2, \ldots, q_{u})$ from the set
\begin{align*}
	S\left(0\right) \times \left(\prod_{i=1}^{u} S(l_{q_i} + v_{q_i})  \right).
\end{align*} 
Then keeping the arrangements of marked edges of the rest of the paths the same as in $A'$ and choosing the arrangements for the chosen paths as $x$, we have a preimage $A \in T_{r}^{\sss II}$ under the construction described above.

From the two constructions above we have
\begin{align*}
	\vert T_{r+1} \vert \leq  O(k) \vert T_r \vert + O(k^{t+1}) \vert T_r \vert. \numberthis \label{eqn:bound for segments node2vec}
\end{align*}

\end{proof}
We can now compute asymptotics for $E_{m,r}$. 

\begin{lem}\label{lem:sum over r>rmin node2vec}
$\sum_{r= \rmin}^{\rmax} E_{m,r} =  E_{m, \rmin} (1 + o(1))$.
\end{lem}		
\begin{proof}
We start by giving a bound for $E_{m,r}$.
The probability of the $m$ marked edges is bounded by $\rho_n^{m}$. 
The upper bound for the unmarked edges is $kt-m$.
We now compute a bound for the choices of marked vertices and factors arising from backtracks. For this we note that for both the constructions in the proof of Lemma \ref{lem:num arrangement node2vec} we modify at most $t+1$ paths and create an additional segment in the modified paths.
By similar reasoning as in the proof of Lemma \ref{lem:sum over m node2vec}, for any Type I segment with $m'$ marked edges in the modified paths we have an upper bound of $n^{m'-1}$ and for a Type II segment (in the modified paths) we have an upper bound of $O(1)$. 
Since we create an additional segment in the modified paths, we have an associated factor of $O(n^{-1})$.
Combining these we have
\begin{align*}
E_{m,r} \leq N_{m,r} n^{m-r} \rho_n^{m} m^{kt - m}. 
\end{align*}
This implies that
\begin{align*}
\sum_{r= \rmin}^{\rmax} E_{m,r}  
= E_{m,\rmin} \sum_{r= \rmin}^{\rmax} \left(\frac{O(k^{t+1})}{n}\right)^{r - \rmin}
= E_{m,\rmin} (1 + o(1)).
\end{align*}

\end{proof}

By Lemmas \ref{lem:sum over m node2vec} and \ref{lem:sum over r>rmin node2vec} we have
\begin{align*}
\sum_{m < kt} E_m \leq C \sum_{m < kt} E_{m, \rmin} = o(U_b^k).
\end{align*}
This completes the proof of Proposition \ref{prop:node2vec kth moment-2}.

\subsection{Concentration of path counts}\label{sec:concent-node2vec}

We will prove the following concentration result for $Y_{b,\alpha}$.  
\begin{prop}\label{prop:Lower tail node2vec}
Let $\tL = \tU = t \geq 3$ be given and $k = \ceil{\log n}$. Suppose that \eqref{eq:t-condn-2} holds.
Then we have 
\begin{align*}
	\Prob\left(|Y_{b,\alpha} - \E Y_{b,\alpha}| > \delta \E Y_{b,\alpha}\right)  = O(n^{-3}),
\end{align*}
where $\delta = \Theta((\log n)^{-\eta})$ for some $\eta > 0$. 
\end{prop}

\begin{proof}
Recall the definition of $X_{p, \alpha}$ from \eqref{eqn:node2vec X-p-def}.
By Markov's inequality, and using Proposition~\ref{prop:node2vec kth moment}, 
\begin{align} \label{eqn:lower tail bound 1 node2vec}
	\Prob\left(|Y_{b,\alpha} - \E Y_{b,\alpha}| > \delta \E Y_{b,\alpha}\right) 
	&\leq \frac{\E (Y_{b,\alpha} - \E 	Y_{b,\alpha})^{2k}}{\delta^{2k} \left(	\E Y_{b,\alpha}\right)^{2k}} \leq  \frac{\E (Y_{b,\alpha} - \E 	Y_{b,\alpha})^{2k}}{\delta^{2k} (L_b)^{2k}},
\end{align}
and moreover, 
\begin{align*}
	\E (Y_{b,\alpha} - \E 	Y_{b,\alpha})^{2k} 
	&= \E \bigg(\sum_{p \in \cP_b} (X_{p,\alpha} - \E X_{p,\alpha})	\bigg)^{2k}
	= \sum_{(p_1, p_2, \ldots, p_{2k}): p_l \in \cP_b} \E \prod_{l=1}^{2k} (X_{p_l,\alpha} - \E X_{p_l,\alpha}),\\
	&=\sum_{(p_1, p_2, \ldots, p_{2k}): p_l \in \cP_b}
	\prod_{l=1}^{2k} \alpha^{N(p_l)} \cdot \E \prod_{l=1}^{2k} (X_{p_l,1} - \E X_{p_l,1}),\\
	&\leq \sum_{(p_1, p_2, \ldots, p_{2k}): p_l \in \cP_b}
	\E \prod_{l=1}^{2k} (X_{p_l,1} - \E X_{p_l,1})
	\numberthis\label{eqn:lower tail bound 2 node2vec}
\end{align*}
where $X_{p_l,1}$ is obtained by plugging $\alpha =1$.
We observe that for any path $p$, $X_{p,1} = X_{p}$ where $X_p$ is as defined in \eqref{eqn:X-p-def}. 
This observation combined with \eqref{eqn:lower tail bound 1 node2vec} and \eqref{eqn:lower tail bound 2 node2vec} implies that we may bound as in the proof of Proposition \ref{prop:Lower tail A^t} to complete this proof.

\end{proof}

\section{Analysis of spectral clustering for node2vec}\label{sec:node2vec}

We analyze the matrix $M_0$ and the eigendecomposition of $M$ in section~\ref{sec:M matrix node2vec}. We then prove Theorem~\ref{prop:node2vec frob norm} in section~\ref{sec:Frob norm node2vec}. We then end with the proof of Theorem~\ref{thm:node2vec misclassified nodes}.

\subsection{Analysis of \texorpdfstring{$M$}{M}-matrix}\label{sec:M matrix node2vec}

We start with the proof of Lemma~\ref{lem:order of M(alpha)} which shows that $M_0$ has an approximate block structure. This then leads to the proof of Proposition~\ref{prop:eigenvector M_0}. 
Using these two results, in Lemma~\ref{lem:Rows of V_0} we then provide bounds for the inner products of rows of $V_0$ similar to the bounds given for DeepWalk in Proposition~\ref{prop:null-matrix}\ref{prop:null-matrix-3}.

	Now we show that $M_0$ has a block structure when $t_L \geq 3$ but not when $t_L=2$.
	We also show that the entries of $M_0$ are $O(1)$ when $t_L \geq 3$.
				
	\begin{lem}\label{lem:order of M(alpha)}
			Suppose that $t_L \geq 2$. 
			Then we have
			\[	M_0 = M_0(\alpha_n) = \log \left(\sum_{t=t_L}^{t_U} \Theta  G_{t,n} \Theta^T + R_{t,n} \right),\] 
			where $G_{t,n}$ is a $K \times K$ matrix, $R_{t,n}$ is a diagonal matrix, $R_{t,n} = 0$ when $t$ is odd, and $(R_{t,n})_{ij} = O(n^{-\frac{t}{2}})$ when $t > 2$ and $t$ is even. 
			
			Further $(G_{t,n})_{ij} = O(1)$ and $M_{ij} = O(1)$ (as a function of $n$) if $t_L > 2$ or if $t_L=2$ and $i \neq j$. If $t_L=2$ and $i =j$ and  $\alpha \to 0$, then $(G_{t,n})_{ii} = 0$ and $M_{ii} \to -\infty$.
		\end{lem}

		\begin{proof}[Proof of Lemma \ref{lem:order of M(alpha)}]

Let
\begin{align}
	P^{\sss (t)}_{ij} =  \sum_{(i_0, i_1, \ldots, i_t) : i_0 = i, i_t = j} P_{i_0 i_1} P_{i_1 i_2} \cdots P_{i_{t-1} i_t} \frac{1}{\vert P \vert}  \left(\prod_{l=1}^{t-1} \frac{1}{|P_{i_l\star}| -1 + \alpha}\right) \cdot \alpha^{\mathfrak{n}((i_0, i_1, \ldots, i_t))}, \label{eqn:node2vec trans prob} 
\end{align}
be the $t$-step transition probability for node2vec. Then
\begin{align*}
(M_0)_{ij} = \log \bigg(	\frac{\sum_{t = t_L}^{t_U} (l-t) \cdot (P^{\sss (t)}_{ij} + P^{\sss (t)}_{ji}) }{2b \gamma(l, t_L, t_U) \frac{|P_{i\star}|}{|P|}\times  \frac{|P_{j\star}|}{|P|} }	\bigg). \numberthis\label{eqn:M_0 node2vec}
\end{align*}

And we define $M_0'$ as follows.
			\begin{align*}
				(M_0')_{ij} : = \bigg(	\frac{\sum_{t = t_L}^{t_U} (l-t) \cdot (P^{\sss (t)}_{ij} + P^{\sss (t)}_{ji}) }{2b \gamma(l, t_L, t_U) \frac{|P_{i\star}|}{|P|}\times  \frac{|P_{j\star}|}{|P|} }	\bigg). \numberthis \label{eqn: def M'}
			\end{align*}
			We first describe a decomposition of $M_0'$.
			Towards this, for any sequence $k = (k_0, k_1, \ldots, k_{m+1})$ such that $0 = k_0 < k_1 < \cdots < k_m < k_{m+1} = t$, and any path type $b = (b_0, b_1, \ldots, b_t)$ such that $b_0 = g(i)$ and $b_t = g(j)$ we define
			\begin{align*}
				E_{b,k,i,j,t} &:= \sum_{(i_0, i_1, \ldots, i_t) | i_0 = i, i_t = j} \ind\{(g(i_0), g(i_1), \ldots, g(i_t)) = (b_0, b_1, \ldots, b_t)	\}\\ 
				&\quad \ind \{	(i_{l}, i_{l+1}) \text{ is a backtrack for } k_r < l < k_{r+1}, (i_{k_r}, i_{k_{r+1}}) \text{ is not a backtrack for } 1\leq r \leq m	\}\\
				&\quad P_{i_0 i_1} P_{i_1 i_2} \cdots P_{i_{t-1} i_t} \frac{1}{\vert P \vert}  \left(\prod_{l=1}^{t-1} \frac{1}{|P_{i_l\star}| -1 + \alpha}\right) \cdot \alpha^{\mathfrak{n}((i_0, i_1, \ldots, i_t))}.\numberthis\label{eqn:P(i,j,t) dec path types}
			\end{align*}
			$E_{b,k,i,j,t}$ is a sum over paths of length $t$ between nodes $i$ and $j$ with locations of edges which are not backtracks given by the sequence $k$ and the block labels of the vertices along the paths given by the sequence $b$.
			We can note that 
			\begin{align*}
				P^{\sss (t)}_{ij} = \sum_k \sum_b E_{b,k,i,j,t}. \numberthis\label{eqn:P(i,j,t) as sum}
			\end{align*}
			For a given $t$, consider a sequence $k = (k_0, k_1, \ldots, k_{m+1})$ such that $k_{r+1} - k_r$ is odd for some $r$ where $1 \leq r \leq m$. This implies that the endpoints of the path do not have an equality constraint between them.  
			From this we can see that the summands $E_{b,k,i,j}$ in equation \ref{eqn:P(i,j,t) dec path types} depend on $i$ and $j$ only through the block types $g(i)$ and $g(j)$.
			On the other hand, consider a sequence $k$ such that $k_{r+1} - k_r$ is even for $1 \leq r \leq m$. Then $t$ must be even and $m \leq t/2$. 
			By Lemma \ref{lem:choice-vertices} there are $O(n^{\frac{t}{2}})$ paths associated with the sequence $k$. 
			Further, we have the equality constraint $i=j$ for this case.
			Then we define
			\begin{align*}
				(N_{t,n})'_{ij} &:= \sum E_{b,k,i,j,t} \ind\{k = (k_0,k_1, \ldots, k_{m+1}), k_{r+1} - k_r \text{ is odd for some } 1 \leq r \leq m, k_{m+1} = t	\},\\
				(R_{t,n})'_{ij} &:=  \sum E_{b,k,i,j,t} \ind\{k = (k_0,k_1, \ldots, k_{m+1}), k_{r+1} - k_r \text{ is even for } 1 \leq r \leq m, k_{m+1} = t	\},\\ 
				(N_{t,n})_{ij} &:= \left(	\frac{\sum_{t=t_L}^{t_U}(l-t) \cdot \left((N_{t,n})'_{ij} + (N_{t,n})'_{ji}\right)}{2 b \gamma(l,t_L, t_U) \frac{|P_{i\star}|}{|P|} \times \frac{|P_{j\star}|}{|P|}}	\right),\\
				(R_{t,n})_{ij} &:= \left(	\frac{\sum_{t=t_L}^{t_U}(l-t) \cdot \left((R_{t,n})'_{ij} + (R_{t,n})'_{ji}\right)}{2 b \gamma(l,t_L, t_U) \frac{|P_{i\star}|}{|P|} \times \frac{|P_{j\star}|}{|P|}}	\right). \numberthis \label{eqn:N R def}
			\end{align*}
			Then $M_0' = N_{t,n} + R_{t,n}$ by equations \ref{eqn: def M'} and \ref{eqn:P(i,j,t) as sum}. 
			Further, by the discussion in the previous paragraph $N_{t,n}$ is a block matrix with the same block structure as $P$, $R_{t,n} = 0$ when $t$ is odd and when $t$ is even and $i \neq j$, and $(R_{t,n})_{ij} = O(n^{-\frac{t}{2}}) ( N_{t,n})_{ij}$ when $t$ is even, $t > 2$ and $i=j$.
			In the case when $t=2$ and $i=j$, we have $( N_{t,n})_{ij} = 0$.
			This completes the first part of the proof.
			
			Now we describe the order of the coefficients of $M_0$ and $N$. We note that if $\alpha =1$, we can see that $M_0(1)$ is the matrix in the DeepWalk in which case by Proposition~\ref{prop:null-matrix} we have $(M_0(1))_{ij} = O(1)$.
			Thus more generally $(M_0)_{ij} = O(1)$ iff for all $t_L \leq t \leq t_U$
			\begin{align*}
				\frac{\sum_{(i_0, i_1, \ldots, i_t) : i_0 = i, i_t = j} P_{i_0 i_1} P_{i_1 i_2} \cdots P_{i_{t-1} i_t} \frac{1}{\vert P \vert}  \left(\prod_{l=1}^{t-1} \frac{1}{|P_{i_l\star}| -1 + \alpha}\right) \cdot \alpha^{\mathfrak{n}((i_0, i_1, \ldots, i_t))}}
				{\sum_{(i_0, i_1, \ldots, i_t) : i_0 = i, i_t = j} P_{i_0 i_1} P_{i_1 i_2} \cdots P_{i_{t-1} i_t} \frac{1}{\vert P \vert}  \left(\prod_{l=1}^{t-1} \frac{1}{|P_{i_l\star}|}\right)} = \theta(1). \numberthis \label{eqn:M alpha ratio}
			\end{align*} 
			Towards this note that $\frac{p_{i_l}}{p_{i_l} -1 + \alpha} \to 1$ as $p_{i_l} = \theta(n\rho_n)$.
			Next consider paths between $i$ and $j$ without any backtracks i.e. $\mathfrak{n}((i_0, i_1, \ldots, i_t)) = 0$. There are $\theta(n^{t-1})$ such paths if $t >2$ and if $t=2$ and $i \neq j$. 
			Now consider paths with $\mathfrak{n}((i_0, i_1, \ldots, i_t)) > 0$ and consider the case $t>2$.
			Then there are $O(n^{t-2})$ such paths and so the contribution from such paths is of a smaller order.
			Thus the fraction in  \eqref{eqn:M alpha ratio} tends to $1$.
			We can note that the leading term, i.e. paths with $\mathfrak{n}((i_0, i_1, \ldots, i_t)) = 0$, is a summand in the definition of $(N')_{t,n}$ in \eqref{eqn:N R def}. Thus $(N'_{t,n})_{ij} = \theta(1)$.
			
			Now consider the case when $t=2$.
			If $t=2$ and $i \neq j$, there cannot be a backtrack and so the leading contribution is from the case when $\mathfrak{n}((i_0, i_1, \ldots, i_t)) = 0$.
			Thus, the fraction in \eqref{eqn:M alpha ratio} tends to $1$. In this case $(R_{2,n})'_{ij} = 0$ and $(N_{2,n})'_{ij} = O(1)$.
			If $t=2$ and $i=j$ then the fraction in \eqref{eqn:M alpha ratio} tends to $0$ when $\alpha \to 0$.
			In this case $(N_{2,n})'_{ii} = 0$ and $(R_{2,n})'_{ii} \to 0$.
			This completes the proof.
			
		\end{proof}
		
		Now we are ready to prove Proposition \ref{prop:eigenvector M_0}.
		
		\begin{proof}[Proof of Proposition \ref{prop:eigenvector M_0}]
			Let $\tilde{\log}: [0, \infty] \to \bbR$ be defined by
			\begin{align*}
				x &\mapsto \log x, \quad x > 0,\\
				x &\mapsto 0, \quad x=0.
			\end{align*}
			To simplify notation, we will write $\log$ for $\tilde{\log}$ at all places in the proof.
			By Lemma \ref{lem:order of M(alpha)}, $(\log (Z_0))_{ij}$ is $O(1)$. So by approximating $\log$ around $1$ we have
			\begin{align*}
				M_0 &= \log \left(\Theta Z_0 \Theta^T + R\right),\\
				&= \Theta \log Z_0 \Theta^T + R',
			\end{align*} 
			where $R'$ is a diagonal matrix and $R'_{ii} = O(n^{-2})$ as we consider $t_L > 2$ for node2vec.
			By Proposition~\ref{prop:null-matrix}\ref{prop:null-matrix-2} (applied to $\Theta \log Z_0 \Theta^T$) and since $\log Z_0$ has rank $K$, $\Theta \log Z_0 \Theta^T$ has $K$ non-zero eigenvalues, $\tilde{\lambda}_i$, where $\tilde{\lambda}_i = \theta(n)$. 
			By Weyl's theorem on eigenvalues, $M_0$ has rank $K$ and each of the $K$ non-zero eigenvalues is $\theta(n)$.
			Now let $v_i$, $1 \leq i \leq K$ be the columns (i.e. eigenvectors of $M_0$) in $V_0$. Then 
			\begin{align*}
				\lambda_i v_i  = M_0 v_i =   \Theta \log Z_0 \Theta^T v_i  + R'v_i, \numberthis\label{eqn:expand M_0 v_i}\\
				\implies v_i = \Theta \left( \lambda_i^{-1} \log Z_0 \Theta^T v_i \right) + \lambda^{-1}_i R' v_i.
			\end{align*}
			Let $\Lambda = \diag{\lambda_1, \lambda_2, \ldots, \lambda_K}$.
			Then taking $X_0 = \log Z_0 \Theta^T V_0 \Lambda^{-1}$ and $E_0 = R' V \Lambda^{-1}$ completes the proof.
		\end{proof}	
		
		Next, we compute bounds on the inner products of rows of $V_0$. The bounds are similar to the DeepWalk case except that we have small error terms.
		
		\begin{lem}\label{lem:Rows of V_0}
			Consider the decomposition $M_0 = \log \left(	\Theta Z_0 \Theta^T + R \right)$. 
			Let $V_0 \in \bbR^{n \times K}$ be the matrix of top $K$ left singular vectors of $M_0$ and let $V_0 = \Theta X_0 + R_0$ be the decomposition from Proposition \ref{prop:eigenvector M_0}. Then 
			\begin{align*}
				\langle (V_0)_{i \star} , (V_0)_{j \star}  \rangle = \ind\left\{g_i = g_j\right\}\left( \Theta \left(\frac{1}{\sqrt{n_1}}, \frac{1}{\sqrt{n_2}}, \ldots, \frac{1}{\sqrt{n_K}}\right)^T\right)_{i} + O(n^{-2.5}).
			\end{align*}
			If $i$ and $j$ are two nodes such that $g_i \neq g_j$, then we have
			\begin{align*}
				\frob{ (V_0)_{i \star} - (V_0)_{j \star} } = \sqrt{\frac{1}{n_{g_i}} + \frac{1}{n_{g_j}}} + O(n^{-3}). 
			\end{align*}
		\end{lem}

		\begin{proof}[Proof of Lemma \ref{lem:Rows of V_0}]
			As shown in proof of Proposition \ref{prop:eigenvector M_0} 
			\begin{align*}
				M_0 = \Theta \log Z_0 \Theta^T + R',
			\end{align*}
			where $R'$ is a diagonal matrix and $R'_{ii} = O(n^{-2})$.
			Let $V \in \bbR^{n \times K}$ be matrix of top $K$ left singular vectors of $N = \Theta \log Z_0 \Theta^T$.
			Then by Proposition~\ref{prop:null-matrix}\ref{prop:null-matrix-2}\ref{prop:null-matrix-3} (applied to $N$), $V = \Theta X$ satisfying the following: 
			\begin{enumerate}
				\item $X \in \bbR^{K \times K}$ and has $K$ distinct rows. And so, $V$ has $K$ distinct rows.
				\item Rows of $X$ are orthogonal and the row norms are given by $\Theta \left(\frac{1}{\sqrt{n_1}}, \frac{1}{\sqrt{n_2}}, \ldots, \frac{1}{\sqrt{n_K}}\right)^T$.
				\item Let $\tilde{\lambda}_i$ for $1 \leq i \leq K$ be the top $K$ eigenvalues. Then $\tilde{\lambda}_i = \theta(n)$.
			\end{enumerate} 
			Let $\lambda_i$ be the top $K$ eigenvalues of $V_0$.
			By Weyl's theorem on eigenvalues, $\vert \lambda_i - \tilde{\lambda}_i \vert \leq O(n^{-1.5})$. 
			Let $\hat{E}_i$ (resp. $E_i$) be the eigenspace corresponding to $\tilde{\lambda}_i$ (resp. $\lambda_i$).
			Let $V_{\hat{E}}$ and $(V_0)_E$ be the eigenvectors corresponding to the eigenspaces.
			Then by Davis-Kahan theorem there exists $O$ such that
			\begin{align*}
				\frob{(V_0)_E - V_{\hat{E}}O} \leq \frac{\frob{R'}}{\theta(n)} = O(n^{-2.5}).
			\end{align*}
			We can note that if we replace $V_{\hat{E}}$ by $V_{\hat{E}}O$, then we continue to have the three properties for $V$ listed above. 
			As a consequence, $\frob{V - V_0} = O(n^{-2.5})$ which implies the first part of the result.
			
			For the second part we have for any $i,j$ such that $g(i) \neq g(j)$
			\begin{align*}
				\frob{ (V_0)_{i \star} - (V_0)_{j \star} } &= \sqrt{\frob{(V_0)_{i \star}}^2 + \frob{(V_0)_{j \star}}^2 
					-2 \langle (V_0)_{i \star} , (V_0)_{j \star} \rangle},\\
				&= \sqrt{\frac{1}{n_{g_i}} + \frac{1}{n_{g_j}} + O(n^{-2.5})},\\
				&= \sqrt{\frac{1}{n_{g_i}} + \frac{1}{n_{g_j}}} + O(n^{-3}).
			\end{align*}
		\end{proof}

\subsection{Bound on \texorpdfstring{$\frob{M-M_0}$}{Frobenius norm M-M0}.}\label{sec:Frob norm node2vec}
We provide a proof of Theorem \ref{prop:node2vec frob norm} in this section.
	\begin{proof}[Proof of Theorem \ref{prop:node2vec frob norm}]
			
			To begin the proof we can note that $M_{ij} = 0$ iff $\sum_{t=t_L}^{t_U} A^{(t)}_{ij} = 0$. Using the same arguments as in the proof of Proposition~\ref{prop:Frob norm t > 2} and Theorem~\ref{prop:Frob norm tL to tU} along with the lower tail inequality in Proposition~\ref{prop:Lower tail A^t} we can conclude that 
			\begin{align*}
				\Prob\left({M}_{ij} = 0 \text{ for some } i,j		\right) = o(1).
			\end{align*}
			Thus we can assume that ${M}_{ij} \neq 0$ for the rest of the proof.

			Let $a_n = 4n(\log n)^{-\eta}$ and
			let 
			$b_{n} = \exp\left\{\frac{a_n}{\sqrt{2}n}\right\}$. 
			Recall from \eqref{eqn:node2vec trans prob} that $P^{\sss (t)}_{ii}$ is the $t$-step transition probability for node2vec, and recall from \eqref{eqn:M_0 node2vec} the form of $M_0$ for node2vec.
			Then we have
			\begin{align*}
				&\Prob\left(	\frob{{M} - M_0}  \geq a_n	\right) 
				= \Prob\left( \sum_{1 \leq i,j \leq n}	\log^2 \left(	\frac{M_{ij}}{(M_0)_{ij}}	\right) \geq a_n^2	\right),\\
				&\leq o(1) + 
				\sum_{1 \leq i, j \leq n} \sum_{t=t_L}^{t_U} \Prob\left(	\frac{\Prob\left( \blw^{(1)}_{1} =i , \blw^{(1)}_{1+t} = j \right) \frac{|P_{i\star}|}{|P|} \times \frac{|P_{j\star}|}{|P|} }
				{P^{\sss (t)}_{ij}\Prob\left(\blw^{(1)}_{1} = i\right) \Prob\left(\blw^{(1)}_{1} = j\right)} \geq b_{n}	\right)\\
				&\quad + \sum_{1 \leq i, j \leq n} \sum_{t=t_L}^{t_U} \Prob\left(	\frac{P^{\sss (t)}_{ij}\Prob\left(\blw^{(1)}_{1} = i\right) \Prob\left(\blw^{(1)}_{1} = j\right)}
				{\Prob\left( \blw^{(1)}_{1} =i , \blw^{(1)}_{1+t} = j \right) \frac{|P_{i\star}|}{|P|} \times \frac{|P_{j\star}|}{|P|} } \geq b_{n}	\right). \numberthis\label{eqn:node2vec frob overall bd 1}
			\end{align*}

			Fix $t_L \leq t \leq t_U$ and $i \neq j$. We show how to bound a typical term in the last summand in \eqref{eqn:node2vec frob overall bd 1}. 
			Recall from \eqref{eqn:P_b def} that $\cP_b$ is the set of paths with vertices having community assignment $b$ for $b\in \cB_{i,j}$.
			Similar to the proof of Proposition \ref{prop:Frob norm t > 2}, 
			for $p = (i_0,\dots,i_t)\in \cP_b$, let 
\begin{gather*}\numberthis\label{eqn:def node2vec X_b, E_b}
    \bar{X}_{p, \alpha} = \frac{A_{i_0 i_1}}{|A_{i_0\star}|} \frac{1}{|A_{i_t\star}|} \left(\prod_{l=1}^{t-1} \frac{A_{i_{l}i_{l+1}}}{|A_{i_l\star}| - 1 + \alpha}\right) \cdot \alpha^{\mathfrak{n}(p)},\quad \text{and} \quad  \bar{Y}_{b,\alpha} = \sum_{p\in \cP_b} \bar{X}_{p, \alpha},\\ 
    \bar{X}_{p,\alpha}^* = \frac{P_{i_0 i_1}}{|P_{i_0\star}|} \frac{1}{|P_{i_t\star}|} \left(\prod_{l=1}^{t-1} \frac{P_{i_{l}i_{l+1}}}{|P_{i_l\star}| - 1 + \alpha}\right) \cdot \alpha^{\mathfrak{n}(p)}, \quad \text{and} \quad \bar{Y}_{b,\alpha}^* = \sum_{p\in \cP_b} \bar{X}_{p,\alpha}^*.
\end{gather*}

			Then we can write 
			\begin{align*}
				&\Prob\left(	\frac{P^{\sss (t)}_{ij}\Prob\left(\blw^{(1)}_{1} = i\right) \Prob\left(\blw^{(1)}_{1} = j\right)}
				{\Prob\left( \blw^{(1)}_{1} =i , \blw^{(1)}_{1+t} = j \right) \frac{|P_{i\star}|}{|P|} \times \frac{|P_{j\star}|}{|P|}}
				\geq b_{n}	\right)\\
				&= \Prob\left(	\frac{\vert P \vert}{\vert A \vert}
				\frac{\sum_{b \in \cB_{i,j}} {\bar{Y}_{b,\alpha}^*} }{\sum_{b \in \cB_{i,j}} \bar{Y}_{b,\alpha} } \geq b_{n} \right) 
				\leq \sum_{b \in \cB_{i,j}} \Prob\left( \frac{\vert P \vert}{\vert A \vert} 
				\frac{\bar{Y}_{b,\alpha}^*}{\bar{Y}_{b,\alpha}} \geq b_{n}		\right) + o(n^{-4}).
				\numberthis \label{eqn:node2vec frob overall bd 2}
			\end{align*}
			The last term $o(n^{-4})$ above, obtained by an application of Proposition~\ref{prop:Lower tail node2vec}, is a bound on the probability $\Prob\left(X_b = 0\right)$ so that $\frac{\bar{Y}_{b,\alpha}^*}{\bar{Y}_{b,\alpha}}$ is well-defined for the rest of the arguments below.
			We now show how to bound the typical summand in \eqref{eqn:node2vec frob overall bd 2}. Towards this we use the Chernoff bound as in the proof of Proposition~\ref{prop:Frob norm t > 2} to bound $\frac{\vert P \vert}{\vert A \vert}$.
			We next show how to bound the fractions coming from the degree terms in \eqref{eqn:def node2vec X_b, E_b}.
			Towards this we note that for any $ 0 < \delta < 1$ and for any $l \in [n]$ we have 
			\begin{align*}
				\Prob\left(	|A_{i_l\star}| -1 + \alpha \leq (1-\delta) (|P_{i_l\star}| -1 + \alpha) 	\right) 
				&= \Prob\left(|A_{i_l\star}| \leq \left(	(1-\delta) + \frac{\delta (1-\alpha)}{|P_{i_l\star}|}	\right) |P_{i_l\star}|\right),\\
				&\leq \Prob\left(|A_{i_l\star}| \leq \left(1 - \frac{\delta}{2}	\right) |P_{i_l\star}|\right),
			\end{align*}
			as $|P_{i_l\star}| = \theta(n \rho_n)$. By similar reasoning we also have for  $ 0 < \delta < 1$ and for any $l \in [n]$
			\begin{align*}
				\Prob\left(	|A_{i_l\star}| -1 + \alpha \geq (1+\delta) (|P_{i_l\star}| -1 + \alpha) 	\right) 
				\leq \Prob \left(	|A_{i_l\star}| \geq \left(1 + \frac{\delta}{2}\right) |P_{i_l\star}|	\right).
			\end{align*}
			These inequalities combined with the Chernoff bound help bound $\Prob\left( \frac{\vert P \vert}{\vert A \vert} 
				\frac{\bar{Y}_{b,\alpha}^*}{\bar{Y}_{b,\alpha}} \geq b_{n}		\right)$ as follows. Define
			\begin{gather*}\label{eqn:def E'b X'b}
			    {X}_{p,\alpha}^* = \prod_{l=1}^t {P_{i_{l-1}i_l}} \alpha^\mathfrak{n}, \quad \text{and} \quad {Y}_{b,\alpha}^* = \sum_{p\in \cP_b} {X}_{p,\alpha}^*.
			\end{gather*}
				
			Then we have 
			\begin{align*}
				\Prob\left( \frac{\vert P \vert}{\vert A \vert} 
				\frac{\bar{Y}_{b,\alpha}^*}{\bar{Y}_{b,\alpha}} \geq b_{n}		\right)
				&\leq \Prob\left( \frac{{Y}_{b,\alpha}^*}{{Y}_{b,\alpha}} \geq \exp \left\{
				\frac{a_n}{\sqrt{2}n} -
				\theta \left(\frac{1}{\sqrt{n}}\right) - 
				\theta \left(	\sqrt{\frac{\log n}{n\rho_n}}	\right) 
				\right\}		\right) + O(n^{-4}),\\
				&\leq \Prob\left( \frac{{Y}_{b,\alpha}^*}{{Y}_{b,\alpha}} \geq 1 + (\log n)^{-\eta}
				\right) + O(n^{-4}), \numberthis \label{eqn:node2vec frob overall bd 3}
			\end{align*}
			as $a_n = 4n(\log n)^{-\eta}$.

			Fix $(i,j)$. We estimate the difference 	
			\begin{align*}
				\left| \E {Y}_{b,\alpha} - {Y}_{b,\alpha}^* \right| &= \left|  \sum_{p \in \cP_b} \left( \E A_{i_0 i_1} \cdots A_{i_{t-1} i_t} -  P_{i_0 i_1} \cdots P_{i_{t-1} i_t}\right) \alpha^\mathfrak{n} \right|,\\
				&\leq O\bigg(\frac{1}{n\rho_n}\bigg) \E Y_{b,1} = O\bigg(\frac{1}{n\rho_n}\bigg) \E Y_{b,\alpha},
				\numberthis \label{eqn:E X'_b - E'_b diff}
			\end{align*}
			where the inequality follows from \eqref{eqn:deepwalk bd EX'b E'b 1} in the proof of Proposition~\ref{prop:Frob norm t > 2} for DeepWalk as $\alpha \leq 1$, and the last equality follows as both $\E Y_{b,\alpha}$ and $\E Y_{b,1}$ are $\theta(n^{t-1} \rho_n^t)$ by Proposition~\ref{prop:node2vec kth moment}.
			These computations show that 
        \begin{gather*}
    \frac{{Y}_{b,\alpha}}{{Y}_{b,\alpha}^*}= \frac{{Y}_{b,\alpha}}{\E {Y}_{b,\alpha} \left( 1 + O((n \rho_n)^{-1}) \right)}, 
    \quad \text{and} \quad
    \frac{{Y}_{b,\alpha}^*}{{Y}_{b,\alpha}}= \frac{\E {Y}_{b,\alpha} \left( 1 + O((n \rho_n)^{-1}) \right)}{{Y}_{b,\alpha}}.
    \numberthis \label{eqn: Y_b Y_b* bound node2vec}
    \end{gather*}

			With this bound we complete our calculation from \eqref{eqn:node2vec frob overall bd 3}.
			We use Proposition~\ref{prop:Lower tail node2vec} and \eqref{eqn: Y_b Y_b* bound node2vec} to conclude that
			\begin{align*}
			\Prob\left( \frac{{Y}_{b,\alpha}^*}{{Y}_{b,\alpha}} \geq 1 + (\log n)^{-\eta}
				\right) = O(n^{-3}).
			\end{align*}
			This completes the argument to show that the last term in \eqref{eqn:node2vec frob overall bd 1} goes to $0$. 
			The second term in  \eqref{eqn:node2vec frob overall bd 1} can be bounded similarly. This completes the first part of the proof.

			The second part of the proof, for the range $n^{t_U-1}\rho_n^{t_U} \ll 1$, is similar to the analogous proof in the proof of Theorem~\ref{prop:Frob norm tL to tU} except the following two changes:
			\begin{enumerate}
				\item Lemma \ref{lem:order of M(alpha)} is used to show that the entries of $M_{ij}$ are $O(1)$.
				\item We use Proposition~\ref{prop:node2vec kth moment} to bound $\E Y_{b, \alpha}$.
			\end{enumerate}

		\end{proof}

\subsection{Bounding the number of missclassified nodes}\label{sec:prediction error node2vec}
We provide a proof of Theorem~\ref{thm:node2vec misclassified nodes} in this section.
\begin{proof}[Proof of Theorem~\ref{thm:node2vec misclassified nodes}]
For this proof, we choose $O \in \bbR^{K \times K}$ obtained by an application of Proposition~\ref{lem:eigenvector bound} which satisfies
\begin{align}
	\frob{V - V_0O} \leq \frac{\frob{M - M_0}}{Cn} = \SOP(1), \label{eqn:thm1 Davis Kahan node2vec}
\end{align}
where the last step follows using Theorem~\ref{prop:node2vec frob norm}.
To simplify notation, we denote $U =V_0O$ in the rest of the proof.
Let $n_{\max} = \max_{r \in [K]} n_r$. 
Recall $(\bar{\Theta},\bar{X})$ from \eqref{eqn:kmeans2} and let $\bar{V} = \bar{\Theta}\bar{X}$.
Then let 
\begin{align*}
	S = \bigg\{	i \in [n] : \frob{\bar{V}_{i \star}  - U_{i \star}} \geq \frac{1}{5} \sqrt{\frac{2}{n_{\max}}}	\bigg\}.
\end{align*}
We will show that for $i \notin S$ the community is predicted correctly using $\bar{\Theta}$.
The proof of the theorem is in two steps.
			
\paragraph*{Step 1: Bounding $\vert S \vert$.}
By the definition of $S$ we have
\begin{align*}
	\sum_{i \in S} \sqrt{\frac{2}{n_{\max}}} \leq \sum_{i \in S} 5 \frob{\bar{V}_{i \star}  - U_{i \star}}
	\leq 5 \frob{\bar{V}  - U}
	\implies  \vert S \vert \leq  \frac{5}{\sqrt{2}} \sqrt{n_{\max}}  \frob{\bar{V} - U}.
	\numberthis \label{eqn:|S| bound 1 node2vec}
\end{align*}
Next, we recall the optimization problem in \eqref{eqn:kmeans2} is given as follows.
	\begin{align*}
		\frob{ \bar{\Theta} \bar{X} - V }^2 \leq (1+\varepsilon)  \min_{\Theta \in \{0,1\}^{n\times K}, X \in \bbR^{K \times K}} \frob{\Theta X - V }^2. 
	\end{align*}
	We substitute $U$ for $\Theta X$ to get the following upper bound:
			\begin{align}
				\frob{\bar{V} - V}^2 \leq (1 + \varepsilon) \frob{U - E_0 - V}^2. \label{eqn:Vbar - V bound node2vec}
			\end{align}
			Then by \eqref{eqn:thm1 Davis Kahan node2vec} and \eqref{eqn:Vbar - V bound node2vec} we have 
			\begin{align*}
				\frob{\bar{V} - U} \leq \frob{\bar{V} - V} + \frob{V- U}
				\leq (1 + o(n^{-1.5}) + \sqrt{1+\varepsilon})  \frob{V- U}
				= \SOP(1).
			\end{align*}
			This combined with \eqref{eqn:|S| bound 1 node2vec} we have $\vert S \vert = \SOP(\sqrt{n})$.
			
			\paragraph*{Step 2:  Bounding the prediction error.}
			For any community $r$, there exists $i_r \in [n]$ such that $g_{i_r} =r$ and $i_r \notin S$ as $n_r = \theta(n)$ and $\vert S \vert = \SOP(\sqrt{n})$.
			By Lemma \ref{lem:Rows of V_0}, for $r\neq s$ we have
			\begin{align*}
				\frob{\bar{V}_{i_r \star} - \bar{V}_{i_s \star}} &\geq \frob{ U_{i_r \star} - U_{i_s \star} } - \frob{\bar{V}_{i_r \star} - U_{i_r \star}}	- \frob{\bar{V}_{i_s \star} - U_{i_s \star}},\\
				&\geq \sqrt{\frac{1}{n_{r}} + \frac{1}{n_{s}}} + O(n^{-3}) - \frac{2}{5}  \sqrt{\frac{2}{n_{\max}}},\\
				&\geq \frac{2.9}{5} \sqrt{\frac{2}{n_{\max}}}. \numberthis\label{eqn:step 2 bd 1 node2vec}
			\end{align*}
			Next, let $i$ be such that $g_i = r$ and $i \notin S$. We show that $\bar{V}_{i \star} = \bar{V}_{i_r \star}$. 
			For this we note that $U_{i \star} = U_{i_r \star}$ and
			\begin{align*}
				\frob{\bar{V}_{i \star} - \bar{V}_{i_r \star}} &\leq \frob{\bar{V}_{i \star} - U_{i \star}} +
				\frob{U_{i_r \star} - \bar{V}_{i_r \star}},\\
				&< \frac{1}{5} \sqrt{\frac{2}{n_{\max}}} + \frac{1}{5} \sqrt{\frac{2}{n_{\max}}} + O(n^{-3}),\\
				&< \frac{2.1}{5} \sqrt{\frac{2}{n_{\max}}}.  \numberthis\label{eqn:step 2 bd 2 node2vec}
			\end{align*}
			In view of \eqref{eqn:step 2 bd 1 node2vec} and \ref{eqn:step 2 bd 2 node2vec} we must have $\bar{V}_{i \star} = \bar{V}_{i_r \star}$.
			Let $C$ be a permutation matrix defined so that $\Theta_0 C$ assigns community $r$ to node $i_r$, for $1 \leq r \leq K$.
			Then we have, 
			\begin{align*}
				\sum_i \ind\{	\bar{\Theta}_{i\star} \neq (\Theta_0 C)_{i \star}	\} \leq \vert S \vert.
			\end{align*}
			From the bound on $\vert S \vert$ from Step 1, we have that $\err(\bar{\Theta},\true)$ is $\SOP\left(1/\sqrt{n}\right)$ and this completes the proof of Theorem~\ref{thm:node2vec misclassified nodes}.
\end{proof}

\bibliographystyle{abbrv}

\appendix 

\section{Properties of optimizers for embedding algorithms} \label{appendix-1}
\begin{proof}[Proof of Proposition~\ref{lem:factorization}]
Note that 
\begin{align*}
    \E_{l\sim P_{\sss C}} \big[ \log \sigma(-\langle \blf_i, \blf_{l}'\rangle ) )\big]  = \sum_{j'=1}^n \frac{|C_{\star j'}|}{|C|}\log \sigma(-\langle \blf_i, \blf_{j'}'\rangle ) ). 
\end{align*}
Thus, \eqref{eq:loss2} reduces to 
    \begin{align} \label{eq:loss3}
		L_{\sss C}(F,F') = \sum_{i,j=1}^{n} 
		\bigg[C_{ij}  \log \sigma(\langle \blf_i, \blf_{j}'\rangle) + b \frac{|C_{i \star}| |C_{\star j}|}{|C|}\log \sigma(-\langle \blf_i, \blf_{j}'\rangle) \bigg].
	\end{align}
Let $\ell_{\sss C}(\langle \blf_i,\blf_j'\rangle)$ be the $(i,j)$-th summand above. 
Defining $x = \langle \blf_i,\blf_j'\rangle$, we optimize $\ell_{\sss C}(x)$.  Taking partial derivative with respect to $x$, 
\begin{align*}
    \frac{\partial \ell_{\sss C}}{\partial x} = \frac{C_{ij}}{1+\e^x} -  b \frac{|C_{i \star}| |C_{\star j}|}{|C|} \times \frac{\e^x}{1+\e^x}.
\end{align*}
Equating the derivative to zero,  we have $x= \log \big(\frac{C_{ij} \cdot  |C|}{ |C_{i \star}| |C_{\star j}|}  \big) - \log b$, and thus, 
\begin{align} \label{eq:fixed-point}
     \langle \blf_i,\blf_j'\rangle = \log \bigg(\frac{C_{ij} \cdot  |C|}{ |C_{i \star}| |C_{\star j}|}  \bigg) - \log b = (M_{\sss C})_{ij}. 
\end{align}
Thus, if $\bar{M}_{\sss C} = F F'^T$, then $\langle F_{i \star} , F'_{j \star} \rangle$ satisfies \eqref{eq:fixed-point}. 
\end{proof}
	
	\begin{proof}[Proof of Lemma \ref{lem:M for node2vec}]
			We recall that co-occurences are given by 
			\begin{align*}
				C_{ij} = \sum_{t=t_L}^{t_U} \sum_{m=1}^{r} \sum_{k=1}^{l-t} \ind\left\{	\blw^{(m)}_k = i, \blw^{(m)}_{k+t} = j	\right\} + \ind\left\{	\blw^{(m)}_k =j, \blw^{(m)}_{k+t} = i	\right\}.
			\end{align*}
			By the strong law of large numbers, 
			\begin{align*}
				\frac{\sum_{m=1}^{r} \ind\left\{	\blw^{(m)}_k = i, \blw^{(m)}_{k+t} = j	\right\}}{r} \xrightarrow[r \to \infty]{a.s.} \Prob\left(	\blw^{(1)}_{k} = i, \blw^{(1)}_{k+t} = j	\right). \numberthis \label{eqn:Cij limit 1}
			\end{align*}
			The initial distribution for DeepWalk is an invariant distribution for each of the random walks. We now note that the initial distribution for node2vec is also invariant for the walks for node2vec, i.e.:
			\begin{align*}
				\Prob\left(	\blw^{(1)}_{1} = i, \blw^{(2)}_{1} = j	\right) = \Prob\left(	\blw^{(1)}_{k} = i, \blw^{(1)}_{k+1} = j	\right), \quad 1 \leq k < l. \numberthis \label{eqn:invaraiant dist}
			\end{align*} 
			The computation below shows that this follows by induction. Suppose that $(i_{k+1}, i_{k+2}) \in E$.
			\begin{align*}
				&\Prob \left( \blw^{(1)}_{k+1} = i_{k+1},  \blw^{(1)}_{k+2} = i_{k+2}	\right)\\
				&= \sum_{i_k | (i_k,i_{k+1}) \in E} \Prob\left(	\blw^{(1)}_{k} = i_{k}, \blw^{(1)}_{k+1} = i_{k+1},  \blw^{(1)}_{k+2} = i_{k+2}	\right),\\
				&= \sum_{i_k | (i_k,i_{k+1}) \in E} \Prob\left(	\blw^{(1)}_{k+2} = i_{k+2} | \blw^{(1)}_{k} = i_{k}, \blw^{(1)}_{k+1} = i_{k+1}	\right) \cdot \Prob\left(\blw^{(1)}_{k} = i_{k}, \blw^{(1)}_{k+1} = i_{k+1}	\right),\\
				&= \frac{\alpha A_{i_{k+1} i_{k}} }{d_{i_{k+1}} + (-1 + \alpha)A_{i_{k+1} i_{k}} } \cdot \frac{A_{i_{k} i_{k+1}}}{ \vert A \vert}  +
				\sum_{i_k \neq i_{k+2}} \frac{A_{i_{k+1} i_{k+2}}}{d_{i_{k+1}} + (-1 + \alpha)A_{i_{k+1} i_{k}}}  \cdot \frac{A_{i_{k} i_{k+1}}}{\vert A \vert} \ind\left\{(i_k,i_{k+1}) \in E\right\} ,\\
				&= \frac{1}{\vert A\vert}.
			\end{align*}
			In view of \eqref{eqn:Cij limit 1} and \eqref{eqn:invaraiant dist} we have
			\begin{align*}
				\frac{C_{ij}}{r} &\xrightarrow[r \to \infty]{a.s.} \sum_{t = t_L}^{t_U}	(l-t) \left(  \Prob\left( \blw^{(1)}_{1} =i , \blw^{(1)}_{1+t} = j \right)
				+ \Prob\left( \blw^{(1)}_{1} = j , \blw^{(1)}_{1+t} = i \right)	\right),\\
				\frac{\sum_{i'} C_{i'j}}{r} &\xrightarrow[r \to \infty]{a.s.} 2 \Prob\left(\blw^{(1)}_{1} = j\right) \cdot \sum_{t=t_L}^{t_U} (l-t) = 2 \gamma(l,t_L, t_U) \Prob\left(\blw^{(1)}_{1} = j\right), \\
				\frac{\sum_{j'} C_{ij'}}{r} &\xrightarrow[r \to \infty]{a.s.} 2 \Prob\left(\blw^{(1)}_{1} = i\right) \cdot \sum_{t=t_L}^{t_U} (l-t) = 2 \gamma(l,t_L, t_U) \Prob\left(\blw^{(1)}_{1} = i\right),\\
				\frac{\sum_{i'j'} C_{i'j'}}{r} &\xrightarrow[r \to \infty]{a.s.}  2 \gamma(l,t_L, t_U).
			\end{align*}
			The result follows from these equations. We note that if, $\sum_{t=t_L}^{t_U} A^{(t)}_{ij} > 0$ then,
			$$\sum_{t = t_L}^{t_U}	(l-t) \left(  \Prob\left( \blw^{(1)}_{1} =i , \blw^{(1)}_{1+t} = j \right)
			+ \Prob\left( \blw^{(1)}_{1} = j , \blw^{(1)}_{1+t} = i \right)	\right) > 0,$$ 
			and so the logarithm is well-defined for the limiting term. In this case, we also have $C_{ij} > 0$ for large $r$ and so $\log\left(	\frac{C_{ij} \cdot \left(	\sum_{i'j'} C_{i'j'}	\right)}{b \sum_{j'} C_{ij'} \sum_{i'} C_{i'j}}	\right)$ is well-defined for large enough $r$.
		\end{proof}
\begin{proof}[Proof of Lemma~\ref{lem:rank M is K}]
Let 
\begin{align*}
Z = \{ X  \in \mathbb{R}^{K\times K} : \rank(X) < K \},
\end{align*}
be the set of all matrices which are not of full rank. Note that $\lambda(Z) = 0$ as $Z$ is the zero set of the determinant function on $\mathbb{R}^{K\times K}$ which is a polynomial function of the matrix entries.
Next, let 
\begin{align*}
    S = \{ X \in \mathbb{R}_+^{K\times K} : \rank(X) = K,\, \rank(\bar{\log} X) < K \},
\end{align*}
be the set of all matrices in $\mathbb{R}_+^{K\times K}$ for which the rank drops after applying $\bar{\log}$. 
Let $\bar{\exp}$ be the mapping given by applying the exponential function element-wise to the matrix entries. 
Then we can note that $S \subset \bar{\exp}(Z)$. 
Thus to complete the proof, it is enough to show that $\lambda(\bar{\exp}(Z)) = 0$. 
This final assertion follows by \cite[Lemma 7.25]{rudin} as $\bar{\exp}$ is a smooth function and so it must map null sets to null sets.	
\end{proof}

\end{document}

%% file: macros.tex
\newcommand{\cB}{\mathcal{B}}
\newcommand{\cE}{\mathcal{E}}\newcommand{\cF}{\mathcal{F}}

\newcommand{\cN}{\mathcal{N}}
\newcommand{\cP}{\mathcal{P}}

\newcommand{\E}{\mathbb{E}}
\newcommand{\R}{\mathbb{R}}
\newcommand{\Prob}{\mathbb{P}}
\newcommand{\ind}{\mathds{1}}

\usepackage{environ}
\NewEnviron{eq}{%
\begin{equation}\begin{split}
  \BODY
\end{split}\end{equation}
}
\newcommand{\e}{\mathrm{e}}

\def\sss{\scriptscriptstyle}

\newcommand{\blf}{\bm{f}}

\newcommand{\blw}{\bm{w}}

\newcommand{\DA}{D_{\sss A}}
\newcommand{\DP}{D_{\sss P}}
\newcommand{\WA}{W_{\sss A}}
\newcommand{\WP}{W_{\sss P}}
\newcommand{\AG}{A_{\sss G}}

\newcommand{\tL}{t_{\sss L}}
\newcommand{\tU}{t_{\sss U}}
\newcommand{\OP}{O_{\sss \mathbb{\PR}}}
\newcommand{\oP}{o_{\sss \mathbb{\PR}}}

\newcommand{\err}{\mathrm{Err}}
\newcommand{\true}{\Theta_0}

\newcommand{\rmin}{r_{*}(m) }
\newcommand{\rmax}{r^*(m)}
\newcommand{\sF}{\mathcal{F}}

\usepackage{bbm}
\newcommand{\1}{\mathbbm{1}}

\newcommand\blfootnote[1]{%
  \begingroup
  \renewcommand\thefootnote{}\footnote{#1}%
  \addtocounter{footnote}{-1}%
  \endgroup
}

\DeclareMathOperator*{\argmin}{argmin}

\DeclareMathOperator*{\rank}{rank}